\documentclass{article}

\usepackage{newa4}
\usepackage{latexsym}
\usepackage{amssymb}
\usepackage{amsmath}
\usepackage{amsthm}
\usepackage{amsfonts}
\usepackage{mathabx}

\usepackage[T1]{fontenc}
\usepackage[utf8]{inputenc}
\usepackage{autobreak}
\usepackage{graphicx}

\newcommand{\bA}{\mathbb{A}}
 
\newcommand{\CC}{\mathbb{C}}
 
\newcommand{\II}{\mathbb{I}}
\newcommand{\FF}{\mathbb{F}}

\newcommand{\NN}{\mathbb{N}}
\newcommand{\ZZ}{\mathbb{Z}}
\newcommand{\QQ}{\mathbb{Q}}

\newcommand{\KK}{\mathbb{K}}
\newcommand{\MM}{\mathbb{M}}
\newcommand{\RR}{\mathbb{R}}
\newcommand{\TT}{\mathbb{T}}

\newcommand{\AAA}{\mathcal{A}}

\newcommand{\CCC}{\mathcal{C}}

\newcommand{\KKK}{\mathcal{K}}

\newcommand{\JJJ}{\mathcal{J}}
\newcommand{\PPP}{\mathcal{P}}
\newcommand{\QQQ}{\mathcal{Q}}
\newcommand{\TTT}{\mathcal{T}}

\newcommand{\ZZZ}{\mathcal{Z}}

\newcommand{\SD}{\mathrm{SD}}
\newcommand{\av}[1]{\mathrm{av}_{#1}}
\newcommand{\ball}[3]{\mathrm{B}_{#1}(#2,#3)}

\newcommand{\hdm}{\rho_{\mathrm H}}

\newcommand{\card}[1]{\mathopen\parallel #1 \mathclose\parallel}
\newcommand{\semantik}[1]{[\![#1]\!]}
\newcommand{\val}[1]{\semantik{#1}}
\newcommand{\ar}[1]{\mathrm{ar}(#1)}

\newcommand{\hd}[1]{\mathrm{hd}(#1)}
\newcommand{\rt}[1]{\mathrm{root}(#1)}
\newcommand{\sbt}[1]{\mathrm{subtree}(#1)}
\newcommand{\tl}[1]{\mathrm{tl}(#1)}
\newcommand{\br}{\,\mathbf{r}\,}
\newcommand{\set}[2]{\mbox{$\{\,#1 \mid #2 \,\}$}}
\newcommand{\fun}[3]{\mbox{$#1 \colon #2 \rightarrow #3$}}
\newcommand{\pfun}[3]{\mbox{$#1 \colon #2 \rightharpoonup #3$}}
\newcommand{\mfun}[3]{\mbox{$#1 \colon #2 \rightrightarrows #3$}}
\newcommand{\pair}[1]{\langle #1 \rangle}
\newcommand{\emptywd}{\pair{\,\,}}

\newtheorem{lemma}{Lemma}[section]
\newtheorem{ex}[lemma]{Example}
\newenvironment{example}{\begin{ex}\em}{\end{ex}}
\newtheorem{defin}[lemma]{Definition}
\newenvironment{definition}{\begin{defin}\em}{\end{defin}}
\newtheorem{proposition}[lemma]{Proposition}
\newtheorem{theorem}[lemma]{Theorem}
\newtheorem{corollary}[lemma]{Corollary}

\def\dom{\mathop{\mathstrut\rm dom}}
\def\range{\mathop{\mathstrut\rm range}}

\def\ev{\mathop{\mathstrut\rm ev}}
\def\diag{\mathop{\mathstrut\rm diag}\nolimits}
\def\int{\mathop{\mathstrut\rm int}\nolimits}

\def\UC{\mathop{\mathstrut\cal C}\nolimits}
\def\pr{\mathop{\mathstrut\rm pr}\nolimits}
\def\id{\mathop{\mathstrut\rm id}\nolimits}
\def\equ{\mathbin{\overset{\cup}{=}}}

\def\Pr{\mathop{\mathstrut\rm Pr}\nolimits}
\def\Cons{\mathop{\mathstrut\rm Cons}\nolimits}

\title{Computing with Continuous  Objects:\\ A Uniform Co-inductive Approach\thanks{
\protect\includegraphics[width=1.1em]{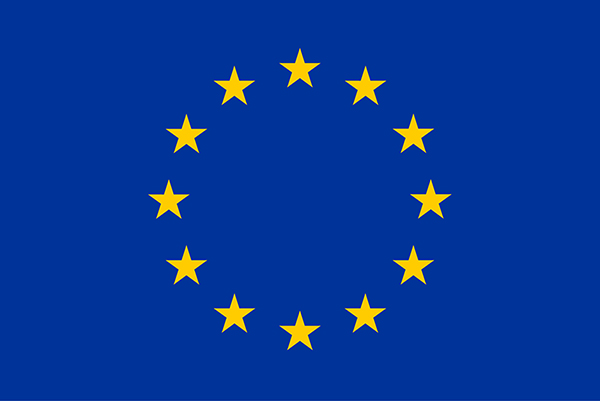}
This project has received funding from the European Union's Horizon 2020 research and innovation programme under the Marie Sk\l{}odowska-Curie grant agreement No 731143.}}

\author{Dieter Spreen\\
Department of Mathematics, University of Siegen\\
57068 Siegen, Germany}

\date{}

\begin{document}

\maketitle

\begin{abstract}
A uniform approach to computing with infinite objects like real numbers, tuples of these, compacts sets, and uniformly continuous maps is presented. In work of Berger it was shown how to extract certified algorithms working with the signed digit representation from constructive proofs. Berger and the present author generalised this approach to complete metric spaces and showed how to deal with compact sets. Here, we unify this work and lay the foundations for doing a similar thing for the much more comprehensive class of compact Hausdorff spaces occurring in applications. The approach is of the same computational power as Weihrauch's Type-Two Theory of Effectivity. 
\end{abstract}

 \tableofcontents

\section{Introduction}\label{sec-intro}

In investigations on exact computations with continuous objects such as the real numbers, objects are usually represented by streams of finite data. This is true for theoretical studies in the Type-Two Theory of Effectivity approach (cf.\ e.g.\ \cite{wei}) as well as for practical research, where prevalently the signed digit representation is used (cf.\ \cite{cg,em,bh}), but also others \cite{es,eh,ts}. Berger \cite{be} showed how to use the method of program extraction from proofs to extract certified algorithms working with the signed digit representation in a semi-constructive logic allowing inductive and co-inductive definitions.

In order to generalise from the different finite objects used in the various stream representations, Berger and the present author \cite{bs} used the abstract framework of what was coined \emph{digit space}, i.e.\ a bounded complete non-empty metric space $X$ enriched with a finite set $D$ of contractions on $X$, called \emph{digits}, that cover the space, that is
\[
X = \bigcup\set{d[X]}{d \in  D},
\]
where $d[X] = \set{d(x)}{x \in X}$. Spaces of this kind were studied by Hutchinson in his basic theoretical work on self-similar sets \cite{hu} and used later also by Scriven \cite{sc} in the context of exact real number computation. 

Digit spaces are compact and weakly hyperbolic, where the latter property means that for every infinite sequence $d_0, d_1, \ldots$ of digits the intersection $\bigcap_{n \in \NN} d_0 \circ \cdots \circ d_n[X]$ contains \emph{at most} one point \cite{ed}. Compactness, on the other hand, implies that each such intersection contains \emph{at least} a point. By this way every stream of digits denotes a uniquely determined point in $X$. Because of the covering property it follows conversely that each point in $X$ has such a code. 

With respect to the operations of adding a digit $d$ to the left side of a stream and applying a map $d$ in $D$ to an element of $X$, respectively, the space of digit streams as well as $X$ are algebras of the same signature and the coding map is a morphism respecting these operations.

A central aim of the joint research was to lay the foundation for computing with non-empty compact sets and for extracting algorithms for such computations from mathematical proofs. It is a familiar fact that the set of all non-empty compact subsets of a bounded complete metric space is a bounded and complete space again with respect to the Hausdorff metric \cite{en,mu}. However, as was shown in \cite{bs}, in general there is no finite set of contractions that covers the hyperspace. On the other hand, non-empty compact subsets can be represented in a natural way by finitely branching infinite trees of digits. Moreover, all characterisations in \cite{bs} derived for the stream representation of the elements of the digit space hold true for the tree representation of the non-empty compact subsets of the space.

The goal of the investigation presented in this paper is to show that a uniform approach to computing with continuous objects comprising the non-empty compact sets case can be obtained by allowing the contractions of a digit space to be multi-ary. Points are then no longer represented by digit streams but by finitely branching infinite trees, called \emph{$D$-trees}.
 As we will see, not only a uniform version of the results in \cite{bs} can be derived, but also an analogue of Berger's inductive co-inductive characterisation of the (constructively) uniformly continuous endofunctions on the unit interval \cite{be}, which allows representing also such functions as finitely branching infinite trees.

There is also a second objective which results from the observation that the essential properties needed in the approach pursued in \cite{bs} are covering, compactness and weak hyperbolicity. So, it seems that more generality is gained by starting from spaces with these properties. 

Besides the general framework and the hyperspace of non-empty compact subsets, the construction of product spaces is presented. In both cases it is investigated whether important properties are inherited under the constructions. Moreover, to demonstrate the power of the framework, several results from topology are derived that are relevant for applications.

It is well known that the product and the hyperspace construction are both functorial. Here, we give proofs of the functoriality on the basis of the co-inductive characterisations of the spaces involved, that means, we use co-induction and/or a combination of induction and co-induction. 

In his seminal 1951 paper on spaces of subsets~\cite{mi}, Michael showed that compact unions of compact sets are compact again. We give a non-topological proof of this result, based on co-induction. Other results we derive in a similar way include the fact that singleton sets are compact, as are direct images of compact sets under uniformly continuous functions.  From the proofs algorithms transforming tree representations of points $x$ into tree representations of the compact sets $\{ x \}$, and tree representation of uniformly continuous functions $f$ as well as tree representations of compact sets $K$ into tree representations of the compact sets $f[K]$, respectively, can be extracted.

The paper is organised as follows: Section~\ref{ind-def} contains a short introduction to inductive and co-inductive definitions and the proof methods they come equipped with. As a first application, finite and finitely-branching infinite trees are defined in Section~\ref{tree}. Digit spaces are iterated function systems. In Section~\ref{sec-ifs} function systems with multi-ary functions are considered and essential results derived. 

As said, the central aim of this paper is to present a uniform approach to computing with infinite objects like real numbers, and tuples or compact sets of such. In Section~\ref{sec-comput} we restrict our study to the case of extended iterated function systems where the underlying space is a compact metric space and the maps making up the function system are contractions. There is a  vast number of effectivity studies for metric spaces. Here, the aim is to show that what was obtained in \cite{bs} remains true in the more general case of multi-ary digit maps. A particular consequence is the equivalence of the present approach with Weihrauch's Type-Two Theory of Effectivity~\cite{wei}. 

In Weihrauch's approach one has to deal with representations explicitly. Often this requires involved codings which makes it hard for the usual mathematician to follow the proof argument. Proof extraction is an approach that avoids dealing with representations: the representations of objects as well as the algorithms computing with them are automatically extracted from formal proofs. Section~\ref{sec-progex} contains a short introduction. As is shown, the tree representation considered in Section~\ref{sec-ifs} results from the co-inductive characterisation of the space.

In Section~\ref{sec-cauchy} the equivalence between the property that every point of a digit space is the limit of a regular Cauchy sequence of elements of a dense base, and the co-inductive characterisation of the space is derived in a constructive fashion. Via proof extraction computable translations between the Cauchy representation used in Type-Two Theory of Effectivity and the tree representation can be obtained. 

In the following two sections  the construction of new spaces from given ones is considered. 
In Section~\ref{sec-prod} products are examined. All properties of extended iterated function systems and digit spaces, respectively, considered in this paper are inherited from the factor spaces to their product.

The hyperspace of non-empty compact subsets is studied in Section~\ref{sec-compact}. All but one of the properties investigated are inherited to the hyperspace. Only for weak hyperbolicity this is still open in the general case. If the underlying space is a metric one and all digits are contracting, also this property holds.

As a consequence, in both cases the equivalence result derived in Section~\ref{sec-cauchy} carries over to the derived spaces. In \cite{bs} separate proofs had to be given for digit spaces and their hyperspaces. Particularly in the latter case, the proof was quite involved.

Section~\ref{sec-ufun} contains a generalisation of Berger's inductive co-inductive characterisation of the uniformly continuous functions on the unit interval to the digit space case. On the basis of the characterisation it is shown that the function class is closed under composition

The last two sections address applications of the framework to topology. In Section~\ref{sec-prop} the functoriality of the hyperspace construction as well as properties of continuous functions that map into the hyperspace are derived. Section~\ref{sec-mich}, finally, contains a co-inductive treatment of Michael's result. 

The paper finishes with a Conclusion.

\section{Inductive and co-inductive definitions}\label{ind-def}

Let $X$ be a set and $\PPP(X)$ its powerset. An operator $\fun{\Phi}{\PPP(X)}{\PPP(X)}$ is \emph{monotone} if for all $Y, Z \subseteq X$,
\[
\text{
if $Y \subseteq Z$, then $\Phi(Y) \subseteq \Phi(Z)$;
}
\]
and a set $Y \subseteq X$ is \emph{$\Phi$-closed} (or a pre-fixed point of $\Phi$) if $\Phi(Y) \subseteq Y$. Since $\PPP(X)$ is a complete lattice, $\Phi$ has a least fixed point $\mu\Phi$ by the Knaster-Tarski Theorem. If $P \subseteq X$, we mostly write
\[
P(x) \overset{\mu}{=} \Phi(P)(x),
\]
instead of $P = \mu\Phi$. $\mu\Phi$ can be defined to be the least $\Phi$-closed subset of $X$. Thus, we have the \emph{induction principle} stating that for every $Y \subseteq X$.
\[\text{
If $\Phi(Y) \subseteq Y$ then $\mu\Phi \subseteq Y$
}\] 

For monotone operators $\fun{\Phi, \Psi}{\PPP(X)}{\PPP(X)}$ define
\[
\Phi \subseteq \Psi :\Leftrightarrow (\forall Y \subseteq X)\, \Phi(Y) \subseteq \Psi(Y).
\]
It is easy to see that the operation $\mu$ is monotone, i.e., if $\Phi \subseteq \Psi$, then $\mu\Phi \subseteq \mu\Psi$. This allows us to derive the following strengthening of the induction principle.

\begin{lemma}[\textbf{Strong Induction Principle~\cite{be}}]\label{lem-strong}
Let $\fun{\Phi}{\PPP(X)}{\PPP(X)}$ be a monotone operator. Then, 
\[\text{
If $\Phi(Y \cap \mu\Phi) \subseteq Y$, then $\mu\Phi \subseteq Y$.
}\]
\end{lemma}
\begin{proof}
Let $\Psi(Y) := \Phi(Y \cap \mu\Phi)$ and assume that $\Psi(Y) \subseteq Y$. Then $\mu\Psi \subseteq Y$, by the induction principle. Since $Y \cap \mu\Phi \subseteq Y$, we moreover have that $\Phi(Y) \supseteq \Phi(Y \cap \mu\Phi) = \Psi(Y)$. Hence, $\mu\Psi \subseteq \mu\Phi$. It follows that 
\[
\mu\Psi = \Psi(\mu\Psi) = \Phi(\mu\Psi \cap \mu\Phi) = \Phi(\mu\Psi).
\]
Therefore, $\mu\Phi \subseteq \mu\Psi$, again by induction, and whence $\mu\Phi \subseteq Y$.
\end{proof}

Dual to inductive definitions are \emph{co-inductive definitions}. A subseteq $Y$ of $X$ is called \emph{$\Phi$-co-closed} (or a post-fixed point of $\Phi$) if $Y \subseteq \Phi(Y)$. By duality, $\Phi$ has a largest fixed point $\nu\Phi$ which can be defined as the largest $\Phi$-co-closed subset of $\Phi$. So, we have the \emph{co-induction} and the \emph{strong co-induction principle}, respectively, stating that for all $Y \subseteq X$.
\[\text{If $Y \subseteq \Phi(Y)$, then $Y \subseteq \nu\Phi$.}\]
and
\[\text{ If $Y \subseteq \Phi(Y \cup \nu\Phi)$, then $Y \subseteq \nu\Phi$.}\]

Note that for $P \subseteq X$ we also write
\[
P(x) \overset{\nu}{=} \Phi(P)(x)
\]
instead of $P = \nu\Phi$. 

\begin{lemma}[\textbf{Half-strong Co-induction Principle~\cite{beg}}]\label{lem-halfstrong}
Let $\fun{\Phi}{\PPP(X)}{\PPP(X)}$ be a monotone operator. Then, 
\[\text{
if $Y \subseteq \Phi(Y) \cup \nu \Phi$ then $Y \subseteq \nu \Phi$.
}\]
\end{lemma}
\begin{proof}
Let $\Psi(Y) := \Phi(Y) \cup\nu \Phi$. Then $\Psi$ is monotone and pointwise larger than $\Phi$. Hence,
\begin{equation}\label{eq-halfstrong-1}
\nu \Phi \subseteq \nu \Psi.
\end{equation}
On the other hand
\[
\nu \Psi= \Psi(\nu \Psi) = \Phi(\nu \Psi) \cup  \nu \Phi = \Phi(\nu \Psi)
\]
since $\nu \Phi = \Phi(\nu \Phi) \subseteq \Phi(\nu \Psi)$, by (\ref{eq-halfstrong-1}). Hence,
\begin{equation}\label{eq-halfstrong-2}
\nu \Psi \subseteq  \nu \Phi,
\end{equation}
by co-induction.

The premise of half-strong co-induction means $Y \subseteq \Psi(Y)$. Therefore,
$Y \subseteq \nu \Psi$, by co-induction, from which we obtain with (\ref{eq-halfstrong-2}) that $Y \subseteq \nu \Phi$.
\end{proof}

The following examples are taken from~\cite{be}.

\begin{example}({\bf natural numbers})
Define $\fun{\Phi}{\PPP(\RR)}{\PPP(\RR)}$ by
\[
\Phi(Y): = \{ 0 \} \cup \set{y + 1}{y \in Y}.
\]
Then $\mu\Phi = \NN = \{\, 0, 1, \ldots \,\}$. The induction principle is logically equivalent to the usual zero-successor-induction on $\NN$; if $0 \in Y$ and $(\forall y \in Y) (y \in Y \to y+1 \in Y)$, then $(\forall y \in \NN)\, y \in Y$.
\end{example}

\begin{example}({\bf signed digits and the interval $[-1, 1]$})
Set $\II := [-1, 1]$ and for every signed digit $d \in \SD := \{ -1, 0 , 1 \}$ define $\fun{\av{d}}{\II}{\II}$ by
\[
\av{d}(x) := (x+d) / 2.
\]
Let $\II_d := \range(\av{d})$. Then $\II_d = [d/2 - 1/2, d/2 + 1/2]$ and $\II$ is the union of the $\II_d$.

Define $\fun{\Phi}{\PPP(\II)}{\PPP(\II)}$ by
\[
\Phi(Z) := \set{x \in \II}{(\exists d \in \SD) (\exists y \in Z)\, x = \av{d}(y)}
\] 
and let $\CC_\II := \nu \Phi$. Then $\CC_\II \subseteq \II$. Since moreover, $\II \subseteq \Phi(\II)$, it follows with co-induction that also $\II \subseteq \CC_\II$. Hence, $\CC_\II = \II$. The point of this definition is that the proof of $\II \subseteq \Phi(\II)$ has an interesting computational content: $x \in \II$ must be given in such a way that it is possible to find $d \in \SD$ so that $x \in \II_d$. This means that $d/2$ is a \emph{first approximation} of $x$. The computational content of the proof of $\II \subseteq \CC_\II$, roughly speaking, iterates the process of finding approximations to $x$ ad infinitum, i.e.\ it computes a \emph{signed digit representation} of $x$, that is, a stream $a_0 : a_1 : \cdots$ of signed digits with 
\[
x = \av{a_0}(\av{a_1}( \cdots )) = \sum_{i \ge 0} a_i \cdot 2^{-(i+1)}.
\]
\end{example}

\section{$D$-Trees}\label{tree}

Let $\NN_{0} := \NN \setminus \{ 0 \}$ and $\NN_{0}^{n}$ be the set of words of length $n$ over $\NN_{0}$. Define $\NN_{0}^{*} := \bigcup_{n \in \NN} \NN_{0}^{n}$. The empty word will be denoted by $\emptywd$ and the concatenation operation by $\star$. We identify single-letter words with the corresponding letter. Moreover, for $S \subseteq \NN_{0}^{*}$ and $i \in  \NN_{0}$ we set $i \star S := \set{i \star s}{s \in S}$.

A \emph{tree} is a subset of $\NN_{0}^{*}$ that is closed under initial segments. We will consider trees the nodes of which are labelled with elements of a fixed non-empty set D that comes equipped with an outdegree $\fun{\mathrm{ar}}{D}{\NN_0}$. 

The subsequent definition says when $T = (S, L)$ with $S \subseteq \NN_{0}^{*}$ and $\fun{L}{S}{D}$ is a finite $D$-tree of height $n$.
\begin{definition}\label{dn-treen}
\begin{itemize}
\item $T$ is a \emph{$D$-tree of height $0$}, if $S = \{ \emptywd \}$ and for some $d \in D$, $L(\emptywd) = d$;
 
\item $T$ is a \emph{$D$-tree of height $n+1$}, if $S \subseteq D^{n+1}$ and there are $d \in D$ and $D$-trees $T_1, \ldots, T_{\ar{d}}$ of height $n$ such that
\begin{gather*}
S= \{ \emptywd \} \cup \bigcup_{i = 1}^{\ar{d}} i \star S_{i}, \\
L(\emptywd) = d \quad\text{and} \quad L(i \star s) = L_{i}(s), 
\end{gather*}
for $1 \le i \le \ar{d}$ and $s \in S_{i}$.  We write $T = [d; T_{1}, \ldots, T_{\ar{d}}]$ in this case.
\end{itemize}
\end{definition}
Let $\TTT^{(n)}_D$ be the set of all $D$-trees of height $n$ and $\TTT^\ast_D := \bigcup_{n \in \NN} \TTT^{(n)}_D$.

\begin{definition}\label{dn-pref}
For $D$-trees $T$ and $T'$ of height $n$ and $n+1$, respectively, we say that \emph{$T$ is an immediate prefix of $T'$} and write $T \prec T'$, if 
\begin{itemize}
\item either $n = 0$, $T =  (\{ \emptywd \}, \emptywd \mapsto d)$, and $T' =[d; T_1, \ldots, T_{\ar{d}}]$,

\item or $T = [e; T_1, \ldots, T_{\ar{e}}]$ and $T' = [d; T'_1, \ldots, T'_{\ar{d}}]$ so that $e = d$ and $T_\nu \prec T'_\nu$, for all $1 \le \nu \le \ar{e}$.
\end{itemize}
\end{definition}

Let $S \subseteq \NN_{0}^{*}$ be a tree. A labelling $\fun{L}{S}{D}$ is \emph{compatible} if for all $s \in S$,
\[
\set{ i\in \NN_{0}}{(\exists s \in S)\, s\star i \in S} = \{ 1, \ldots, \ar{L(s)} \}.
\]

Define the set $\TTT^\omega_D$ of \emph{infinite $D$-trees} by
\[
\TTT^\omega_D := \set{(S, L)}{\text{$S$ is a tree with compatible labelling $L$}}.
\]

Since every node $d$ has finite outdegree $\ar{d}$, every infinite $D$-tree is a finitely branching tree with only infinite paths. 

For $T = (S, L) \in \TTT^\omega_D$ set
\begin{align*}
&\rt{T} := L(\emptywd),   \\
&\sbt{T} := ((S_{1}, L_{1}), \ldots, (S_{\ar{L(\emptywd)}}, L_{\ar{L(\emptywd)}})),
\end{align*}
where  $S_{i} := \set{s \in \NN_{0}^{*}}{i \star s \in S}$ and $L_{i}(s) := L(i \star s)$, for $1 \le i \le \ar{L(\emptywd)}$ and $s \in S_{i}$.

\begin{theorem}\label{thm-treecoalg}
Define the functor $\fun{\Lambda_{D}}{\mathbf{Set}}{\mathbf{Set}}$ by 
\[
\Lambda_{D}(X) :=  \bigcup_{d \in D} \{ d \} \times X^{\ar{d}}.
\]
Then $(\TTT^\omega_D, \mathrm{root} \times \mathrm{subtree})$ is a terminal co-algebra of $\Lambda_{D}$.
\end{theorem}
The theorem is a consequence of general results on the existence of terminal co-algebras in \cite{amm,ru}.

\begin{definition}\label{dn-coindp}
A relation $R \subseteq \TTT^\omega_D \times \TTT^\omega_D$ is a \emph{bisimulation} if for $T, T' \in \TTT^\omega_D$ with $\rt{T} = e$, $\rt{T'} = d$, $\sbt{T} = (T_{1}, \ldots, T_{\ar{e}})$, and $\sbt{T'} = (T'_{1}, \ldots, T'_{\ar{d}})$, 
\[
R(T, T') \rightarrow e = d \wedge (\forall 1 \le \nu \le \ar{e})\, R(T_\nu, T'_\nu).
\]
\end{definition}

\begin{lemma}\label{lem-coindp}
For $T, T' \in \TTT^\omega_D$, $T = T'$ if, and only if, there is a bisimulation $R \subseteq \TTT^\omega_D \times \TTT^\omega_D$ with $R(T, T')$.
\end{lemma}
The `only-if' part is obvious as the identity on $\TTT^\omega_D$ is a bisimulation. The `if' part is known as the \emph{co-induction proof principle} and holds as $(\TTT^\omega_D, \mathrm{root} \times \mathrm{subtree})$ is a terminal co-algebra~\cite{amm,ru}.

As we will see, each infinite $D$-tree $T$ with $\rt{T} =  d$ and $\sbt{T} = (T_1, \ldots, T_{\ar{d}})$ is uniquely determined by its finite initial segments $T^{(n)}$ recursively defined by
\begin{align*}
& T^{(0)} := (\{ \emptywd \}, \emptywd \mapsto d), \\
&T^{(n+1)} := [d; T^{(n)}_1, \ldots, T^{(n)}_{\ar{d}}].
\end{align*}

\begin{lemma}\label{lem-initseq}
For every $n \in \NN$,
\begin{enumerate}
\item $T^{(n)}$ is a $D$-tree of height $n$.
\item $T^{(n)} \prec T^{(n+1)}$.
\end{enumerate}
\end{lemma}

Let $\TT_D$ be the set of all infinite sequences $(T_\nu)_{\nu \in \NN}$ with  $T_\nu \in \TTT^{(\nu)}_D$ and $T_\nu \prec T_{\nu+1}$. Then $(T^{(n)})_{n \in \NN} \in \TT_D$. Define $\fun{G}{\TTT^\omega_D}{\TT_D}$ by $G(T) := (T^{(n)})_{n \in \NN}$.

Next, let conversely $(T_\nu)_{\nu \in \NN} \in \TT_D$. Then $T_0 = (\{ \emptywd \}, \emptywd \mapsto d_0)$ and $T_{\nu} \prec T_{\nu+1}$, for $\nu \ge 0$. Since $T_{\nu+1} \in \TTT_{D}^{(\nu+1)}$, there are $d_{\nu+1} \in D$ and $Q_{\nu}^{\pair{1}}, \ldots, Q_{\nu}^{\pair{\ar{d_{\nu+1}}}} \in \TTT_{D}^{(\nu)}$ so that 
\[
T_{\nu+1} = [d_{\nu+1}; Q_{\nu}^{\pair{1}}, \ldots, Q_{\nu}^{\pair{\ar{d_{\nu+1}}}}].
\]
If $\nu > 0$, similarly $T_{\nu} = [d_{\nu}; Q_{\nu-1}^{\pair{1}}, \ldots, Q_{\nu-1}^{\pair{
\ar{d_{\nu}}}}]$ with $d_{\nu} \in D$ and $Q_{\nu-1}^{\pair{1}}, \ldots, Q_{\nu-1}^{\pair{\ar{d_{\nu}}}} \in \TTT_{D}^{(\nu-1)}$. As $T_{\nu} \prec T_{\nu+1}$, it follows that $d_{\nu} = d_{\nu+1}$ and $Q_{\nu-1}^{\pair{\kappa}} \prec Q_{\nu}^{\pair{\kappa}}$, for $1 \le \kappa \le \ar{d_\nu}$. Let $d := d_0$. Then $d_\nu = d$, for all $\nu \ge 0$. In addition, set $Q^{\pair{\kappa}} := (Q^{\pair{\kappa}}_\nu)_{\nu \in \NN}$, for $1 \le \kappa \le \ar{d}$. It follows that $Q^{\pair{\kappa}} \in \TT_D$. Set
\[
\hd{(T_{\nu})_{\nu \in \NN}} := d \quad\text{and}\quad 
\tl{(T_{\nu})_{\nu \in \NN}} := (Q^{\pair{1}}, \ldots, Q^{\pair{\ar{d}}}).
\]
Then $(\TT_{D}, \mathrm{hd} \times \mathrm{tl})$ is a $\Lambda_{D}$-co-algebra.
Now, co-recursively define $\fun{F}{\TT_D}{\TTT^\omega_D}$ by
\[
\rt{F((T_\nu)_{\nu \in \NN})} := \hd{(T_\nu)_{\nu \in \NN}} \quad\text{and}\quad 
\sbt{F((T_\nu)_{\nu \in \NN})} := \tl{(T_\nu)_{\nu \in \NN}}.
\]
Then 
\begin{equation}\label{eq-F(n)}
F((T_\nu)_{\nu \in \NN})^{(\kappa)} = T_\kappa,
\end{equation}
for $\kappa \in \NN$. Moreover, it follows with the co-induction proof principle that $F \circ G$ is the identity on $\TTT^\omega_D$.

\begin{proposition}\label{pn-seqtree}
$\TTT^\omega_D$ and $\TT_D$ are isomorphic $\Lambda_{D}$-co-algebras.
\end{proposition}

In what follows we will mostly assume that $D$ is finite, with one exception where in different tree levels different labels may be used. However, for each  level the set of labels used for nodes of that level is finite.  We will now extend the above consideration to this case.

Let $\vec{D} = (D_{n})_{n \in \NN}$ be a family of finite sets, and $\widehat{D} := \bigcup_{n \in \NN} D_{n}$. For a tree $S \subseteq \NN_{0}^{*}$ a labelling $\fun{L}{S}{\widehat{D}}$ is \emph{appropriate}, if for all $n \in \NN$ and all $s \in S$ of length $n$, $L(s) \in D_{n}$. Then
\[
\TTT^{\omega}_{\vec D} := \set{(S, L)}{\text{$S$ is a tree with appropriate compatible labelling $L$}}
\]
is the set of \emph{infinite $\vec D$-trees}. Obviously, every infinite $\vec D$-tree is an infinite $\widehat{D}$-tree with appropriate labelling, and vice versa.

For $m \in \NN$ let $\vec{D}^{(m)} := (D_{m+n})_{n \in \NN}$. Set $\overrightarrow{\TTT^{\omega}_{\vec{D}}} := (\TTT^{\omega}_{\vec{D}^{(n)}})_{n \in \NN}$. Moreover, for $n \in \NN$ and $T = (S, L) \in \TTT^{\omega}_{\vec{D}^{(n)}}$ define 
\begin{align*}
& \mathrm{root}_{n}(T) := L(\emptywd), \\
& \mathrm{subtree}_{n}(T) := ((S_{1}, L_{1}), \ldots, (S_{\ar{\mathrm{root}_{n}(T)}}, L_{\ar{\mathrm{root}_{n}(T)}})),
\end{align*}
where $S_{i} := \set{s \in \NN_{0}^{*}}{i \star s \in S}$ and $L_{i}(s) := L(i \star s)$, for $1 \le i \le \ar{\mathrm{root}_{n}(T)}$ and $s \in S_{i}$. Then $\mathrm{root}_{n}(T) \in D_{n}$ and $(S_{i}, L_{i}) \in \TTT^{\omega}_{\vec{D}^{(n+1)}}$.

Let $\mathbf{Set}^{\omega}$ be the $\omega$-fold product of the category $\mathbf{Set}$ with itself: objects are infinite sequences of sets  $(X_{n})_{n \in \NN}$ and morphisms infinite sequences of functions $(\fun{f_{n}}{X_{n}}{Y_{n}})_{n \in \NN}$ with componentwise composition. 
\begin{theorem}\label{thm-treecoalgvec}
Define the functor $\fun{\Lambda_{\vec D}}{\mathbf{Set}^{\omega}}{\mathbf{Set^{\omega}}}$ by 
\[
(\Lambda_{\vec D}(\vec X))_{n \in \NN} :=  (\bigcup_{d \in D_{n}} \{ d \} \times X_{n+1}^{\ar{d}})_{n \in \NN}.
\]
Then $(\overrightarrow{\TTT^\omega_{\vec D}}, (\mathrm{root}_{n} \times \mathrm{subtree}_{n})_{n \in \NN})$ is a terminal co-algebra of $\Lambda_{\vec D}$.
\end{theorem}

The set of all infinite $\vec D$-trees comes equipped with a canonical metric:
\[
\delta(T, T') = \begin{cases}
                                   0 & \text{if $T = T'$,}\\
                                   2^{-\min \{ n \mid T^{(n)} \not= T'^{(n)} \}} & \text{otherwise.}
                     \end{cases}
\]                                    
Obviously, $\delta(T, T') \le 1$, for all $T, T' \in \TTT^\omega_{\vec D}$, i.e.\ $(\TTT^\omega_{\vec D}, \delta)$ is bounded. By definition, the metric topology on $\TTT^\omega_{\vec D}$ is generated by the collection of all balls $\ball{\delta}{T}{2^{-n}}$ of radius $2^{-n}$ around $T$. 

\begin{proposition}\label{pn-treemcompl}
$(\TTT_{\vec D}^{\omega}, \delta)$ is complete.
\end{proposition}
\begin{proof}
Let $(T_i)_{i \in \NN}$ be a regular Cauchy sequence in $\TTT^\omega_{\vec D}$. Then $\delta(T_n, T_m) < 2^{-n}$, for all $m \ge n$. Thus, $T_m^{(n)} = T_n^{(n)}$, for all $m \ge n$. In particular, we have that $T^{(n)}_n = T^{(n)}_{n+1}$. By Lemma~\ref{lem-initseq}, $T^{(n)}_{n+1} \prec T^{(n+1)}_{n+1}$. Moreover, $T^{(n)}_n \in \TTT^{(n)}_{\widehat{D}}$. Therefore, $(T^{(n)}_n)_{n \in \NN} \in \TT_{\widehat{D}}$. Let $T := F((T^{(n)}_n)_{n \in \NN})$. Then $T \in \TTT^{\omega}_{\widehat{D}}$. In addition, $T$ has appropriate labelling. Hence, $T \in \TTT^{\omega}_{\vec D}$. By (\ref{eq-F(n)}),  $T^{(n)} = T^{(n)}_n = T^{(n)}_m$, for all $m \ge n$, which means that $\delta(T_m, T) < 2^{-n}$, for all $m \ge n$. Thus, $(T_i)_{i \in \NN}$ converges to $T$.
\end{proof}

Since there are only finitely many $\widehat{D}$-trees of height $n$ with appropriate labelling, for every natural number $n$, and, on the other hand, every infinite $\vec D$-tree has an initial segment of height $n$, $(\TTT^\omega_{\vec D}, \delta)$ is also totally bounded.

\begin{theorem}\label{thm-treecomp}
$(\TTT^\omega_{\vec D}, \delta)$ is compact.
\end{theorem}

\begin{corollary}\label{cor-separb}
$(\TTT^\omega_{\vec D}, \delta)$ is separable.
\end{corollary}
\begin{proof}
Consider the countable family of coverings $\CCC_{n} := \set{\ball{\delta}{x}{2^{-n}}}{x \in X}$ with $n \in \NN$ and apply compactness to obtain a countable dense set of points. 
\end{proof}

\section{Extended iterated function systems}\label{sec-ifs}

In the remainder of this paper $X$ is a Hausdorff space and $D$ a finite set of continuous self-maps on $X$. 

\begin{definition}\label{dn-ifs}
$(X, D)$ is an \emph{iterated function system (IFS)} if $X$ is a non-empty Hausdorff space and $D$ a finite set of unary self-maps $\fun{d}{X}{X}$.
\end{definition}

For our aims we will extend this notion by allowing the maps $d$ to be of any positive finite arity $\ar{d}$.

\begin{definition}\label{dn-eifs}
An \emph{extended} IFS $(X, D)$ consists of a non-empty Hausdorff space $X$ and a finite set $D$ of continuous maps $\fun{d}{X^{\ar{d}}}{X}$ of positive finite arity $\ar{d}$.
\end{definition}

\begin{definition}\label{dn-cov}
An extended IFS $(X, D)$ is \emph{covering} if 
\[
X = \bigcup\set{\range(d)}{d \in D},
\]
where $\range(d) := d[X^{\ar{d}}]$.
\end{definition}
In the context of iterated function systems the notion \emph{self-similar} is used instead. We think, however, that in a topological context the above notion is more appropriate.

Covering extended IFS can be characterised co-inductively. Define $\CC_X \subseteq  X$ by 
\[
\CC_X(x) \overset{\nu}{=} (\exists d \in D) (\exists y_1, \ldots, y_{\ar{d}} \in X)\, x = d(y_1, \ldots, y_{\ar{d}}) \wedge (\forall 1 \le \kappa \le \ar{d})\,  \CC_X(y_\kappa).
\]

\begin{lemma}\label{lem-coindX}
Let $(X, D)$ be covering. Then $X = \CC_X$.
\end{lemma}
\begin{proof}
By definition, $\CC_X \subseteq X$. The converse inclusion follows with co-induction. Observe to this end that because of the covering property the defining right hand side in the definition of $\CC_X$ remains true when $\CC_X$ is replaced by $X$.
\end{proof} 

Note that in general, $d$ and $y_1$, \ldots, $y_{\ar{d}}$ are not uniquely determined by $x$. So, there is no canonical way to turn $(X, D)$ into a $\Lambda_{D}$-co-algebra. Moreover, bear in mind that the lemma holds only classically: there is no way, in general, to compute $d$ and $y_1$, \ldots, $y_{\ar{d}}$.

Each finite $D$-tree $T$ defines a continuous map $\fun{f_T}{X^{\ar{T}}}{X}$ of arity $\ar{T}$:
\begin{itemize}
\item If $T = \{ d \}$, for some $d \in D$, then $\ar{T} = \ar{d}$ and $f_T = d$.

\item If $T = [d; T_1, \ldots, T_{\ar{d}}]$, then $\ar{T} = \sum\nolimits^{\ar{d}}_{\kappa = 1} \ar{T_\kappa}$ and $f_T = d \circ (f_{T_1} \times \cdots \times f_{T_{\ar{d}}})$.
\end{itemize}

The next result follows by induction on $n$.

\begin{lemma}\label{lem-ncov}
Let $(X, D)$ be covering. Then for every $n \in \NN$,
\[
X = \bigcup\set{\range(f_T)}{T \in \TTT^{(n)}_D}.
\]
\end{lemma}

Obviously, by increasing $n$ one obtains finer coverings. So, one can use finite $D$-trees as a road map to find elements of $X$. Can we therefore use infinite $D$-trees as their exact addresses?

\begin{lemma}\label{lem-compisec}
Let $(X, D)$ be an extended IFS so that $X$ is compact. Then, for any $T \in \TTT^\omega_D$,
\[
\bigcap\nolimits_{n \in \NN} \range(f_{T^{(n)}}) \not= \emptyset.
\]
\end{lemma}
\begin{proof}
By induction it follows that for every $T \in \TTT^\omega_D$ and each $n \in \NN$, $\range(f_{T^{(n+1)}}) \subseteq \range(f_{T^{(n)}})$. Moreover, $\range(f_{T^{(n)}}) \not= \emptyset$, as $X$ is not empty. Since $X$ is compact, the same holds for $X^{\ar{T^{(n)}}}$. Hence, $f_{T^{(n)}}[X^{\ar{T^{(n)}}}]$, as continuous image of a compact set, is compact as well and thus closed, since $X$ is Hausdorff. Thus, every finite intersection of members of the family of sets $\{ \range(f_{T^{(n)}}) \}_{n \in \NN}$ is non-empty and therefore, by the finite intersection property, also $\bigcap\nolimits_{n \in \NN} \range(f_{T^{(n)}})$ is non-empty. 
\end{proof}

If $X$ is a metric space and all maps in $D$ are contracting, the above intersection will also contain at most one element. However, this need not hold in general.

\begin{definition}\label{dn-propifs}
An extended IFS $(X, D)$ is 
\begin{enumerate}

\item \emph{compact} if $X$ is compact,

\item \emph{weakly hyperbolic}, if for all $T \in \TTT_D^{\omega}$,
$\card{\bigcap\nolimits_{n \in \NN} \range(f_{T^{(n)}})} \le 1$,

\end{enumerate}
\end{definition}

\begin{proposition}\label{pn-valT}
Let $(X, D)$ be compact and weakly hyperbolic. Then for all $T \in \TTT_D^{\omega}$,
\[
\card{\bigcap\nolimits_{n \in \NN} \range(f_{T^{(n)}})} = 1.
\]
\end{proposition}
We denote the uniquely determined element in $\bigcap\nolimits_{n \in \NN} \range(f_{T^{(n)}})$ by $\val{T}$. $\fun{\val{\cdot}}{\TTT^\omega_D}{X}$ is called the \emph{coding map}.

\begin{corollary}\label{cor-tree}
Let $(X, D)$ be compact as well as weakly hyperbolic, and $T \in \TTT^\omega_D$ with $T = [d;T_1, \ldots, T_{\ar{d}}]$. Then
\[
\val{T} = d(\val{T_1}, \ldots, \val{T_{\ar{d}}}).
\]
\end{corollary}
\begin{proof}
By definition, $\{ \val{T} \} = \bigcap_{n \in \NN} \range(f_{T^{(n)}})$. Thus,
\begin{align*}
\{ \val{T} \} 
&= \bigcap\nolimits_{n \ge 1} \range(f_{T^{(n)}})\\
&= \bigcap\nolimits_{n \ge 1} d[\range(f_{T_1^{(n-1)}}) \times \cdots \times \range(f_{T_{\ar{d}}^{(n-1)}})]\\
&\supseteq d[\bigcap\nolimits_{n \ge 1} \range(f_{T_1^{(n-1)}}) \times \cdots \times \range(f_{T_{\ar{d}}^{(n-1)}})]\\
&= d[\bigcap\nolimits_{n \ge 1} \range(f_{T_1^{(n-1)}}) \times \cdots \times \bigcap\nolimits_{n \ge 0} \range(f_{T_{\ar{d}}^{(n-1)}})]\\
&= \{ d(\val{T_1}, \ldots, \val{T_{\ar{d}}}) \},
\end{align*}
from which the statement follows.
\end{proof}

In what follows we mostly deal with extended IFS that are covering, compact and weakly hyperbolic.

\begin{definition}\label{dn-topifs}
A \emph{topological digit space} $(X, D)$ is a compact, covering and weakly hyperbolic extended IFS.
\end{definition}

The next question to be addressed is whether every element of $X$ can be coded in the above way. We start with a  technical result.

\begin{lemma}\label{lem-valonto}
Let $(X, D)$ be covering. Then for all $n \in \NN$, all $D$-trees $S \in \TTT^{(n)}_D$ and all $x \in \range(f_S)$ there is a $D$-tree $T \in \TTT^{(n+1)}_D$ with $x \in \range(f_{T})$ so that $S \prec T$.
\end{lemma}
\begin{proof}
We prove the statement by induction on $n$. If $n = 0$, $S = \{ e \}$, for some $e \in D$. Otherwise, $S = [e; S_1, \ldots, S_{\ar{d}}]$. Since $x \in \range(f_S)$, it follows that there is some $\vec y \in X^{\ar{e}}$, say $\vec y = (y_1, \ldots, y_{\ar{e}})$, so that $x = e(\vec y)$. If $n = 0$, $y_\kappa \in X$, for $1 \le \kappa \le \ar{e}$. Otherwise, $y_\kappa \in \range(f_{S_\kappa})$.

Let us first consider the case that $n =0$. As $(X, D)$ is covering, there is some $d_\kappa \in D$ and with $y_\kappa \in \range(d_\kappa)$, for each $1 \le \kappa \le \ar{e}$. Set $T := [e; d_1, \ldots, d_{\ar{e}}]$. Then $T \in \TTT^{(1)}$ so that $S \prec T$ and $x \in \range(f_T)$.

If $n > 0$, it follows by the induction hypothesis that for each $1 \le \kappa \le \ar{e}$ there exists $T_\kappa \in \TTT^{(n)}_D$ such that $y_\kappa \in \range(f_{T_\kappa})$ and $S_\kappa \prec T_\kappa$. Let $T := [e; T_1, \ldots, T_{\ar{e}}]$. Then $T \in \TTT^{(n+1)}_D$. Moreover, $S \prec T$ and $x \in \range(f_T)$.
\end{proof}

With the Axiom of Dependent Choice it now follows that for every $x \in X$ there is a sequence $(T_\nu)_{\nu \in \NN} \in \TT_D$ with $x \in\range( f_{T_\nu})$, from which in turn we obtain with Proposition~\ref{pn-seqtree} that there exists a $D$-tree $T \in \TTT^\omega_D$ with $x \in \range(f_{T^{(n)}})$, for every $n \in \NN$.

\begin{lemma}\label{lm-valonto}
Let $(X, D)$ be a topological digit space.
For all $n \in \NN$, every $S \in \TTT^{(n)}_D$ and all $x \in X$ with $x \in \range(f_S)$ there is some $T \in \TTT^\omega_D$ so that $T^{(n)} = S$ and $x = \val{T}$.
\end{lemma}

It follows that $\val{\cdot}$ is onto.

\begin{theorem}\label{thm-valcont}
Let $(X, D)$ be a topological digit space. Then the following four statements hold:
\begin{enumerate}
\item\label{thm-valcont-1} $\fun{\val{\cdot}}{\TTT^\omega_D}{X}$ is onto and uniformly continuous.

\item\label{thm-valcont-2} The topology on $X$ is equivalent to the quotient topology induced by $\val{\cdot}$.

\item\label{thm-valcont-3} $X$ is metrisable.

\item\label{thm-valcont-4} $X$ is separable.

\end{enumerate}
\end{theorem}
\begin{proof}
(\ref{thm-valcont-1}) Onto-ness is a consequence of what has just been said. Moreover, because of compactness we only need to prove continuity. We show that for each $T \in \TTT^\omega_D$ and every open set $O$ in the Hausdorff topology on $X$ with $\val{T} \in O$ there is some number $m$ such that $\range(f_{T^{(m)}}) \subseteq O$. 

Since $O$ is open, its complement is closed and hence compact. Moreover,
\[
\bigcap\nolimits_n \range(f_{T^{(n)}}) \cap (X \setminus O) = \emptyset,
\]
as $\bigcap\nolimits_n \range(f_{T^{(n)}}) = \{ \val{T} \}$ and $\val{T} \in O$. Because of the compactness of $X$ there is thus some $m \in \NN$ so that also
\[
\range(f_{T^{(m)}}) \cap (X \setminus O) = \emptyset
\]
Hence, $\range(f_{T^{(m)}}) \subseteq O$, which implies that for all $S \in \ball{\delta}{T}{2^{-m}}$, $\val{S} \in O$.

(\ref{thm-valcont-2}) As a continuous map on a compact Hausdorff space, $\val{\cdot}$ is a closed. Continuous closed maps are well known to be quotient maps (cf.\ \cite[Theorem 9.2]{wil}).

(\ref{thm-valcont-3}) $X$ is Hausdorff and, by statement (\ref{thm-valcont-1}), the continuous image of a compact metric space. As  consequence of Urysohn's metrisation theorem it is therefore metrisable (cf.~\cite[Corollary 23.2]{wil}).

(\ref{thm-valcont-4}) is a consequence of Corollary~\ref{cor-separb}, as continuous images of separable spaces are separable (cf.~\cite[Theorem 16.4a]{wil}).
\end{proof}

\begin{corollary}[\cite{ed}]\label{cor-arconv}
Let $(X, D)$ be a topological digit space, and $T \in \TTT^\omega_D$. Then for any $x \in X$,
\[
\lim_{n \to \infty} f_{T^{(n)}}(x^{(\ar{T^{(n)}})}) = \val{T},
\]
where for $n>0$, $x^{(n)} = (x, \ldots, x)$ ($n$ times).
\end{corollary}

By Theorem~\ref{thm-valcont}(\ref{thm-valcont-4}), $X$ is separable. As we will see, several dense subsets of $X$ can be constructed. Let to this end $x \in X$ and define
\[
Q^{(x)}_D := \set{f_S(x^{(\ar{S})})}{S \in \TTT^\ast_D}.
\]

\begin{lemma}\label{lem-cons}
Let $(X, D)$ be a topological digit space. Moreover, let $z \in X$. Then $Q^{(z)}_D$ is dense in $X$.
\end{lemma}
\begin{proof}
Let $O$ be open in the topology on $X$ and $z \in X$. Then are $T \in \TTT^\omega_D$ and $m \in \NN$ so that $f_{T^{(m)}}[X^{\ar{T^{(m)}}}] \subseteq O$. In particular, it follows that  $f_{T^{(m)}}(z^{(\ar{T^{(m)}})}) \in O$. Note that $T^{(m)} \in \TTT^\ast_D$.
\end{proof}

As follows from the second statement of the preceding theorem, $X$ is homeomorph to a quotient of $\TTT^\omega_D$. In the classical setting of IFS with all maps in $D$ being unary, Kameyama~\cite{ka} showed how the equivalence classes are generated by the kneading invariant of the system. As a consequence, the topology on $X$ is determined by the kneading invariant of the system. 

Theorem~\ref{thm-valcont}(\ref{thm-valcont-3}) shows that no generality is lost if we restrict our considerations to topological digit spaces $(X, D)$ such that $X$ is a metric space. The central problem of the research in \cite{ka} was the question whether the metrisation result can be improved in such a way that the maps in $D$ will be contracting as well. Kameyama gave an example showing that this is not the case in general. However, important function classes have been found since this study allowing such a choice, that is, if the maps in $D$ are chosen from one of these classes then there is a metric that as well generates the given topology on $X$ and turns the maps in $D$ into contractions. (See \cite{ba} for further hints.)  

\begin{definition}\label{dn-digit}
A \emph{digit space} $(X, D)$ is a compact covering extended IFS such that $X$ is a metric space, say with metric $\rho$, and all maps in $D$ are contracting, where powers of $X$ are endowed with the maximum metric\footnote{That is for $\vec x = (x_1, \ldots, x_n)$ and $\vec y = (y_1, \ldots, y_n)$, $\rho(\vec x, \vec y) = \max \set{\rho(x_\nu, y_\nu)}{1 \le \nu \le n}$. Note that we use the same notation for the metric and its associated maximum metric.}.
\end{definition}    

Every digit space is in particular a topological digit space.  
The maps in $D$ will also be called \emph{digit maps}, or simply \emph{digits}. As $X$ is compact and the metric $\rho$ is continuous, we have that $X$ is bounded, that is, there is a number $M \in \NN$, the \emph{bound} of $X$, so that for all $x, y, \in X$, $\rho(x, y) \le M$. Let $q < 1$ be the maximum of the contraction factors of the digit maps. Then for all $d \in D$ and $\vec x, \vec y \in X^{\ar{d}}$, $\rho(d(\vec x), d(\vec y)) \le q \cdot \rho(\vec x, \vec y)$. It follows that every digit space is weakly hyperbolic.

\section{Computable digit spaces}\label{sec-comput}

The aim of the research in \cite{be,bh,bs} as well as the present paper is to provide a logic-based approach to computing with continuous data. The generally accepted approach to compute with such data is Weihrauch's Type-Two Theory of Effectivity~\cite{wei}. In this section the equivalence of both approaches will be derived. To this end we restrict our considerations to digit spaces. This is the case generally used in application. Moreover, metric spaces have a well developed computability theory. We need to adapt the notions and proofs presented in \cite{bs} to the case of multi-ary digit functions. 

\begin{definition}[Brattka and Presser~\cite{bp}]\label{dn-compmet}
Let $(X, \rho)$ be a metric space with countable dense subspace $Q$, say
\[
Q = \{ u_0, u_1, \ldots \},
\]
the elements of which are called \emph{basic elements}. Then $(X, \rho, Q)$ is \emph{computable} if the two sets 
\begin{align*}
&\set{(u, v, r) \in Q \times Q \times \QQ}{\rho(u, v) < r} \\
&\set{(u, v, r) \in Q \times Q \times \QQ}{\rho(u, v) > r}
\end{align*}
are effectively enumerable, i.e.\ the function $\lambda (u, v) \in Q^2.\ \rho(u, v)$ is computable.
\end{definition}

Note that usually computability is defined on the natural numbers or the set of finite words over some finite alphabet. Computability in a more abstract setting is then reduced to this case by using appropriate coding functions. In what follows we will work with finite objects such as basic elements, tuples as well as finite sets of basic elements, digits or $D$-trees directly as in the above definition and leave it to the reader to make statements precise, if wanted. By doing so we will identify a digit $d$ with the letter $d$, as we did already in the preceding sections.

\begin{definition}\label{dn-compmap}
Let $(X, \rho, Q), (X', \rho', Q')$ be metric spaces with countable dense subspaces $Q$ and $Q'$, respectively. A  map $\fun{h}{X^i}{X}'$ is
\begin{enumerate}
\item\label{dn-compmap-1}
\emph{uniformly continuous} if there is a map $\fun{\zeta}{\QQ_{+}}{\QQ_{+}}$, called \emph{modulus of continuity}, such that for all $\varepsilon \in \QQ_{+}$ and $\vec{x}, \vec{y} \in X^{i}$, whenever $\rho(\vec{x}, \vec{y}) < \zeta(\varepsilon)$ then $\rho'(h(\vec{x}), h(\vec{y})) < \varepsilon$.

\item\label{dn-compmap-2}
 \emph{computable} if it has a computable modulus of continuity and there is a procedure $G_h$, which given $\vec u \in Q^i$ and $n \in \NN$ computes a basic element $v \in Q'$ with $\rho'(h(\vec{u}), v) < 2^{-n}$.
\end{enumerate}
\end{definition}

It is readily seen that the set of computable maps on $X$ is closed under composition.

\begin{definition}\label{dn-compdig}
Let $(X, D)$ be a digit space such that the underlying metric space $(X, \rho)$ has a countable dense subset $Q$ with respect to which it is computable. $(X, D, Q)$ is said to be a \emph{computable digit space} if, in addition, all digits $d \in D$ are computable.
\end{definition}

As will be shown in Section~\ref{sec-prod}, each power of a computable digit space is a computable digit space again.

Let
\begin{equation*}\label{eq-aeff}
\begin{split}
A_X^{\text{eff}} := \{\,x \in X \mid \text{there is a procedure that given $n \in \NN$}\hspace{3cm}\\
                \text{computes a basic element $u \in Q$ with $\rho(x, u) < 2^{-n}$} \,\}. 
\end{split}
\end{equation*}
 
 We have seen that besides $Q$ computable digit spaces possess other canonical dense subspaces $Q^{(z)}_D$, for $z \in A_X^{\text{eff}}$, generated by the digit maps. We want to show that $Q$ and $Q^{(z)}_D$ are \emph{effectively equivalent} in the sense that given $u \in Q$ and $n \in \NN$ a finite $D$-tree $S \in \TTT^\ast_D$ can be computed so that $\rho(u, f_S(z^{(\ar{S})})) < 2^{-n}$, and that similarly there is a computable function $\fun{\varphi}{\NN \times \TTT^\ast_D}{Q}$ with
 \[
 \rho(f_S(z^{(\ar{S})}), \varphi(n, S)) < 2^{-n},
 \]
 for all $n \in \NN$ and $S \in \TTT^\ast_D$. To achieve this, some additional conditions have to hold.
 
\begin{definition}\label{dn-wellcov}
An extended IFS $(X, D)$ is \emph{well-covering} if every element of $X$ is contained in the interior $\int(\range(d))$ of $\range(d)$, for some $d \in D$.
\end{definition}

\begin{lemma}[\cite{bs}]\label{lem-wcnum}
Let $(X, D)$ be a well-covering digit space. Then there exists $\varepsilon \in \QQ_+$ such that for every $x \in X$ there exists $d \in D$ with $\ball{\rho}{x}{\varepsilon} \subseteq \range(d)$.
\end{lemma}

Each such $\varepsilon \in \QQ_+$ will be called \emph{well-covering number}. Note that in the proof one only uses that the sets $\range(d)$ are closed; the $d$ are just indices. Obviously, if $\varepsilon$ is a well-covering number then every $\varepsilon' \in \QQ_+$ with $\varepsilon' < \varepsilon$ is a well-covering number as well. 

\begin{definition}\label{dn-deccd}
Let $(X, D, Q)$ be a computable digit space. We call $(X, D, Q)$
\begin{enumerate}
\item\label{dn-deccd-1} \emph{decidable} if for $u \in Q$, $\theta \in  \QQ_+$ and $d \in D$ it can be decided whether $\ball{\rho}{u}{\theta} \subseteq \range(d)$;

\item\label{dn-deccd-2} \emph{constructively dense} if there is a procedure that, given $\theta \in \QQ_+$, $d \in D$ and $u \in \range(d) \cap Q$, computes a $\vec v \in Q^{\ar{d}}$ with $\rho(u, d(\vec v)) < \theta$.
\end{enumerate}
\end{definition}

\begin{lemma}\label{lem-qqd}
Let $(X, D, Q)$ be a well-covering, decidable and constructively dense computable digit space and $z \in A_X^{\text{eff}}$. Then, for every $u \in Q$ and $n \in \NN$, a finite $D$-tree $S$ can effectively be found such that $\rho(u, f_S(z^{(\ar{S})})) < 2^{-n}$.
\end{lemma}
\begin{proof}
Let $M$ be a bound of $X$ and $q < 1$ the maximum of the contraction factors of the digit maps. Moreover,
let $\varepsilon$ be a well-covering number for $(X, D)$ and set
\begin{equation}\label{eq-j}
j(n) := \min \set{i \in \NN}{q^{i} \cdot M < 2^{-n}}.
\end{equation}
For $k \in \NN$ and $v \in Q$, let $H(k, v)$ be the following recursive procedure:
\begin{quote}
Use the decidability of $(X, D, Q)$ to find some $e \in D$ with $\ball{\rho}{v}{\varepsilon} \subseteq \range(e)$. If $k = 0$, output $e$. Otherwise, let
\[
\theta := q^{k} \cdot M/j(n+1)
\]
and use computable density to find some $\vec{v}' \in Q^{\ar{e}}$ such that $\rho(v, e(\vec{v}')) \le \theta$. Output the $D$-tree
\[
[e; H(k-1, v'_1), \ldots, H(k-1, v'_{\ar{e}})].
\]
\end{quote}
Let $T := H(k, v)$. We show by induction on $k$ that 
\[
\rho(v, f_T(z^{(\ar{T})})) \le (1 + k/j(n+1)) \cdot q^k \cdot M.
\]
If $k = 0$, we have that $T = e$, for some $e \in D$. Hence,
\[
\rho(v, e(z^{(\ar{e})})) \le M \le 1 \cdot q^0 \cdot M.
\]
If $k > 0$, there are $e \in D$, $\vec{v}' \in Q^{\ar{e}}$ and $T_1, \ldots, T_{\ar{e}} \in \TTT^\ast_D$ with $T_\kappa = H(k-1, v'_\kappa)$, for $1 \le \kappa \le \ar{e}$, so that $T = [e; T_1, \ldots, T_{\ar{e}}]$ and $\rho(v, e(\vec{v}')) \le q^k \cdot M/j(n+1)$. Then,
\begin{align*}
\rho(v, f_T(z^{(\ar{T})})) 
&\le \rho(v, e(\vec{v}')) + \rho(e(\vec{v}'), f_T(z^{(\ar{T})})) \\
&=  \rho(v, e(\vec{v}')) + \rho(e(\vec{v}'), e(f_{T_1}(z^{(\ar{T_1})}), \ldots, f_{T_{\ar{e}}}(z^{(\ar{T_{\ar{e}}})}))) \\
&\le q^k \cdot M/j(n+1) +  q \cdot \max_\kappa \rho(v'_\kappa, f_{T_\kappa}(z^{(\ar{T_\kappa})}))\\
&\le q^k \cdot M/j(n+1) +  q \cdot (1+(k-1))/j(n+1) \cdot q^{k-1} \cdot M\\
&= (1 + k/j(n+1)) \cdot q^k \cdot M.
\end{align*}
With $S := H(j(n), u)$ we thus have that $\rho(u, f_S(z^{(\ar{S})})) < 2^{-n}$.
\end{proof}

\begin{lemma}\label{lem-qdq}
Let $(X, D, Q)$ be a computable digit space and $z \in A_X^{\text{eff}}$. Then there is a procedure $H'$, which given $S \in \TTT^\ast_D$ and $n \in \NN$, produces a basic element $v \in Q$ so that $\rho(f_S(z^{(\ar{S})}), v) < 2^{-n}$.
\end{lemma}
\begin{proof}
Since $z \in A_X^{\text{eff}}$, there is a procedure $F$ which on input $n \in \NN$ computes a basic element $u \in Q$ with $\rho(z, u) < 2^{-n}$. Now, define $H'$ to be the following recursive procedure:

On input $n$, if $S = d$, first apply the procedure $F$ on input $n+1$, say the output is $u \in Q$, and then apply $G_d$ to $u^{(\ar{d})}$ and $n+1$. Otherwise, assume that $T = [e; T_1, \ldots, T_{\ar{e}}]$ and that the results of applying $H'$ on input $T_\kappa$ and $n+1$, for each $1 \le \kappa \le \ar{e}$, are $u_1, \ldots, u_{\ar{e}} \in Q$. Then output the result of applying $G_e$ to input $(u_1, \ldots, u_{\ar{e}})$ and $n+1$.  
\end{proof}

Summing up we obtain the following result.

\begin{proposition}\label{pn-beq}
Let $(X, D, Q)$ be a well-covering, decidable and constructively dense computable digit space. Then, for each $z \in A_X^{\text{eff}}$, the topological bases $Q$ and $Q^{(z)}_D$ are effectively equivalent.
\end{proposition}

Next, we will consider the set of computable elements of a digit space. Usually an element of some abstract space is considered computable when it has a computable representative. Elements of an IFS are represented by infinite $D$-trees. 

\begin{definition}\label{dn-comptree}
Let $(X, D)$ be an extended IFS. A $D$-tree $T \in \TTT^\omega_D$ is said to be \emph{computable} if its associated sequence $(T^{(n)})_{n \in \NN}$ of finite $D$-trees is computable.
\end{definition}

An element $x$ of an IFS $(X, D)$ is \emph{computable} if there is a computable tree $T \in \TTT^\omega_D$ with $\val{T} = x$. We denote the set of all computable elements of $X$ by $X_c$.

Let $(X, D)$ be a digit space and $x \in X$ be computable. Moreover, let this be witnessed by $T \in \TTT^\omega_D$. Then we have for any $z \in A_X^{\text{eff}}$ and $n \in \NN$ that $\rho(x, f_{T^{(n)}}(z^{\ar{T^{(n)}}})) \le q^n \cdot M$, where $q < 1$ is the maximum of the contraction factors of the digit maps and $M$ a bound of $X$. Assume that $(X, D, Q)$ is computable. Then it follows with Lemma~\ref{lem-qdq} that, for any given $n \in \NN$, we can compute a basic element $v \in Q$ with $\rho(f_{T^{j(n+1)}}(z^{(\ar{T^{j(n+1)}})}), v) < 2^{-n-1}$. Here, the function $j$ is as in (\ref{eq-j}). It follows that $\rho(x, v) < 2^{-n}$. This shows that $X_c \subseteq A_X^{\text{eff}}$. The converse implication will be a consequence of Theorem~\ref{thm-cacoind} derived in Section~\ref{sec-cauchy} in a constructive fashion. To this end a further condition is needed.

\begin{definition}\label{dn-appchoi}\!\!\!\footnote{This definition slightly differs from the one given in \cite{bs}.}
A computable digit space $(X, D, Q)$  has \emph{approximable choice} if for every $d \in D$ there is an effective procedure $\fun{\lambda (\theta, u).\, v^\theta_u}{\QQ_+ \times \int(\range(d)) \cap Q}{Q^{\ar{d}}}$ such that for all $\theta \in \QQ_+$:
\begin{enumerate}
\item\label{dn-appchoi-1}
For all $u \in \int(\range(d)) \cap Q$ and all $\tilde{\theta} \in \QQ_+$, $\rho(v^\theta_u, v^{\tilde{\theta}}_u) < \max\{ \theta, \tilde{\theta} \}$.

\item\label{dn-appchoi-2}
One can compute $\theta' \in \QQ_+$ such that for all $u, u' \in \int(\range(d)) \cap Q$, if $\rho(u, u') < \theta'$ then $\rho(v^\theta_u, v^\theta_{u'}) < \theta$.

\item\label{dn-appchoi-3}
For all $u \in \int(\range(d)) \cap Q$ there is some $\vec y \in d^{-1}[\{ u \}]$ with $\rho(\vec y, v^\theta_u) < \theta$.

\end{enumerate}
\end{definition}

Obviously, every computable digit space with approximable choice is constructively dense.

\begin{proposition}\label{pn-characcomp}
Let $(X, D, Q)$ be a well-covering and decidable computable digit space with approximable choice. Then $X_c = A^\text{eff}$.
\end{proposition}

\begin{proposition}\label{pn-charunicont}
Let $(X, D, Q_X)$ and $(Y, E, Q_Y)$ be computable digit spaces such that $(Y, E,\linebreak Q_Y)$ is decidable and well-covering. Then a map $\fun{f}{X^{\ar{f}}}{Y}$ is computable if, and only if, there is a computable map $\fun{\zeta}{\QQ_+}{\QQ_+}$ and a procedure $H$ so that for any $\varepsilon \in \QQ_+$ and every $\vec{u} \in Q_X^{\ar{f}}$, $H$ outputs a $v \in Q_Y$ with $f[\ball{\rho_X}{\vec{u}}{\zeta(\varepsilon)}] \subseteq \ball{\rho_Y}{v}{\varepsilon}$.
\end{proposition}
\begin{proof}
Assume that $\fun{f}{X^m}{Y}$ is computable. Then $f$ has a computable modulus of continuity $\fun{\zeta}{\QQ_+}{\QQ_+}$. Moreover, let $\varepsilon, \varepsilon' \in \QQ_+$ such that $\varepsilon$ is a well-covering number of $(Y, E, Q_Y)$. Set $\delta := \zeta(\varepsilon' / 2)$ and, without restriction, suppose that $\varepsilon' \le \epsilon$. Then for $\vec{u} \in Q^m$ use decidability to pick some $e \in E$ with $\ball{\rho_Y}{f(\vec{u})}{\bar{\varepsilon'}} \subseteq \range(e)$. Since $e$ is computable, we can effectively find some $v \in Q_Y$ with $\rho_Y(f(\vec{u}), v) < \varepsilon' / 2$. Then it follows for $\vec{x} \in \ball{\rho_x}{\vec{u}}{\delta}$ that $\rho_Y(f(\vec{x}), f(\vec{u})) < \varepsilon' / 2$. Hence, $\rho(f(\vec{x}), v) \le \rho_Y(f(\vec{x}), f(\vec{u})) + \rho_Y(f(\vec{u}), v) < \varepsilon'$. 

For the converse implication let $\varepsilon \in \QQ_+$ and $\vec{x} \in X^m$. Set $\delta := \zeta(\varepsilon / 2)$.  Because of the density of $Q_X$ there is some $\vec{u} \in Q_X^m$ with $\rho_X(\vec{x}, \vec{u}) < \delta / 2$. Then we have for $\vec{x}' \in X^m$ with $\rho_X(\vec{x}, \vec{x}') < \delta / 2$ that $\vec{x}, \vec{x}' \in \ball{\rho_X}{\vec{u}}{\delta}$, which implies that $f(\vec{x}), f(\vec{x}') \in \ball{\rho_Y}{v}{\varepsilon / 2}$, i.e., $\rho(f(\vec{x}), f(\vec{x}')) < \varepsilon$.

Since we have in addition that for any $n \in \NN$ and all $\vec{u} \in Q_X^m$ a $v \in Q_Y$ with $$f[\ball{\rho_X}{\vec{u}}{\zeta(2^{-n})}] \subseteq \ball{\rho_Y}{v}{2^{-n}}$$ can effectively be found, which in particular means that $\rho(f(\vec{u}), v) < 2^{-n}$, it follows that $f$ is computable.
\end{proof} 

Thus, the computable maps between decidable and well-covering computable digit spaces are exactly the maps that are uniformly continuous in a constructive sense.

\begin{definition}\label{dn-rinv}
For $\fun{f}{X^{\ar{f}}}{X}$, a map $\fun{f'}{\range(f)}{X^{\ar{f}}}$ is a \emph{right inverse} of $f$, if $f \circ f'$ is the identity on $\range(f)$. 
\end{definition}

\begin{proposition}\label{pn-epsrightinv}
A computable digit space $(X, D, Q)$ has approximable choice if, and only if, every $d \in D$ has a computable right inverse.
\end{proposition}
\begin{proof}
Assume that $(X, D, Q)$ has approximable choice, and let $d\in D$ and $x \in \range(d)$. Because of density there is some $u_m \in Q \cap \int(\range(d))$ with $\rho(x, u_m) < 2^{-m}$, for all $m \in \NN$. 

Use approximable choice to pick the function $\lambda(\theta,u).v^{\theta}_u$.  For $\theta_n := 2^{-n-4}$, pick $\theta' \in \QQ_+$ according to approximable choice, part (\ref{dn-appchoi-2}). Let $N_n \ge 0$ such that $\rho(x, u_m) < \theta'/3$, for $m \ge N_n$. Without restriction let $N_n$ be such that $N_n \ge N_i$, for all $i < n$. Set $v_n := v^{\theta_n}_{u_{N_n}}$.
By approximable choice, part (\ref{dn-appchoi-3}), there is some $z_n \in d^{-1}[\{ u_{N_n}\}]$ with
$\rho(z_n, v_n) < \theta_n$. Because of the assumption on $N_n$, we have that $\rho(u_{N_m}, u_{N_n}) < \theta'$, for $m \ge n$. Hence, $\rho(v_m, v_n) < \theta_n$. It follows that  $\rho(z_m, z_n) < 3\theta_n < 2^{-n}$. Thus, $(z_n)_{n \in \NN}$ is a regular Cauchy sequence. Since $(X^{\ar{d}},\rho)$ is complete, it converges to some $y_{(u_i)} \in X^{\ar{d}}$.  As $d$ is continuous, we obtain that 
\[
d(y_{(u_m)}) = \lim_{n \to \infty} d(z_n) = \lim_{n \to \infty} u_n = x.
\]

Now, let $x' \in \range(d)$ with $\rho(x, x') < \theta'/3$ as well as $u'_i \in Q \cap \int(\range(d))$ with $\rho(x', u'_m) < 2^{-m}$, for $m \in \NN$. Moreover, let $N'_n \ge 0$ such that $\rho(x', u'_m) < \theta'/3$, for $m \ge N'_n$. Without restriction assume that $N'_n \ge N'_i$, for all $i < n$. Finally, let $v'_n := v^{\theta_n}_{u'_{N'_n}}$ and $z'_n \in d^{-1}[\{ u_{N'_n}\}]$ with $\rho(z'_n, v'_n) < \theta_n$. Then $\rho(u_{N_n}, u'_{N'_n}) < \theta'$ and hence $\rho(v_n, v'_n) < \theta_n$. It follows that $\rho(z_n, z'_n) < 3\theta_n$ and thus $\rho(y_{(u_m)}, y_{(u'_m)}) < 9\theta_n < 2^{-n}$.

For $x=x'$, we obtain that $y_{(u_m)} = y_{(u'_m)}$, i.e., $y$ does not depend on the choice of the approximating sequence $(u_m)_{m \in \NN}$. Define $d'(x) := y$. By what we have just shown, $d'$ is uniformly continuous with computable modulus of continuity. Moreover, since $\rho(d'(x), v_n) \le \rho(d'(x), z_n) + \rho(z_n, v_n) < 4\theta_n = 2^{-n}$, it follows that $d'$ is also computable.

Conversely, let $d'$ be a right inverse of $d$. For $\theta \in \QQ_+$ let 
\[
m(\theta) := \min\set{m\in \NN}{2^{-m} \le \theta}.
\]
Since $d'$ is computable, we can compute for any given $u \in Q \cap \int(\range(d))$ and $n \in \NN$ a basic element $v_n \in Q^{\ar{d}}$ so that $\rho(d'(u), v_n) < 2^{-n}$. Set $v^{\theta}_u := v_{m(\theta)+2}$. It remains to verify the conditions in Definition~\ref{dn-appchoi}:

(\ref{dn-appchoi-1}) Let $\theta, \bar{\theta} \in \QQ_+$. Without restriction let $\theta \ge \bar{\theta}$. Then 
\[
\rho(v^{\theta}_u, v^{\bar{\theta}}_u) \le \rho(v^{\theta}_u, d'(u)) + \rho(d'(u), v^{\bar{\theta}}_u) < 2^{-m(\theta)-1} < \theta.
\]

(\ref{dn-appchoi-2}) As $d'$ has a computable modulus of continuity, for given $\theta \in \QQ_+$ we can compute a $\theta' \in \QQ_+$ such that for $u, u' \in Q \cap \int(\range(d))$, if $\rho(u, u') < \theta'$ then $\rho(d'(u), d'(u')) < \theta/2$, from which it follows that $\rho(v^{\theta}_u, v^{\theta}_{u'}) < \theta$. 

(\ref{dn-appchoi-3}) is obvious: choose $\vec{y} := d'(u)$.
\end{proof}

In Type-Two Theory of Effectivity an element $x \in X$ is defined to be computable, if it is contained in $A^\text{eff}_X$. So, it follows that both computability notions coincide. In the present approach elements of $X$ are represented by infinite $D$-trees $T$, and/or the corresponding sequences $(T^{(n)})_{n \in \NN}$ of initial segments. Similarly, in Type-Two Theory of Effectivity an element $x$ is represented by an infinite sequence $(u_n)_{n \in \NN}$ of basic elements with $\rho(x, u_n) < 2^{-n}$. The resulting representation is called \emph{Cauchy representation} $\rho_C$. As follows from the results in this section, one can computably pass from an infinite stream $T^{(0)} : T^{(1)} : \cdots$ of finite $D$-trees to an infinite sequence $(u_n)_{n \in \NN}$ of basic elements so that $\rho(\val{T}, u_n) < 2^{-n}$, and vice versa. This means that there are computable translations between both representations as summarised by the next result.

\begin{theorem}\label{thm-compeq}
Let $(X, D, Q)$ be a well-covering and decidable computable digit space with approximable choice. Then there are computable operators $\fun{F}{\TTT^\omega_D}{Q^\omega}$ and $\pfun{G}{Q^\omega}{\TTT^\omega_D}$ such that for $T \in \TTT^\omega_D$ and $w \in \dom(G)$,\footnote{We use the notation $\pfun{f}{X}{Y}$ to denote partial maps $f$ from $X$ to $Y$ with $\dom(f)$ as its domain of definition.}
\[
\rho_C(F(T)) = \val{T} \quad\text{and}\quad \val{G(w)} = \rho_C(w).
\]
\end{theorem}

\section{Extracting digital trees from co-inductive proofs}\label{sec-progex}

In this section we recast the theory of topological digit spaces in a constructive setting with the aim to extract programs that provide effective representations of certain objects or transformations between different representations. As one of the main results on this basis we will obtain effective transformations between the Cauchy representation of digit spaces and the digital tree representation showing that the two representations are effectively equivalent. The method of program extraction is based on a version of realisability, and the main constructive definition and proof principles will be induction and co-induction. The advantage of the constructive approach lies in the fact that proofs can be carried out in a representation-free way. Constructive logic and the Soundness Theorem guarantee automatically that proofs are witnessed by effective and provably correct transformations on the level of representations.

Regarding the theory of realisability and its applications to constructive analysis we refer the reader to Schwichtenberg and Wainer~\cite{sw}, Berger and Seisenberger~\cite{bse} and Berger~\cite{be}. Here, we only recall main facts. We largely follow the exposition in~\cite{bs}.  The logic used is many-sorted first-order logic extended by the formation of inductive and co-inductive predicates. Note that although the logic is based on intuitionistic logic a fair amount of classical logic is available. For example, any disjunction-free formula that is classically true may be admitted as an axiom~\cite{bt}.

Realisability assigns to each formula $A$ an unary predicate $\mathbf{R}(A)$ to be thought of as the set of realisers of $A$. Instead of $\mathbf{R}(A)(a)$ one often writes $a \br A$ (``$a$ realises $A$''). The realiser $a$ can be typed or untyped, but for the understanding of what follows, details about the nature of realisers are irrelevant. It suffices to think of them as being (idealised, but executable) functional programs or (Oracle-)Turing machines. The crucial clauses of realisability for the propositional connectives are: 
\begin{align*}
&c \br (A \vee B) := (\exists a) (c = (0, a) \wedge a \br A) \vee (\exists b) (c = (1, b) \wedge b \br B) \\
&f \br (A \to B) := (\forall a) (a \br A \to f(a) \br B) \\
&c \br (A \wedge B) := \mathbf{p}_0(c) \br A \wedge \mathbf{p}_1(c) \br B \\
&c \br \bot := \bot.
\end{align*}
Hence, an implication is realised by a function and a conjunction by a pair (accessed by left and right projections, $\mathbf{p}_0(\cdot)$, $\mathbf{p}_1(\cdot)$).

Quantifiers are treated uniformly in this version of realisability:
\begin{align*}
&a \br (\forall x)\, A(x) := (\forall x)\, a \br A(x) \\
&a \br (\exists x)\, A(x) := (\exists x)\, a \br A(x).
\end{align*}
By this way, variables $x$ are allowed to range over abstract mathematical objects without prescribed computational meaning. Therefore, the usual interpretation of $a \br (\forall x)\, A(x)$ to mean $(\forall x)\, a(x) \br A(x)$ makes no sense as one would use the abstract object $x$ as input to the program $a$. 
 
For atomic formulas $P(\vec t)$, where $P$ is a predicate and $\vec t$ are terms, realisability is defined in terms of a chosen predicate $\tilde{P}$ with one extra argument place, that is,
\begin{equation}\label{realis}
a \br P(\vec t) := \tilde{P}(a,\vec t).
\end{equation}
The choice of the predicate $\tilde{P}$ allows to fine tune the computational content of proofs. 

So far, only first-order logic has been covered. Next, we explain how inductive and co-inductive definitions are realised. An inductively defined predicate $P$ is defined as the least fixed point of a monotone predicate transformer $\Phi(X, \vec x)$, that is the formula $(\forall \vec x)(X(\vec x) \to Y(\vec x)) \to (\forall \vec x)(\Phi(X, \vec x) \to \Phi(Y, \vec x))$, with free predicate variables $X$ and $Y$, must be provable. Then one has the closure axiom
\[
(\forall \vec x)(\Phi(P, \vec x) \to P(\vec x))
\]
as well as the induction schema
\[
(\forall \vec x)(\Phi(\AAA, \vec x) \to \AAA(\vec x)) \to (\forall \vec x)(P(\vec x) \to \AAA(\vec x))
\]
for every predicate $\AAA$ defined by some formula $A(\vec x)$ as $\AAA(\vec x) \leftrightarrow A(\vec x)$.
Realisability for $P$ is defined as in (\ref{realis}) by defining $\tilde{P}$ inductively via the operator $\tilde{\Phi}(\tilde{X}, a, \vec x)$, where $\tilde{\Phi}$ is obtained from $\Phi$ by replacing every occurrence of the form $a \br X(\vec x)$ in the expression obtained by unravelling the formula $a \br \Phi(X, \vec{x})$ according to the definition of $\Phi$, by $\tilde{X}(a, \vec x)$, for a fresh predicate variable $\tilde{X}$. Then one has the closure axiom
\[
(\forall a, \vec x)(a \br \Phi(P, \vec x) \to \tilde{P}(a, \vec x))
\]
as well as the induction schema:
\[
(\forall a, \vec x)(a \br \Phi(\AAA, \vec x) \to a \br \AAA(\vec x)) \to (\forall a, \vec x)(\tilde{P}(a, \vec x) \to a \br \AAA(\vec x)).
\]

Dually, $\Phi$ also gives rise to a co-inductively defined predicate $Q$ defined as the greatest fixed point of $\Phi$. Hence, one has the co-closure axiom
\[
(\forall \vec x)(Q(\vec x) \to \Phi(Q, \vec x))
\]
and the co-induction schema:
\[
(\forall \vec x)(\AAA(\vec x) \to \Phi(\AAA, \vec x)) \to (\forall \vec x)(\AAA(\vec x) \to Q(\vec x)).
\]
Realisability for $Q$ is defined by defining $\tilde{Q}$ co-inductively by the same operator $\tilde{\Phi}$ as above, hence, the co-closure axiom
\[
(\forall a, \vec x)(\tilde{Q}(a, \vec x) \to a \br \Phi(Q, \vec x))
\]
and the co-induction schema:
\[
(\forall a, \vec x)(a \br \AAA(\vec x) \to a \br \Phi(\AAA, \vec x)) \to (\forall a, \vec x)(a \br \AAA(\vec x) \to \tilde{Q}(a, \vec x)).
\]

The basis of program extraction from proofs is the Soundness Theorem.

\begin{theorem}[\bf Soundness Theorem \cite{ub,bse2}]
From a constructive proof of a formula $A$ from assumptions $B_1, \ldots, B_n$ one can extract a program $M(a_1, \ldots, a_n)$ such that $M(a_1, \ldots, a_n) \br A$ is provable from the assumptions $a_1 \br B_1, \ldots, a_n \br B_n$.
\end{theorem}  

If one wants to apply this theorem to obtain a program realising formula $A$ one must provide terms $K_1, \ldots, K_n$ realising the assumptions $B_1, \ldots, B_n$. Then it follows that the term $M(K_1, \ldots, K_n)$ realises $A$.

That realisers do actually \emph{compute} witnesses is shown in Berger~\cite{ub} and Berger and Seisenberger~\cite{bse} by a \emph{Computational Witness Theorem} that relates the denotational definition of realisability with a lazy operational semantics.

There is an important class of formulas where realisers do not matter: A formula $B$ is \emph{non-computational} if
\[
(\forall a)(a \br B \leftrightarrow B).
\]
Non-computational formulas can simplify program extraction of realisers dramatically. We will, however, not go into further details here and refer the reader to~\cite{bs}. 

In formalising the theory of topological digit spaces, the real number set as well as the underlying space $X$ are regarded as a sort. All arithmetic constants and functions we wish to talk about as well as the metric $\rho$, in case $X$ is a metric space, are admitted as constant or function symbols. The predicates $=$, $<$ and $\le$ are considered as non-computational. Furthermore, all true non-computational statements about real number as well as the axioms of a metric space are admitted as axioms. 

In order to be able to deal with the hyperspace of non-empty compact sets, which will be studied in Section~\ref{sec-compact}, a powersort $\PPP(x)$ is added for every sort $x$, equipped with a non-computational element-hood relation $\epsilon$, as well as a function sort $s \to t$ for any two sorts $s$ and $t$, equipped with an application operation and operations such as composition. In addition, for every non-computational formula $A(x)$ the comprehension axiom
\[
(\exists u)(\forall x)(x\, \epsilon\, u \leftrightarrow A(x))
\]
is added. ($A(x)$ may contain other free variables than $x$.) This is an example of a non-computational formula we wish to accept as true. Again, we refer to~\cite{bs} for further details and examples.

 In Lemma~\ref{lem-coindX} a co-inductive characterisation for covering extended IFS was derived. Classically, this is rather uninteresting, but, constructively, it is significant, since, as we will see next, from a constructive proof of $\CC_X(x)$ one can extract a $D$-tree $T \in \TTT^\omega_D$ so that $\val{T} = x$.
 
 \begin{theorem}\label{thm-treerealis}
 Let $(X, D)$ be a topological digit space. Then the realisers of a statement $\CC_X(x)$ are exactly the  $D$-trees $T \in \TTT^\omega_D$ representing $x$, that is
 \[
 T \br (\CC_X(x)) \Longleftrightarrow \val{T} = x.
 \]
 In particular, from a constructive proof of $\CC_X(x)$ one can extract an infinite $D$-tree representation of $x$.
 \end{theorem}
 \begin{proof}
 By the realisability definition above the predicate $T \br \CC_X(x)$ is defined co-inductively as
 \begin{equation}\label{eq-realis}
 \begin{split}
 [d; T_1, \ldots, T_{\ar{d}}] \br \CC_X(x) \overset{\nu}{=}\hspace{9cm}\\ (\exists y_1, \ldots, y_{\ar{d}} \in X)\, x = d(y_1, \ldots, y_{\ar{d}}) \wedge (\forall 1 \le \kappa \le \ar{d})\, T_\kappa \br \CC_X(y_\kappa).
 \end{split}
 \end{equation}
 This allows us to show the `if' part by co-induction. That means we have to show that the implication from left to right in (\ref{eq-realis}) holds if the relation $\cdot \br \cdot$ is replaced by the relation $\val{\cdot} = \cdot$. This, however, is  consequence of Corollary~\ref{cor-tree}.
 
 For the converse implication it is sufficient to show that
 \[
 (\forall n \in \NN)(\forall T \in \TTT^\omega_D)(\forall x \in X)\, (T \br \CC_X(x) \Rightarrow x \in \range(f_{T^{(n)}})),
 \]
 which we do by induction on $n$. Assume that $T \br \CC_X(x)$. Then $T = [d; T_1, \ldots, T_{\ar{d}}]$ and there are $y_1, \ldots, y_{\ar{d}} \in X$ so that $x = d(y_1, \ldots, y_{\ar{d}})$ and $T_\kappa \br  \CC_X(y_\kappa)$, for $1 \le \kappa \le \ar{d}$. It follows that $x \in \range(d)$. Thus the case $n=0$ is true. If $n > 0$, it follows by the induction hypothesis that $y_\kappa \in \range(f_{T_\kappa^{(n-1)}})$. Hence, $x \in \range(f_{T^{(n)}})$. 
\end{proof}

\section{Equivalence with the Cauchy representation}\label{sec-cauchy}

As pointed out above, we will now derive the equivalence between the Cauchy representation used in Type-Two Theory of Effectivity and the tree representation introduced here in a constructive fashion. Formalised in many-sorted first-order logic extended by the formation of inductive and co-inductive predicates, the proofs allow the extraction of programs computing translations between the two representations. The method of proof extraction is based on a version of realisability. As said, the advantage of the constructive approach lies in the fact that proofs can be carried out in a representation-free way. 

Let $(X, D, Q)$ be a computable digit space and let the predicate $\bA_X \subseteq X$  be defined by:
\[
\bA_X(x) :\Leftrightarrow (\forall n \in \NN)(\exists u \in Q)\, \rho(x, u) < 2^{-n}.
\]
A realiser of $\bA_X(x)$ is a regular Cauchy sequence in $X$ converging to $x$ and a realiser of $\CC_X(x)$ is a $D$-tree $T \in \TTT^\omega_D$ such that $x = \val{T}$.

\begin{theorem}\label{thm-ctoa}
Let $(X, D, Q)$ be a computable digit space. Then $\CC_X \subseteq \bA_X$.
\end{theorem}
\begin{proof}
Fix $z \in \bA_X^\text{eff}$. Because of Lemma~\ref{lem-qdq} it suffices to show that
\[
(\forall n \in \NN)(\forall x \in X)(\CC_X(x) \rightarrow (\exists u \in Q^{(z)}_D)\, \rho(x, u) \le M \cdot q^n),
\]
which will be done by induction on $n$. If $n = 0$, let $u$ be any element in $Q^{(z)}_D$. For $n+1$, assume $\CC_X(x)$. Then there are $d \in D$ and $y_1, \ldots, y_{\ar{d}} \in X$ with $\CC_X(y_\kappa)$, for $1 \le \kappa \le r$, so that $x = d(y_1, \ldots, y_{\ar{d}})$. By induction hypothesis, there exist $v_1, \ldots, v_{\ar{d}} \in Q^{(z)}_D$ such that for all $1 \le \kappa \le \ar{d}$, $\rho(y_\kappa, v_\kappa) < M \cdot q^n$. Set $u := d(v_1, \ldots, v_{\ar{d}})$. Then $u \in Q^{(z)}_D$ and $\rho(x, u) \le q \cdot \max_\kappa \rho(y_\kappa, v_\kappa) \le M \cdot q^{n+1}$.
\end{proof}

\begin{theorem}\label{thm-cacoind}
Let $(X, D, Q)$ be a well-covering and decidable computable digit space with approximable choice. Then $\bA_X \subseteq \CC_X$.
\end{theorem}
\begin{proof}
The theorem is derived by co-induction. Hence assume $\bA_X(x)$. We have to find $d \in D$ and $y_1, \ldots, y_{\ar{d}} \in X$ so that $x = d(y_1, \ldots, y_{\ar{d}})$ and $\bA_X(y_\kappa)$, for $1 \le \kappa \le \ar{d}$. Let $\varepsilon \in \QQ_+$ be a well-covering number. Using $\bA_X(x)$, pick $\hat{u}\in Q$ such that $\rho(x,\hat{u}) < \varepsilon/2$. Pick $d\in D$ such that $\ball{\rho}{\hat{u}}{\varepsilon}\subseteq \range(d)$. Then $x \in \ball{\rho}{\hat{u}}{\varepsilon}$. 

By Proposition~\ref{pn-epsrightinv}, $d$ has a computable right inverse $d'$. Set $\vec{y} := d'(x)$. Since $d'$ has a computable modulus of continuity, we can, given $n \in \NN$, compute a number $k(n)$ so that for $x', x'' \in \range(d)$, if $\rho(x', x'') < 2^{-k(n)}$ then $\rho(d'(x'), d'(x'')) < 2^{-n-1}$. Using assumption $\bA_X(x)$ again, we find $u\in Q$ such that $\rho(x, u)< 2^{-k(n)}$. It follows that $\rho(d'(x), d'(u)) < 2^{-n-1}$. By the computability of $d'$ we can moreover compute a basic element $\vec{v} \in Q^{\ar{d}}$ with 
$\rho(d'(u), \vec{v}) < 2^{-n-1}$. Hence, $\rho(\vec{y}, \vec{v}) < 2^{-n}$. Let $\vec{y} = (y_1, \ldots, y_{\ar{d}})$ and $\vec{v} = (v_1, \ldots, v_{\ar{d}})$, then we have that $\rho(y_\kappa, v_\kappa) < 2^{-n}$, for $1 \le \kappa \le \ar{d}$,  which shows that $\bA_X(y_\kappa)$. 
\end{proof}

\section{Products}\label{sec-prod}

In this and the following section we study how canonical IFS structures can be introduced on spaces obtained by the usual constructions of new spaces from given ones, and whether the properties examined so far are inherited in these cases. We start with the product construction.

Let $X_1, \ldots, X_n$ be non-empty topological spaces and $X_1 \times \cdots \times X_n$ endowed with the product topology. As is well known, $X_1 \times \cdots \times X_n$ is Hausdorff, exactly if $X_i$ is Hausdorff, for each $1 \le i \le n$; and analogously for compactness by Tychonov's Theorem~\cite{wil}.

Now, assume that $(X_1, D_1)$, \ldots, $(X_n, D_n)$ are extended IFS. Without restriction suppose that all $d \in \bigcup_{i=1}^n D_i$ have the same arity, say $s_{D}$. Otherwise, let $s_{D} := \max\set{\ar{d}}{d \in \bigcup_{i=1}^n D_i}$ and replace $d \in D_i$ by $\hat{d}$ defined by
\[
\hat{d}(x_1, \ldots, x_{s_{D}}) := d(x_1, \ldots, x_{\ar{d}}),
\]
for $x_1, \ldots, x_{s_{D}} \in X_i$.

We will introduce an IFS structure on $\bigtimes_{i=1}^{n} X_i$. For $(d_1, \ldots, d_n) \in \bigtimes_{i=1}^{n} D_{i}$ define 
\[
\fun{\pair{d_1, \ldots, d_n}}{(\bigtimes_{i=1}^{n} X_i)^{s_{D}}}{\bigtimes_{i=1}^{n} X_i}
\]
by
\begin{equation*}\begin{split}
\pair{d_1, \ldots, d_n}((x^{(1)}_{1}, \ldots, x^{(1)}_{n}), \ldots, (x^{(s_{D})}_{1}, \ldots, x^{(s_{D})}_{n})) := \hspace{4cm}\mbox{}\\
(d_1(x^{(1)}_{1}, \ldots, x^{(s_{D})}_{1}), \ldots, d_{n}(x^{(1)}_{n}, \ldots, x^{(s_{D})}_{n})),
\end{split}\end{equation*}
for $(x^{(1)}_{1}, \ldots, x^{(1)}_{n}), \ldots, (x^{(s_{D})}_{1}, \ldots, x^{(s_{D})}_{n}) \in \bigtimes_{i=1}^{n} X_{i}$,
and let 
\[
D^\times := \set{\pair{d_1, \ldots, d_n}}{(d_1, \ldots, d_n) \in \bigtimes_{i=1}^{n} D_i}.
\]
\begin{proposition}\label{pn-prod}
Let $(X_1, D_1), \ldots, (X_n, D_n)$ be extended IFS. Then also $(\bigtimes_{i=1}^n X_i, D^\times)$ is an extended IFS. Moreover, the following statements hold:
\begin{enumerate}
\item\label{pn-prod-1}
If $(X_i, D_{i})$ is compact, for all $1 \le i \le n$, so is $(\bigtimes_{i=1}^n X_i, D^\times)$.

\item\label{pn-prod-2}
If $(X_i, D_i)$ is covering, for all $1 \le i \le n$, so is $(\bigtimes_{i=1}^n X_i, D^\times)$.

\item\label{pn-prod-3}
If $(X_i, D_i)$ is well-covering, for all $1 \le i \le n$, so is $(\bigtimes_{i=1}^n X_i, D^\times)$.

\end{enumerate}
\end{proposition}
\begin{proof}
(\ref{pn-prod-1}) follows with Tychonov's Theorem.

(\ref{pn-prod-2}) As is ensued by the definition of $\pair{d_1, \ldots, d_n}$,
\begin{equation}\label{eq-prod-2}
\range(\pair{d_1, \ldots, d_n}) = \bigtimes_{i=1}^n \range(d_i).
\end{equation}
Hence,
\begin{align*}
&\bigcup\set{\range(\pair{d_1, \ldots, d_n})}{(d_1, \ldots, d_n) \in \bigtimes_{i=1}^n D_i} \\
&\quad= \bigcup\set{\bigtimes_{i=1}^n \range(d_i)}{(d_1, \ldots, d_n) \in \bigtimes_{i=1}^n D_i} \\
&\quad= \bigtimes_{i=1}^n \bigcup\set{\range(d_i)}{d_i \in D_i} \\
&\quad= \bigtimes_{i=1}^n X_i.
\end{align*}

(\ref{pn-prod-3}) follows similarly as $\int(\bigtimes_{i=1}^n \range(d_i)) = \bigtimes_{i=1}^n \int(\range(d_i))$.
\end{proof}

When dealing with a tuple $\vec x$ via the extended IFS $(\bigtimes_{i=1}^{n} X_{i}, D^{\times})$, $\vec x$ is considered as a structureless object, not as a composed one. We will now investigate how from a tree representing $\vec x$ we can access the components of $\vec x$.

For $1 \le i \le n$ define the map $\fun{\Pr^{(n)}_{i}}{\TTT^{\omega}_{D^{\times}}}{\TTT^{\omega}_{D_{i}}}$ co-recursively by
\begin{align*}
&\rt{\Pr^{(n)}_{i}([\pair{d_{1}, \ldots, d_{n}}; T_{1}, \ldots, T_{s_{D}}])} := d_{i} \\
&\sbt{\Pr^{(n)}_{i}([\pair{d_{1}, \ldots, d_{n}}; T_{1}, \ldots, T_{s_{D}}])} := (\Pr^{(n)}_{i}(T_{1}), \ldots, \Pr^{(n)}_{i}(T_{s_{D}})).
\end{align*}
Similarly, define $\fun{\Cons^{(n)}}{\bigtimes_{i=1}^{n }\TTT^{\omega}_{D_{i}}}{\TTT^{\omega}_{D^{\times}}}$ by
\begin{align*}
&\rt{\Cons^{(n)}([d_{1}; T^{(1)}_{1}, \ldots, T^{(1)}_{s_{D}}], \ldots, [d_{n}; T^{(n)}_{1}, \ldots, T^{(n)}_{s_{D}}])} := \pair{d_{1}, \ldots, d_{n}} \\
&\sbt{\Cons^{(n)}([d_{1}; T^{(1)}_{1}, \ldots, T^{(1)}_{s_{D}}], \ldots, [d_{n}; T^{(n)}_{1}, \ldots, T^{(n)}_{s_{D}}])} := \\
&\hspace{4.5cm} ( \Cons^{(n)}(T^{(1)}_{1}, \ldots, T^{(n)}_{1}), \ldots, \Cons^{(n)}(T^{(1)}_{s_{D}}, \ldots, T^{(n)}_{s_{D}})).
\end{align*}

\begin{proposition}\label{pn-prodtree}
Let $(X_1, D_1), \ldots, (X_n, D_n)$ be extended IFS. Then for every infinite $D^{\times}$-tree $T \in \TTT^{\omega}_{D^{\times}}$,
\[
T = \Cons^{(n)}(\Pr^{(n)}_{1}(T), \ldots, \Pr^{(n)}_{n}(T)).
\]
\end{proposition}
\begin{proof}
Apply the co-induction proof principle.
\end{proof}

This result allows us to study the inheritability of weak hyperbolicity of extended IFS to their product.
\begin{lemma}\label{lem-rgprod}
For all $m \in \NN$ and all  $T \in \TTT^{\omega}_{D^{\times}}$,
\[
\range(f_{T^{(m)}}) = \bigtimes_{i=1}^{n} \range(f_{\Pr^{(n)}_{i}(T)^{(m)}}).
\]
\end{lemma}
\begin{proof}
We proof the statement by induction on $m$. Let $T = [\pair{d_{1}, \ldots, d_{n}}; T_{1}, \ldots, T_{s_{D}}]$. 

In case $m=0$, we have that $T^{(0)} = \{ \pair{d_{1}, \ldots, d_{n}} \}$ and $\Pr^{(n)}_{i}(T)^{(0)} = \{ d_{i} \}$. Hence, $f_{T^{(0)}} = \pair{d_{1}, \ldots, d_{n}}$ and $f_{\Pr^{(n)}_{i}(T)^{(0)}} = d_{i}$. As seen in (\ref{eq-prod-2}), the statement thus holds in  this case.

Now, suppose that $m>0$. Then 
\[
f_{T^{(m)}} 
= \pair{d_{1}, \ldots, d_{n}} \circ (f_{T^{(m-1)}_{1}} \times \cdots \times f_{T^{(m-1)}_{s_{D}}})
\]
and
\[
f_{\Pr^{(n)}_{i}(T)^{(m)}} = d_{i} \circ (f_{\Pr^{(n)}_{i}(T_{1})^{(m-1)}} \times \cdots \times f_{\Pr^{(n)}_{i}(T_{s_{D}})^{(m-1)}}).
\]
Therefore, it follows with the induction hypothesis that
\begin{align*}
\range(f_{T^{(m)}})
&= \pair{d_{1}, \ldots, d_{n}}[\bigtimes_{i=1}^{s_{D}} \range(f_{T^{(m-1)}_{i}})] \\
&= \pair{d_{1}, \ldots, d_{n}}[\bigtimes_{i=1}^{s_{D}} \bigtimes_{j=1}^{n} \range(f_{\Pr^{(n)}_{j}(T_{i})^{(m-1)}})] \\
&= \bigtimes_{j=1}^{n} d_{j}[\bigtimes_{i=1}^{s_{D}} \range(f_{\Pr^{(n)}_{j}(T_{i})^{m-1}})] \\
&= \bigtimes_{j=1}^{n} \range(f_{\Pr^{(n)}_{j}(T)^{m-1}}). \qedhere
\end{align*}
\end{proof}

As a consequence of the preceding lemma we obtain for $T \in \TTT^{\omega}_{D^{\times}}$ that
\[
\bigcap_{m \ge 1} \range(f_{T^{(n)}})
= \bigcap_{m \ge 1} \bigtimes_{i=1}^{n} \range(f_{\Pr^{(n)}_{i}(T)^{(m)}}) 
= \bigtimes_{i=1}^{n}\bigcap_{m \ge 1} \range(f_{\Pr^{(n)}_{i}(T)^{(m)}}),
\]
form which the subsequent results immediately follow.
\begin{proposition}\label{pn-prodwh}
Let $(X_1, D_1), \ldots, (X_n, D_n)$ be extended IFS. Then the following two statements hold:
\begin{enumerate}

\item\label{pn-prodwh-1}
 If $(X_{i}, D_{i})$ is weakly hyperbolic, for $1 \le i \le n$, so is $(\bigtimes_{i=1}^{n} X_{i}, D^{\times})$.

\item\label{pn-prodwh-2}
 If $(X_{i}, D_{i})$ is both, compact and weakly hyperbolic, for $1 \le i \le n$, then
 \[
 \val{T} = (\val{\Pr^{(n)}_{1}(T)}, \ldots, \val{\Pr^{(n)}_{n}(T)}).
 \]

\end{enumerate}
\end{proposition}

Let us sum up what we have obtained so far in this section.

\begin{theorem}\label{thm-prodtd}
Let $(X_1, D_1), \ldots, (X_n, D_n)$ be topological digit spaces. Then $(\bigtimes_{i=1}^{n} X_{i}, D^{\times})$ is a topological digit space as well.
\end{theorem}

For the remainder of this section we will restrict our investigations to the special case of digit spaces.
\begin{proposition}\label{pn-digprod}
For $1 \le i \le n$, let $(X_{i}, D_{i})$ be a digit space with metric $\rho_{i}$. Then the following statements hold:
\begin{enumerate}

\item\label{pn-digprod-1}
$(\bigtimes_{i=1}^{n} X_{i}, D^{\times})$ is a digit space with respect to the maximum metric $\rho$.

\item\label{pn-digprod-2}
If for each $1 \le i \le n$, $Q_{i}$ is a countable dense subset of $X_{i}$ so that $(X_{i}, \rho_{i}, Q_{i})$ is computable, then $\bigtimes_{i=1}^{n} Q_{i}$ is a countable dense subset of $\bigtimes_{i=1}^{n} X_{i}$ with $(\bigtimes_{i=1}^{n} X_{i}, \rho, \bigtimes_{i=1}^{n} Q_{i})$ being computable as well.

\item\label{pn-digprod-3}
If for every $1 \le i \le n$, $(X_{i}, D_{i}, Q_{i})$ is a computable digit space, then so is $(\bigtimes_{i=1}^{n} X_{i}, D^{\times},\linebreak \bigtimes_{i=1}^{n} Q_{i})$.

\end{enumerate}
\end{proposition}
\begin{proof}
(\ref{pn-digprod-1})
As we have already seen, $(\bigtimes_{i=1}^{n} X_{i}, D^{\times})$ is a compact covering extended IFS. It remains to show that all maps in $D^{\times}$ are contractive.

For $1 \le i \le n$, let $q_{i}$ be the contraction factor of $d_{i}$, and be $q$ their maximum.
For $1 \le j \le s_{D}$, let $\vec{x}^{(j)}, \vec{y}^{(j)} \in \bigtimes_{i=1}^{n} X_{i}$ with $\vec{x}^{(j)} = (x^{(j)}_{1}, \ldots, x^{(j)}_{n})$ and similarly for $\vec{y}^{(j)}$. Then
\begin{align*}
&\rho(\pair{d_{1}, \ldots, d_{n}}(\vec{x}^{(1)}, \ldots, \vec{x}^{(s_{D})}), \pair{d_{1}, \ldots, d_{n}}(\vec{y}^{(1)}, \ldots, \vec{y}^{(s_{D})})) \\
&\quad= \rho((d_{1}(x^{(1)}_{1}, \ldots, x^{(s_{D})}_{1}), \ldots, d_{n}(x^{(1)}_{n}, \ldots, x^{(s_{D})}_{n})), \\
&\hspace{4cm} (d_{1}(y^{(1)}_{1}, \ldots, y^{(s_{D})}_{1}), \ldots, d_{n}(y^{(1)}_{n}, \ldots, y^{(s_{D})}_{n}))) \\
&\quad= \max_{i=1}^{n} \rho_{i}(d_{i}(x^{(1)}_{i}, \ldots, x^{(s_{D})}_{i}), d_{i}(y^{(1)}_{i}, \ldots, y^{(s_{D})}_{i})) \\
&\quad \le q_{i} \cdot \max_{i=1}^{n} \rho((x^{(1)}_{i}, \ldots, x^{(s_{D})}_{i}), (y^{(1)}_{i}, \ldots, y^{(s_{D})}_{i})) \\
&\quad \le q \cdot \max_{i=1}^{n} \max_{j=1}^{s_{D}} \rho_{j}(x^{(j)}_{i}, y^{(j)}_{i}) \\
&\quad= q \cdot \max_{j=1}^{s_{D}} \max_{i=1}^{n} \rho_{j}(x^{(j)}_{i}, y^{(j)}_{i}) \\
&\quad= q \cdot \max_{j=1}^{s_{D}} \rho(\vec{x}^{(j)}, \vec{y}^{(j)}) \\
&\quad= q \cdot \rho((\vec{x}^{(1)}, \ldots, \vec{x}^{(s_{D})}), (\vec{y}^{(1)}, \ldots, \vec{y}^{(s_{D})})).
\end{align*}

(\ref{pn-digprod-2})
As is well known, $\bigtimes_{i=1}^{n} Q_{i}$ is a countable dense subset of $\bigtimes_{i=1}^{n}$, and obviously 
$(\bigtimes_{i=1}^{n} X_{i},\linebreak \rho, \bigtimes_{i=1}^{n} Q_{i})$ is a computable metric space.

(\ref{pn-digprod-3})
It remains to verify that all digits in $D^{\times}$ are computable. For $1 \le i \le n$, let $k_{i}$ be a computable modulus of continuity for digit $d_{i}$ and set $k(\theta) := \min \set{k_{i}(\theta)}{1 \le i \le n}$, where $\theta \in \QQ_{+}$. Moreover, for $1 \le j \le s_{D}$, let $\vec{u}^{(j)}, \vec{v}^{(j)} \in \bigtimes_{i=1}^{n} Q_{i}$ with $\vec{u}^{(j)} = (u^{(j)}_{1}, \ldots, u^{(j)}_{n})$ and similarly for $\vec{v}^{(j)}$. If we assume that 
\[
\rho((\vec{u}^{(1)}, \ldots, \vec{u}^{(s_{D})}), (\vec{v}^{(1)}, \ldots, \vec{v}^{(s_{D})})) \le k(\theta)
\]
then we obtain for all $1 \le i \le n$ that 
\[
\rho((u^{(1)}_{i}, \ldots, u^{(s_{D})}_{i}), (v^{(1)}_{i}, \ldots, v^{(s_{D})}_{i})) \le k_{i}(\theta).
\]
Hence, 
\[
\rho_{i}(d_{i}(u^{(1)}_{i}, \ldots, u^{(s_{D})}_{i}), d_{i}(v^{(1)}_{i}, \ldots, v^{(s_{D})}_{i})) \le \theta,
\]
for all $1 \le i \le n$. Thus,
\[
\rho(\pair{d_{1}, \ldots, d_{n}}(\vec{u}^{(1)}, \ldots, \vec{u}^{(s_{D})}), \pair{d_{1}, \ldots, d_{n}}(\vec{v}^{(1)}, \ldots, \vec{v}^{(s_{D})})) \le \theta.
\]

By our assumption there is a procedure $G^{(i)}_{d}$, for each $1 \le i \le n$ and $d \in D_{i}$, such that for any given $\vec{u} \in Q_{i}^{s_{D}}$ and $m \in \NN$, $G^{(i)}_{d}$ computes a basic element $v \in Q_{i}$ with $\rho_{i}(d(\vec u), v) < 2^{-m}$. 

Now, let $\pair{d_{1}, \ldots, d_{n}} \in D^{\times}$ and let $G_{\pair{d_{1}, \ldots, d_{n}}}$ be the procedure that operates as follows:
\begin{quote}
Given $m \in \NN$ and $\vec{u}^{(1)}, \ldots, \vec{u}^{(s_{d})} \in \bigtimes_{j=1}^{n} Q_{j}$ with $\vec{u}^{(j)} = (u^{(j)}_{1}, \ldots, u^{(s_{D})}_{n})$: For each $1 \le i \le n$, apply $G^{(i)}_{d_{i}}$ to $(u^{(1)}_{i}, \ldots, u^{(s_{D})}_{i})$ and $m$, and let $v_{i}$ be the corresponding output.
Then output $(v_{1}, \ldots, v_{n})$.
\end{quote}
It follows that
\begin{align*}
&\rho(\pair{d_{1}, \ldots, d_{n}}(\vec{u}^{(1)}, \ldots, \vec{u}^{(s_{D})}), (v_{1}, \ldots. v_{n})) \\
&\quad= \rho((d_{1}(u^{(1)}_{1}, \ldots, u^{(s_{D})}_{1}), \ldots, d_{n}(u^{(1)}_{n}, \ldots, u^{(s_{D)}}_{n})), (v_{1}, \ldots, v_{n})) \\
&\quad= \max_{i=1}^{n} \rho_{i}(d_{i}(u^{(1)}_{i}, \ldots, u^{(s_{D})}_{i}), v_{i}) \\
&\quad < 2^{-m}. \qedhere
\end{align*}
\end{proof}

\begin{proposition}\label{pn-decprod}
Let $(X_{1}, D_{1}, Q_{1}), \ldots, (X_{n}, D_{n}, Q_{n})$ be decidable computable digit spaces. Then $(\bigtimes_{i=1}^{n} X_{i}, D^{\times}, \bigtimes_{i=1}^{n} Q_{i})$ is decidable as well. 
\end{proposition}
\begin{proof}
Note that for $\vec{u} \in \bigtimes_{i=1}^{n} Q_{i}$ with $\vec{u} = (u_{1}, \ldots, u_{n})$ and $\theta \in \QQ_{+}$, 
\[
\ball{\rho}{\vec{u}}{\theta} = \bigtimes_{i=1}^{n} \ball{\rho_{i}}{u_{i}}{\theta}.
\]
With (\ref{eq-prod-2}) we therefore obtain that for $\pair{d_{1}, \ldots, d_{n}} \in D^{\times}$,
\[
\ball{\rho}{\vec{u}}{\theta} \subseteq \range(\pair{d_{1}, \ldots, d_{n}}) 
\Leftrightarrow (\forall 1 \le i \le n)\, \ball{\rho_{i}}{u_{i}}{\theta} \subseteq \range(d_{i}),
\]
from which the statement follows.
\end{proof}

\begin{proposition}\label{pn-appchprod}
Let $(X_{1}, D_{1}, Q_{1}), \ldots, (X_{n}, D_{n}, Q_{n})$ be computable digit spaces having approximable choice. Then also $(\bigtimes_{i=1}^{n} X_{i}, D^{\times}, \bigtimes_{i=1}^{n} Q_{i})$ has approximable choice.
\end{proposition}
\begin{proof}
For $1 \le i \le n$ and $d \in D_{i}$, let $\fun{\lambda (\theta, u).\, ^{d}v^\theta_u}{\QQ_+ \times \int(\range(d)) \cap Q_{i}}{Q_{i}^{s_{D}}}$ be the effective procedure existing for $d$, as $(X_{i}, D_{i}, Q_{i})$ has approximable choice. Then for $\pair{d_{1}, \ldots, d_{n}} \in D^{\times}$, define the procedure 
\[
\fun{\lambda (\theta, (u_{1}, \ldots, u_{n})).\, v^{\theta}_{(u_{1}, \ldots, u_{n})}}{\QQ_{+} \times \int(\range(\pair{d_{1}, \ldots, d_{n}})) \cap \bigtimes_{j=1}^{n} Q_{j}}{(\bigtimes_{j=1}^{n} Q_{j})^{s_{D}}}
\]
 as follows:

For $1 \le i \le n$, let $^{d_{i}}v^{\theta}_{u_{i}} = (v^{(1)}_{i}, \ldots, v^{(s_{D})}_{i})$. Then set
\[
v^{\theta}_{(u_{1}, \ldots, u_{n})} := ((v^{(1)}_{1}, \ldots, v^{(1)}_{n}), \ldots, (v^{(s_{D})}_{1}, \ldots, v^{(s_{D})}_{n})).
\]
It remains to verify the conditions in Definition~\ref{dn-appchoi}.

(\ref{dn-appchoi-1}) By our assumption we have for $1 \le i \le n$ and $\tilde{\theta} \in \QQ_{+}$ that $\rho_{i}({}^{d_{i}}v^{\theta}_{u_{i}}, {}^{d_{i}}v^{\tilde{\theta}}_{u_{i}}) < \max \{ \theta, \tilde{\theta} \}$. Therefore,
\[
\rho(v^{\theta}_{(u_{1}, \ldots, u_{n})}, v^{\tilde{\theta}}_{(u_{1}, \ldots, u_{n})})
= \max_{j=1}^{s_{D}} \max_{i=1}^{n} \rho_{i}(v^{(j)}_{i}, \tilde{v}^{(j)}_{i}) < \max \{ \theta, \tilde{\theta} \},
\]
where $(\tilde{v}^{(1)}_{i}, \ldots, \tilde{v}^{(s_{D})}_{i}) :=  {}^{d_{i}}v^{\tilde{\theta}}_{u_{i}}$.

(\ref{dn-appchoi-2})\sloppy Let $(u_{1}, \ldots, u_{n}), (u'_{1}, \ldots, u'_{n}) \in \int(\range(\pair{d_{1}, \ldots, d_{n}})) \cap \bigtimes_{j=1}^{n} Q_{j}$. As follows from our supposition, some $\theta'_{i} \in \QQ_{+}$  can be computed, for each $1 \le i \le n$, so that, if $\rho_{i}(u_{i}, u'_{i}) < \theta'_{i}$ then $\rho(^{d_{i}}v^{\theta}_{u_{i}}, 
{}^{d_{i}}v^{\theta}_{u'_{i}}) < \theta$. Set $\theta' := \min\set{\theta'_{i}}{1 \le i \le n}$ and assume that $\rho((u_{1}, \ldots, u_{n}), (u'_{1}, \ldots, u'_{n})) < \theta'$. Then we have that $\rho_{i}(u_{i}, u'_{i}) < \theta'_{i}$, for all $1 \le i \le n$, and therefore that
\[
\rho(v^{\theta}_{(u_{1}, \ldots, u_{n})}, v^{\theta}_{(u'_{1}, \ldots, u'_{n})})
= \max_{j=1}^{s_{D}} \max_{i=1}^{n} \rho_{i}(v^{(j)}_{i}, v'^{(j)}_{i})
= \max_{i=1}^{n} \rho(^{d_{i}}v^{\theta}_{u_{i}}, {}^{d_{i}}v^{\theta}_{u'_{i}})
< \theta.
\]
Here $(v'^{(1)}_{i}, \ldots, v'^{(s_{D})}_{i}) := {}^{d_{i}}v^{\theta}_{u'_{i}}$.

(\ref{dn-appchoi-3}) By our assumption, for every $1 \le i \le n$, given $u_{i} \in \int(\range(d_{i})) \cap Q_{i}$, there is some $\vec{y}_{i} \in d_{i}^{-1}[\{ u_{i}\}]$ so that $\rho(\vec{y}_{i}, {}^{d_{i}}v^{\theta}_{u_{i}}) < \theta$.

Let $\vec{y}_{i} =: (y^{(1)}_{i}, \ldots, y^{(s_{D})}_{i})$ and set $\vec{y} := ((y^{(1)}_{1}, \ldots, y^{(1)}_{n}), \ldots, (y^{(s_{D})}_{1}, \ldots, y^{(s_{D})}_{n}))$. Then $\vec{y} \in \pair{d_{1}, \ldots, d_{n}}^{-1}[\{ (u_{1}, \ldots, u_{n}) \}]$ and 
\[
\rho(\vec{y}, v^{\theta}_{(u_{1}, \ldots, u_{n})}) = \max_{i=1}^{n} \rho(\vec{y}_{i}, {}^{d_{i}}v^{\theta}_{u_{i}}) < \theta. \qedhere
\]
\end{proof}

As a consequence of Theorems~\ref{thm-ctoa} and \ref{thm-cacoind} we now obtain that the equivalence between the Cauchy representation used in Type-Two Theory of Effectivity and the tree representation used in the present approach holds in a constructive fashion also in case of products of computable digit spaces.  Recall that for $\vec x \in \bigtimes_{i=1}^{n} X_{i}$,
\[
\bA_{\bigtimes_{i=1}^{n} X_{i}}(\vec x) \Leftrightarrow (\forall n \in \NN) (\exists \vec{u} \in \bigtimes_{i=1}^{n} Q_{i})\, \rho(\vec x, \vec u) < 2^{-n}
\]
and  $\CC_{\bigtimes_{i=1}^{n} X_{i}} \subseteq \bigtimes_{i=1}^{n} X_{i}$ is given by
\begin{equation*}
\begin{split}
\CC_{\bigtimes_{i=1}^{n} X_{i}}(\vec x) \overset{\nu}{=} (\exists \pair{d_1, \ldots, d_r} \in D^{\times}) (\exists \vec{y}_{1}, \ldots, \vec{y}_{s_{D}} \in \bigtimes_{i=1}^{n} X_{i}) \hspace{3cm}\mbox{}\\ 
 \vec{x} = \pair{d_1, \ldots, d_r}(\vec{y}_1, \ldots, \vec{y}_{s_{D}}) \wedge (\forall 1 \le \kappa \le r)\, \CC_{\bigtimes_{i=1}^{n} X_{i}}(\vec{y}_{\kappa}).
\end{split}
\end{equation*}
By the above mentioned two results we obtain the following consequence which permits to extract computable translations between the realisers resulting from the two definitions.

\begin{theorem}\label{thm-prodca}
Let $(X_{1}, D_{1}, Q_{1}), \ldots, (X_{n}, D_{n}, Q_{n})$ be a well-covering and decidable computable digit spaces with approximable choice. Then $\bA_{\bigtimes_{i=1}^{n} X_{i}} = \CC_{\bigtimes_{i=1}^{n} X_{i}}$.
\end{theorem}

\section{The hyperspace of non-empty compact subsets}\label{sec-compact}

In many applications, e.g.\ in mathematical economics, one needs to deal with multi-valued functions that when considered as set-valued maps have compact sets as values. In this section we will study how spaces of compact sets fit into the framework developed so far. 

For a topological space $X$ with topology $\tau$, let $\KKK(X)$ be the set of all non-empty compact subsets of $X$, endowed with the Vietoris $\tau_V$ topology (also called the finite topology), that is the topology generated by the sets 
\[
\lgroup U; V_1, \ldots, V_n \rgroup := \set{K \in \KKK(X)}{K \subseteq U \wedge (\forall 1 \le \kappa \le n) K \cap V_\kappa \not= \emptyset}
\]
with $U, V_1, \ldots, V_n \in \tau$. The following facts are well known~\cite{kt, mi}.

\begin{theorem}\label{thm-inherit}
Let $(X, \tau)$ and $(Y, \eta)$ be topological spaces, and $\fun{f}{X}{Y}$. Then the following statements hold:
\begin{enumerate}
\item\label{thm-inherit-1}
If $Q$ is dense in $(X, \tau)$, then the set $\QQQ_{\KKK(X)}$ of all non-empty finite subsets of $Q$ is dense in $(\KKK(X), \tau_V)$. Moreover, if $Q$ is countable, so is $\QQQ_{\KKK(X)}$.

\item\label{thm-inherit-2}
$(\KKK(X), \tau_V)$ is Hausdorff if, and only if, $(X, \tau)$ is Hausdorff.

\item\label{thm-inherit-3}
$(\KKK(X), \tau_V)$ is a compact Hausdorff space if, and only if, $(X, \tau)$ is a compact Hausdorff space.

\item\label{thm-inherit-4}
If $f$ is continuous, so is the map $K \in \KKK(X) \mapsto f[K]$ with respect to the Vietoris topologies on $\KKK(X)$ and $\KKK(Y)$.

\item\label{thm-inherit-6}
The operation of taking the union of two compact sets is continuous  with respect to the Vietoris topology.

\end{enumerate}
\end{theorem}
In case of the first statement, the elements of $\QQQ_{\KKK(X)}$ will be called \emph{basic subsets}. For compact Hausdorff spaces $X$ and $Y$, and continuous $\fun{f}{X}{Y}$ define $\fun{\KKK(f)}{\KKK(X)}{\KKK(Y)}$ by $\KKK(f)(K) := f[K]$, for $K \in \KKK(X)$. Then $\KKK$ is an endo-functor on the category of compact Hausdorff spaces and continuous maps.

Next, assume that $(X, D)$ is a compact extended IFS. For $r>0$ and pairwise distinct $d_1, \ldots, d_r \in D$ define $\fun{[d_1, \ldots, d_r]}{\bigtimes_{\kappa=1}^r \KKK(X^{\ar{d_\kappa}})}{\KKK(X)}$ by
\[
[d_1, \ldots, d_r](K_1, \ldots, K_r) = \bigcup_{\kappa = 1}^r d_\kappa[K_\kappa].
\]

Since the $d_\kappa$ are continuous, the sets $d_\kappa[K_\kappa]$ are compact and so is their finite union. As follows from the definition, 
\begin{equation*}
\begin{split}
\range([d_1, \ldots, d_r]) = \{\,K \in \KKK(X) \mid (\exists K_1, \ldots, K_r \in \KKK(X))\, K = \bigcup_{\kappa = 1}^r K_\kappa  \wedge \hspace{2cm}\mbox{}\\
 (\forall 1 \le \kappa \le r)\, K_\kappa \in \KKK(\range(d_\kappa)) \,\}
\end{split}
\end{equation*}

Let $m$ be the cardinality of $D$. If $(d_1, \ldots, d_r) \in D^r$ such that the $d_\kappa$ are pairwise distinct, then $r \le m$. Hence, there are at most $\sum_{\kappa = 1}^m m^\kappa$ such maps. Set
\[
\KKK(D) := \set{[d_1, \ldots, d_r]}{\text{$d_1, \ldots, d_r \in D$ pairwise distinct with $r>0$}}.
\]

Note that according to Definition~\ref{dn-eifs}, in an extended IFS $(X, D)$ the maps $d \in D$ have to be of  type $X^{\ar{d}} \to X$. In case of the maps $[d_{1}, \ldots, d_{n}]$ this is only true if $d_{1}, \ldots, d_{n}$ are unary.

\begin{proposition}\label{pn-comifs}
Let $(X, D)$ be a compact IFS. Then $(\KKK(X), \KKK(D))$ is a compact extended IFS. Moreover, the following statements hold:
\begin{enumerate}

\item\label{pn-comifs-1}
If $(X, D)$ is covering, so is $(\KKK(X), \KKK(D))$.

\item\label{pn-comifs-2}
If $(X, D)$ is well-covering, so is $(\KKK(X), \KKK(D))$.

\end{enumerate}
\end{proposition}
\begin{proof}
(\ref{pn-comifs-1}) It suffices to show that
\begin{equation}\label{eq-comifs}
\KKK(X) \subseteq \bigcup \set{\range(\bar{d})}{\bar{d} \in \KKK(D)}.
\end{equation}

Let $K \in \KKK(X)$, $E := \set{d \in D}{K \cap \range(d) \not= \emptyset}$, say $E = \{ e_1, \ldots, e_r \}$, and for $d \in E$, $K_d = d^{-1}[K \cap \range(d)]$.  As compact sets, $K$ and  $d[X]$ are closed, whence $K_d$ is closed as well and thus compact.

Since $(X, D)$ is covering, it follows for $x \in K$ that there is some $d \in E$ with $x \in \range(d) \cap K$. Hence, $x \in \bigcup_{d \in E} d[K_d]$.  This shows that $K = [e_1, \ldots, e_r](K_{e_1}, \ldots, K_{e_r})$, which means that $K$ is contained in the right hand side of (\ref{eq-comifs}).

(\ref{pn-comifs-2}) We need to show that for every $K \in \KKK(X)$ there is some $\bar{d} \in \KKK(D)$ and some $O \in \tau_V$ so that $K \in O \subseteq \range(\bar{d})$.
 
 Let $K \in \KKK(X)$ and $E := \set{d \in D}{\int(\range(d)) \cap K \not= \emptyset}$, say $E = \{ e_1, \ldots, e_r \}$. Then, we have that 
 \[
 K \in [\bigcup_{\kappa = 1}^r \int(\range(e_\kappa)); \int(\range(e_1)), \ldots, \int(\range(e_r))].
 \]
It remains to show that 
\[
[\bigcup_{\kappa = 1}^r \int(\range(e_\kappa)); \int(\range(e_1)), \ldots, \int(\range(e_r))] \subseteq \range([e_1, \ldots, e_r]).
\]
 Let to this end $K' \in [\bigcup_{\kappa = 1}^r \int(\range(e_\kappa)); \int(\range(e_1)), \ldots, \int(
 \range(e_r))]$ and for $1 \le \kappa \le r$, set $K'_\kappa := e^{-1}_\kappa[K' \cap \range(e_\kappa)]$. Then $K'_\kappa \not= \emptyset$. Moreover, since $K' \subseteq \bigcup_{\kappa = 1}^r \int(\range(e_\kappa))$, we have that $K' = [e_1, \ldots, e_r](K'_1, \ldots, K'_r)$. Thus, $K' \in  \range([e_1, \ldots, e_r])$.
\end{proof}

Similar to the product case, when dealing with a compact set $K$ via the extended IFS $(\KKK(X), \KKK(D))$, $K$ is dealt with as an abstract object, not as a set of other objects. We will now study how from a tree representing $K$ we can access elements of $K$. Note that we are assuming that all digits $d \in D$ have arity~1. 

\begin{definition}\label{dn-devt}
The relation \emph{$S \in \TTT^\omega_D$ is a derived tree of $T \in \TTT^\omega_{\KKK(D)}$} (short: $S \in \partial T$) is co-inductively defined by
\begin{equation*}
[d; S_1] \in \partial [[d_1, \ldots, d_r]; T_1, \ldots, T_r] \overset{\nu}{=} 
 (\exists 1 \le \kappa \le r)\, d = d_\kappa \wedge S_1 \in \partial T_\kappa.
\end{equation*}
\end{definition}

\begin{lemma}\label{lem-setof}
Let $(X, D)$ be a compact IFS. Then for all $n \in \NN$ and $T \in \TTT^\omega_{\KKK(D)}$,
\[
f_{T^{(n)}}(X^{(\ar{T^{(n)}})}) = \bigcup\set{f_{S^{(n)}}[X^{\ar{S^{(n)}}}]}{S \in \partial T}
\]
\end{lemma}
\begin{proof}
We use induction on $n$. Assume that $T = [[d_1, \ldots, d_r]; T_1, \ldots, T_r]$. For $n = 0$ we have that $T^{(0)} = [d_1, \ldots, d_r]$ and $\set{S^{(0)}}{S \in \partial T} = \{ d_1, \ldots, d_r \}$. So the statement holds in this case. For $n > 0$ we obtain with the induction hypothesis that
\begin{align*}
f_{T^{(n)}}(X^{(\ar{T^{(n)}})})
&= \bigcup_{\kappa = 1}^r d_\kappa[f_{T^{(n-1)}_{\kappa}}(X^{(\ar{T^{(n-1)}_\kappa})})]\\
&=  \bigcup_{\kappa = 1}^r d_\kappa[\bigcup_{S_\kappa \in \partial T_\kappa} f_{S^{(n-1)}_\kappa}[X^{\ar{S^{(n-1)}_\kappa}}]] \\
&= \bigcup_{\kappa = 1}^r \bigcup_{S_\kappa \in \partial T_\kappa} d_\kappa[f_{S^{(n-1)}_\kappa}[X^{\ar{S^{(n-1)}_\kappa}}]] \\
&= \bigcup_{S \in \partial T} f_{S^{(n)}}[X^{\ar{S^{(n)}}}] \qedhere
\end{align*}
\end{proof}

\begin{theorem}\label{thm-set}
Let $(X, D)$ be a topological digit space with unary digits only, and $(\KKK(X), \KKK(D))$ be weakly hyperbolic. Then, for every $T \in \TTT^\omega_\KKK(D)$,
\[
\val{T} = \set{\val{S}}{S \in \partial T}
\]
\end{theorem}
\begin{proof}
Let $S \in \partial T$. Then it follows with the preceding lemma that for every $n \in \NN$, 
\[
\val{S} \in f_{S^{(n)}}[X^{\ar{S^{(n)}}}] \subseteq f_{T^{(n)}}(X^{(\ar{T^{(n)}})}).
\]
 Moreover, by Corollary~\ref{cor-arconv} we have that $\val{T} = \lim_{n \to \infty} f_{T^{(n)}}(X^{(\ar{T^{(n)}})})$. Now, let $U$ be an open neighbourhood of $\val{T}$ in $X$. Then $\lgroup U; X \rgroup$ is an open neighbourhood of  $\val{T}$ in $\KKK(X)$. Hence, there is some $n \in \NN$ so that $f_{T^{(m)}}(X^{(\ar{T^{(m)}})}) \subseteq U$, for all $m \ge n$, which implies that $\val{S} \in U$. Thus,
 \[
 \val{S} \in \bigcap \set{U \in \tau}{\val{T} \subseteq U}.
 \]
 
Obviously, $\val{T} \subseteq \bigcap \set{U \in \tau}{\val{T} \subseteq U}$. To show the converse inclusion, assume that there is some $x \in \bigcap \set{U \in \tau}{\val{T} \subseteq U} \setminus \val{T}$. Since compact Hausdorff spaces are regular~\cite{wil}, there exist  disjoint open sets $V, W \in \tau$ with $x \in V$ and $\val{T} \subseteq W$. So, $x \notin W$ and thus $x \notin  \bigcap \set{U \in \tau}{\val{T} \subseteq U}$, a contradiction. It follows that $\val{S} \in \val{T}$. 

Note that for each $n \in \NN$, $f_{T^{(n)}}(X^{(\ar{T^{(n)})}}) \in f_{T^{(n)}}[\KKK(X)^{\ar{T^{(n)}}}]$ and therefore
\begin{align*}
f_{T^{(n)}}(X^{(\ar{T^{(n)}})}) 
&\subseteq \bigcup f_{T^{(n)}}[\KKK(X)^{\ar{T^{(n)}}}] \\
&= \bigcup\set{f_{T^{(n)}}(\vec{K})}{\vec{K} \in \KKK(X)^{\ar{T^{(n)}}}} \subseteq f_{T^{(n)}}(X^{(\ar{T^{(n)}})}).
\end{align*}

 Now, conversely, let $x \in \val{T}$. According to the definition, $\{ \val{T} \} \in f_{T^{(n)}}[\KKK(X)^{\ar{T^{(n)}}}]$, for all $n \in \NN$, that is, $\val{T} \in f_{T^{(n)}}(X^{(\ar{T^{(n)}})})$, by what we have just seen. Thus, $x \in f_{S^{(n)}}[X^{\ar{S^{(n)}}}]$, for some $S \in \partial T$. Observe that for each such tree $S$ the sequence $(f_{S^{(\kappa)}}[X^{\ar{S^{(\kappa)}}}])_{\kappa \in \NN}$ is decreasing. 
 
It follows that for every $n \in \NN$ there is a finite $D$-tree $N$ of height $n$ with $x \in f_N[X^{\ar{N}}]$ and $N = S^{(n)}$, for some $S \in \partial T$. With respect to the prefix relation the set of these finite $D$-trees $N$ is a finitely branching infinite tree. By K\"onig's Lemma it must have an infinite path, which in our case is a sequence $(N_n)_{n \in \NN}$ of initial subtrees of derived trees of $T$ such that $N_n \prec N_{n+1}$. In other words, there is some $\hat{S} \in \TTT^{\omega}_{D}$ with $N_n = \hat{S}^{(n)}$, for all $n$. Hence, $x \in \bigcap_{n \in \NN} f_{\hat{S}^{(n)}}[X^{\ar{\hat{S}^{(n)}}}]$, that is $x = \val{\hat{S}}$.

It remains to verify that $\hat{S} \in \partial T$. Set
\[
R := \set{(N, M) \in \TTT^{\omega}_{D} \times T^{\omega}_{\KKK(D)}}{(\forall n \in \NN) (\exists L_{n} \in \partial M)\, N^{(n)} = L_{n}^{(n)}}.
\]
Then $(\hat{S}, T) \in R$. We use co-induction to show that, if $(N, M) \in R$ then $N \in \partial M$. 

Define $\fun{\Theta}{\PPP(T^{\omega}_{D} \times T^{\omega}_{\KKK(D)})}{\PPP(T^{\omega}_{D} \times T^{\omega}_{\KKK(D)})}$ by
\[
\Theta(W) := \set{([d; N'], [[d_{1}, \ldots, d_{r}]; M_{1}, \ldots, M_{r}])}{(\exists 1 \le \kappa \le r)\, d = d_{
\kappa} \wedge (N', T_{\kappa}) \in W}.
\]
Then $N \in \partial M$, exactly if $(N, M) \in \nu \Theta$.
We need to prove that $R \subseteq \Theta(R)$.

Let $(N, M) \in R$ with $N = [d; N']$ and $M = [[d_{1}, \ldots, d_{r}]; T_{1}, \ldots, T_{r}]$. Then there are $L_{n} \in \partial M$ with $N^{(n)} = L^{(n)}_{n}$, for all $n \in \NN$. It follows that for all $n$, $\rt{L_{n}} = \rt{N} = d$. Since $L_{n} \in \partial M$ and the $d_{1}, \ldots, d_{r}$ are pairwise distinct, there is exactly one $\kappa \in \{ 1, \ldots, r \}$ with $d = d_{\kappa}$ and $\sbt{L_{n}} \in \partial M_{\kappa}$, for all $n$. Let $L'_{n} = \sbt{L_{n+1}}$. Then $L'_{n} \in \partial M_{\kappa}$ and 
\[
L'^{(n)}_{n} = \sbt{L_{n+1}}^{(n)} = N'^{(n)},
\]
where the last equation holds as $L^{(n+1)}_{n+1} = N^{(n+1)}$. This shows that $(N', M_{\kappa}) \in R$. Hence, $(N, M) \in \Theta(R)$.
\end{proof}

Besides covering and well-covering, there are further central properties of which it will be important to know if they are also inherited from $(X, D)$ to $(\KKK(X), \KKK(D))$. Unfortunately, in case of weak hyperbolicity, this has to remain an open problem in the general case. 

If, however, $X$ is a metric space and all digit maps are contracting, then $\KKK(X)$ is also a metric space with respect to the \emph{Hausdorff metric}, $\hdm$, given by
\[
\hdm(K, K') := \inf \set{\varepsilon \ge 0}{K \subseteq \ball{\rho}{K'}{\varepsilon} \wedge K' \subseteq \ball{\rho}{K}{\varepsilon}}
\]
with $\ball{\rho}{K}{\varepsilon} := \set{x \in X}{(\exists y \in K)\,\rho(x, y) < \varepsilon}$, and the maps in $\KKK(D)$ are contracting. The contraction factor even remains the same. The topology induced by the Hausdorff metric is equivalent to the Vietoris topology. So, in this case, which is also the prevalent case in applications, both IFS are weakly hyperbolic. As we know from Theorem~\ref{thm-valcont}, the assumption that $X$ comes equipped with a metric from which the topology on $X$ is induced, is not a restriction.

Finally in this section, let us study how computability properties are inherited to the hyperspace. 

\begin{lemma}\label{lem-comphyp}
Let $(X,Q)$ be a computable metric space. Then $(\KKK(X), \QQQ_{\KKK(X)})$ is computable as well.
\end{lemma}
\begin{proof}
Since $(X, Q)$ is computable, there is a program $P$ which on input $(u, v, r) \in Q \times Q \times \QQ$ terminates its computation after finitely many steps and outputs ACCEPT, if $\rho(u, v) < r$. Otherwise, the computation does not terminate. Let $P_K$ be the following program:

Given $(U, V, r) \in \QQQ_{\KKK(X)} \times \QQQ_{\KKK(X)} \times \QQ$, say with $U = \{ u_1, \ldots, u_m \}$ and $V = \{ v_1, \ldots v_n \}$, proceed as follows
\begin{enumerate}
\item[(0)] Set $i:= 1$ and go to (1).

\item[(1)] Simultaneously start the computation of $P$ on inputs $(u_i, v_1, r)$, \ldots, $(u_i, v_n, r)$. If one of these computations halts after finitely many steps with ACCEPT, halt. Then if $i < m$, increase $i$ by $1$ and go to (1); otherwise output ACCEPT.

\end{enumerate}

It follows that  if $\hdm(U, V) < r$, then $P_K(U, V, r)$ halts after finitely many steps with output ACCEPT. Otherwise the computation does not terminate. The second requirement in Definition~\ref{dn-compmet} can be dealt with analogously.
\end{proof}

\begin{proposition}\label{pn-comphypdig}
Let $(X, D, Q)$ be a computable digit space with only unary digits. Then $(\KKK(X), \KKK(D), \QQQ_{\KKK(X)})$ is a computable digit space as well. 
\end{proposition}
\begin{proof}
Let $[d_1, \ldots, d_r] \in \KKK(D)$. As we have seen, the domain of this map is the space $\KKK(X)^r$, which has $r$-tuples 
$(U_1, \ldots, U_r)$ as basic elements with $U_\kappa$ a basic subset of $\KKK(X)$.

Every digit $d_\kappa$ ($1 \le \kappa \le r$) has a computable modulus of continuity $k_\kappa$. For $\theta \in \QQ_+$, let $k(\theta) := \min\set{k_\kappa(\theta)}{1 \le \kappa \le r}$. Moreover, let $\vec{U}, \vec{V}$ be basic elements of the domain, say $\vec{U} =  (U_1, \ldots, U_r)$ and similarly for $\vec V$, such that $\hdm(U_\kappa, V_\kappa) < k(\theta)$. Then we have for each $\kappa$ and every $u \in U_\kappa$ that there is some $v \in V_\kappa$ with $\rho(u, v) < k(\theta)$ and hence $\rho(d_\kappa(u), d_\kappa(v)) < \theta$. Conversely, there is some $u \in U_\kappa$ with $\rho(u, v) < k(\theta)$, for each $v \in V_\kappa$. Again it follows that $\rho(d_\kappa(u), d_\kappa(v)) < \theta$. This shows that $\hdm([d_1, \ldots, d_r](\vec{U}), [d_1, \ldots, d_r](\vec{V})) < \theta$.

Next, we have to show that there is a procedure $G_{\bar{d}}$ with $\bar{d} = [d_1, \ldots, d_r]$,  which given a basic element $ \vec U$ of the domain of $\bar{d}$ and $n \in \NN$ computes some $V \in \KKK(Q)$ with $\hdm([d_1, \ldots, d_r](\vec{U}), V) < 2^{-n}$. By assumption, for every $1 \le \kappa \le r$,  there are procedures $G_{d_\kappa}$ that given $u \in Q$ and $n \in \NN$ compute a basic element $v \in Q$ so that $\rho(d_\kappa(u), v) < 2^{-n}$. $G_{\bar{d}}$ proceeds as follows:

Given $n \in \NN$ and $\vec{U} = (U^{(1)}, \ldots, U^{(r)})$ with $U^{(\kappa)}$ a basic subset of $\KKK(X)$, say $U^{(\kappa)} = \{ u^{(\kappa)}_1, \ldots, u^{(\kappa)}_{m_\kappa} \}$, do:
For each $1 \le \kappa \le r$, and every $1 \le \iota \le m_\kappa$, use $G_{d_\kappa}$ to compute $v^{(\kappa)}_\iota \in Q$ such that $\rho(d_\kappa(u^{(\kappa)}_\iota), v^{(\kappa)}_\iota) < 2^{-n}$. Then output the set $V$ of all $v^{(\kappa)}_\iota$ with $1 \le \iota \le m_\kappa$ and $1 \le \kappa \le r$.

As is readily verified, $\hdm([d_1, \ldots, d_r](\vec{U}), V) < 2^{-n}$.
\end{proof}

Inheritability  of decidability is an immediate consequence of the subsequent technical lemma.

\begin{lemma}\label{lem-epsballk}
Let $(X, D)$ be a digit space. Then for any finite subset $U \subseteq X$, $[d_1, \ldots, d_r] \in \KKK(D)$, and $\theta \in \QQ_+$, $\ball{\hdm}{U}{\theta} \subseteq \range([d_1, \ldots, d_r])$ if, and only if, 
\begin{enumerate}
\item\label{lem-epsballk-1}
$(\forall u \in U) (\exists 1 \le \kappa \le r)\, \ball{\rho}{u}{\theta} \subseteq \range(d_\kappa)$ and

\item\label{lem-epsballk-2}
$(\forall 1 \le \kappa \le r) (\exists u \in U)\, \ball{\rho}{u}{\theta} \subseteq \range(d_\kappa)$.

\end{enumerate}
\end{lemma}
\begin{proof}Let $U = \{ u_1, \ldots, u_n \}$.
For the `if' part  assume that $K \in \ball{\hdm}{U}{\theta}$. Then there is some rational number $\theta' < \theta$ with $\hdm(U, K) \le \theta'$. For $1 \le \sigma \le n$ let $K_\sigma := \set{y \in K}{\rho(u_\sigma, y) \le \theta'}$. Then $K_\sigma$ is non-empty and compact. Moreover, $K = \bigcup_{\sigma = 1}^n K_\sigma$. By our assumption there is some $1 \le \kappa_\sigma \le r$ so that $\ball{\rho}{u_\sigma}{\theta} \subseteq \range(d_{\kappa_\sigma})$. Hence, there is some non-empty compact $K_\sigma' \subseteq X^{\ar{d_{\kappa_{\sigma}}}}$ with $K_\sigma = d_{\kappa_\sigma}[K'_\sigma]$. Because of our second assumption, every $1 \le \kappa \le r$ is among the $\kappa_\sigma$ with this property. Thus, $K \in \range([d_1, \ldots, d_r])$.

In case of the `only-if' part, we show the contrapositive. First, assume that there is some $1 \le \kappa_0 \le r$ so that for every $1 \le \sigma \le n$, $\ball{\rho}{u_\sigma}{\theta} \not\subseteq \range(d_{\kappa_0})$. Let $y_\sigma \in \ball{\rho}{u_\sigma}{\theta} \setminus \range(d_{\kappa_0})$ and set $K := \set{y_\sigma}{1 \le \sigma \le n}$. Then $K \in \ball{\hdm}{U}{\theta}$, but $K \not\in \range([d_1, \ldots, d_r])$.

Next, assume that there is some $1 \le \sigma_0 \le n$ such that for all $1 \le \kappa \le r$, $\ball{\rho}{u_{\sigma_0}}{\theta} \not\subseteq \range(d_\kappa)$. Let $z_\kappa \in \ball{\rho}{u_{\sigma_0}}{\theta} \setminus \range(d_\kappa)$, and for $K \in \ball{\hdm}{U}{\theta}$, set  $$K' := K \cup  \set{z_\kappa}{1 \le \kappa \le r}.$$ Then $K' \in \ball{\hdm}{U}{\theta}$, but $K' \not\in \range([d_1, \ldots, d_r])$.
\end{proof}

\begin{proposition}\label{pn-dechypdig}
Let $(X, D, Q)$ be a computable digit space with only unary digits. If $(X, D, Q)$ is decidable, so is $(\KKK(X), \KKK(D), \QQQ_{\KKK(X)})$.
\end{proposition}

We have seen so far that $(\KKK(X), \KKK(D), \QQQ_{\KKK(X)})$ is a well-covering, decidable computable digit space, if also $(X, D, Q)$ is a space of this type and all $d \in D$ are unary.

\begin{proposition}\label{pn-wachyp}
Let $(X, D, Q)$ be a well-covering and decidable computable digit space with only unary digits. If $(X, D, Q)$ has approximable choice, then also $(\KKK(X), \KKK(D), \QQQ_{\KKK(X)})$ has approximable choice.
\end{proposition}
\begin{proof}
Let $\bar{d} \in \KKK(D)$ with $\bar{d} = [d_1, \ldots, d_r]$ and $U \in \QQQ_{\KKK(X)} \cap \int(\range(\bar{d}))$ with $U = \{ u_1, \ldots, u_n \}$. Then there is some $m \in \NN$ so that $\ball{\hdm}{U}{2^{-m}} \subseteq \range(\bar{d})$. As we have just seen, $(\KKK(X), \KKK(D),\linebreak \QQQ_{\KKK(X)})$ is decidable. Therefore, a number $m$ with this property can effectively be found. For $1 \le \kappa \le r$ set
\[
U_\kappa := \set{u \in U}{\ball{\rho}{u}{2^{-m-1}} \subseteq \range(d_\kappa)}.
\]
 Then $U_\kappa$ is non-empty, by Lemma~\ref{lem-epsballk}(\ref{lem-epsballk-2}), and $U = \bigcup_{\kappa=1}^r U_\kappa$, by Lemma~\ref{lem-epsballk}(\ref{lem-epsballk-1}). By using the decidability of $(X, D, Q)$, $U_\kappa$ can be computed from $U$ and $d_\kappa$.

Now, for $1 \le \kappa \le r$, use approximable choice to pick the function $\lambda (\theta, u).\, v^\theta_u$, and for $\theta \in \QQ_+$ set
\[
V^\theta_{U_\kappa} := \set{v^\theta_u}{u \in U_\kappa},
\]
and $V^\theta_U := (V^\theta_{U_1}, \ldots, V^\theta_{U_r})$. We have to verify Conditions~\ref{dn-appchoi}(\ref{dn-appchoi-1})-(\ref{dn-appchoi-3}).

(\ref{dn-appchoi-1})
For $\theta' \in \QQ_+$, we have to show that $\hdm(V^\theta_U, V^{\theta'}_U) < \max \{ \theta, \theta' \}$, that is we need to verify that for $1 \le \kappa \le r$, $\hdm(V^\theta_\kappa, V^{\theta'}_\kappa) < \max \{ \theta, \theta' \}$, which is obvious, as for all $u \in U$, $\rho(v^\theta_u, v^{\theta'}_u) < \max \{ \theta, \theta' \}$.

(\ref{dn-appchoi-2})
For $1 \le \kappa \le r$, pick $\theta'_\kappa$ according to approximable choice, part~(\ref{dn-appchoi-2}) and  let $\theta': = \min \{ \theta'_1, \ldots, \theta'_r, 2^{-m-3} \}$. 

For $U' \in \QQQ_{\KKK(X)} \cap \int(\range(\bar{d}))$ assume that $\hdm(U, U') < \theta'$ and let $u \in U$. Then there is some $u' \in U'$ with $\rho(u, u') < \theta'$. Moreover, by Lemma~\ref{lem-epsballk}(\ref{lem-epsballk-1}), there is some $1 \le \kappa \le r$ so that $\ball{\rho}{u}{2^{-m}} \subseteq \range(d_\kappa)$ and hence also $\ball{\rho}{u}{2^{-m-1}} \subseteq \range(d_\kappa)$, that is $u \in U_\kappa$. Since $\theta' \le 2^{-m-3}$, it follows that $\ball{\rho}{u'}{2^{-m-1}} \subseteq \ball{\rho}{u}{2^{-m}} \subseteq \range(d_\kappa)$, that is $u' \in U'_\kappa$. Consequently, $\rho(v^\theta_u, v^\theta_{u'}) < \theta$. By symmetry we obtain that $\hdm(V^\theta_{U_\kappa}, V^\theta_{U'_\kappa}) < \theta$ and thus $\hdm(V^\theta_U, V^\theta_{U'}) < \theta$.

(\ref{dn-appchoi-3})
By approximable choice, part~(\ref{dn-appchoi-3}), there is some $z^\theta_u \in d^{-1}_\kappa[\{ u \}]$ with $\rho(z^\theta_u, v^\theta_u) < \theta$, for each $u \in U_\kappa$ and $1 \le \kappa \le r$. Set $Z^\theta_{U_\kappa} := \set{z^\theta_u}{u \in U_\kappa}$. Then we have that $\hdm(Z^\theta_{U_\kappa}, V^\theta_{U_\kappa}) < \theta$, and hence for $Z^\theta_U := (Z^\theta_{U_1}, \ldots, Z^\theta_{U_r})$ that $\hdm(Z^\theta_U, V^\theta_U) < \theta$.
\end{proof}

Just as in the product case, we now obtain that the equivalence between the Cauchy representation used in Type-Two Theory of Effectivity and the tree representation used in the present approach holds in a constructive fashion also in case of the space of all non-empty compact subsets of a computable digit space. Note again that we need to assume that all maps in $D$ are unary. Remember that for $K \in \KKK(X)$,
\[
\bA_{\KKK(X)}(K) \Leftrightarrow (\forall n \in \NN) (\exists U \in \QQQ_{\KKK(X)})\, \hdm(K, U) < 2^{-n}
\]
and  $\CC_{\KKK(X)} \subseteq \KKK(X)$ is given by
\begin{equation*}
\begin{split}
\CC_{\KKK(X)}(K) \overset{\nu}{=} (\exists [d_1, \ldots, d_r] \in \KKK(D)) (\exists K_1, \ldots, K_r \in \KKK(X)) \hspace{3cm}\mbox{}\\ 
 K = [d_1, \ldots, d_r](K_1, \ldots, K_r) \wedge (\forall 1 \le \kappa \le r)\, \CC_{\KKK(X)}(K_\kappa).
\end{split}
\end{equation*}
The following result is again a consequence of Theorems~\ref{thm-ctoa} and \ref{thm-cacoind}. It allows to extract computable translations between the realisers resulting from the two definitions.

\begin{theorem}\label{thm-compeqca}
Let $(X, D, Q)$ be a well-covering and decidable computable digit space with approximable choice and only unary digits. Then $\bA_{\KKK(X)} = \CC_{\KKK(X)}$.
\end{theorem}

\section{Uniformly continuous functions}\label{sec-ufun}

In \cite{be}, working in a constructive setting, Berger presented a co-inductive inductive characterisation of the uniformly continuous functions $\fun{f}{[-1, 1]^m}{[-1, 1]}$ and showed how it can be used to extract tree representations of such functions. Here, we will lift the characterisation to the general framework of computable digit spaces. 

Assume that $(X, D)$ is an extended IFS. For $\fun{f}{X^n}{X}$, $m \in \NN_{0}$ and $1 \le i \le m$, define $\fun{f^{(i, m)}}{X^{n+m-1}}{X^{m}}$ by
\begin{equation*}
f^{(i, m)}(x_1, \ldots, x_{i-1}, y_1, \ldots, y_n, x_{i+1}, \ldots, x_m) := 
(x_1, \ldots, x_{i-1}, f(y_1, \ldots, y_n), x_{i+1}, \ldots, x_m).
\end{equation*}
If $f$ has a right inverse $f'$, set
\[
 f'^{(i, m)}(x_1, \ldots, x_m) := (x_1, \ldots, x_{i-1}, f'(x_i), x_{i+1}, \ldots, x_m).
\]
Then $f'^{(i, m)}$ is a right inverse of $f^{(i, m)}$. 

Now, assume $(Y, E)$ to be a further extended IFS, define
$\FF(X, Y) := \set{\fun{f}{X^m}{Y}}{m \ge 0}$  to be the set of all maps from $X$ to $Y$ of some arity $m \ge 0$, and for  $m \in \NN$ and $F \subseteq \FF(X, Y)$, set
\[
F^{(m)} := \set{f \in \FF(X, Y)}{\ar{f} = m}.
\]
Moreover, let the operator $\Phi ^{X, Y}\colon \PPP(\FF(X, Y)) \to (\PPP(\FF(X, Y)) \to \PPP(\FF(X, Y)))$ be given by
\begin{equation}\label{eq-phiop}
\begin{split}
\Phi^{X,Y}(F)(G) := \{\, f \mid (\exists e \in E) (\exists f_1, \ldots, f_{\ar{e}} \in F)\, f = e \circ (f_1 \times \cdots \times f_{\ar{e}})  \vee  \hspace{1.5cm} \mbox{} \\
 (\exists 1 \le i \le \ar{f}) (\forall d \in D)\, f \circ d^{(i, \ar{f})} \in G \,\}.
\end{split}
\end{equation}
  
Obviously, $\Phi^{X,Y}(F)$ is monotone in $G$, for all $F \subseteq \FF(X, Y)$. Thus, $\JJJ^{X,Y}(F) := \mu \Phi^{X,Y}(F)$ exists. It follows that $\JJJ^{X,Y}(F)$ is the smallest subset $G$ of $\FF(X, Y)$ such that
\begin{itemize}
\item[(W)]\label{itm-w}
If $e \in E$ and $f_1, \ldots, f_{\ar{e}} \in F$, then $e \circ (f_1 \times \cdots \times f_{\ar{e}}) \in G$.

\item[(R)]\label{itm-r}
If $h \in \FF(X, Y)$ and $1 \le i \le \ar{h}$ so that $h \circ d^{(i, \ar{h})} \in G$, for all $d \in D$, then $h \in G$.

\end{itemize}
Since $\JJJ^{X,Y}$ is monotone as well, also $\CC_{\FF(X, Y)} := \nu \JJJ^{X,Y}$ exists. Observe that if $h \in \CC_{\FF(X, Y)}$, then $\range(h) \subseteq \range(e)$, for some $e \in E$, since 
\[
\CC_{\FF(X, Y)} = \mu G. \Phi^{X,Y}(\CC_{\FF(X, Y)})(G) = \Phi^{X,Y}(\CC_{\FF(X, Y)})(\CC_{\FF(X, Y)}).
\]
Note that we will suppress the superscripts in case the spaces $X$, $Y$ are clear from the context.

For $i \in \NN$ let  $R_i$ be a letter with $\ar{R_i} := \card{D}$. Then a realiser for a map $h \in \CC_{\FF(X, Y)}$ extracted from the definition of $\CC_{\FF(X, Y)}$ by following the exposition in Section~\ref{sec-progex} is an infinite $(E_{n})_{n \in \NN}$-tree, for some family $(E_{n})_{n \in \NN}$ of finite subsets of $E \cup \set{R_i}{i \in \NN}$,  such that
\begin{itemize}
\item each node is either a

\begin{description}
\item[\hspace{.5em} \it writing node] labelled with a digit $e \in E$ and $\ar{e}$ immediate subtrees, or a

\item[\hspace{.5em} \it reading node] labelled with $R_i$ and $\card{D}$ immediate subtrees;

\end{description}

\item each path has infinitely many writing nodes.

\end{itemize}

Writing nodes correspond to (inverted) Rule~(W), respectively the left part in the disjunction defining $\Phi^{X, Y}(F)(G)$. The label $e \in E$ realises the existential quantifier. Reading notes correspond to (inverted) Rule~(R). The index $i$ of $R_{i}$ realises the existential quantifier in the right part of the disjunction. Digit $d \in D$ is the root of the tree realising the $i$-th argument. The condition that every path has infinitely many writing notes reflects the  co-inductive nature of the definition of $\CC_{\FF(X, Y)}$, that is the fact that the greatest fixed point of $\Phi^{X, Y}$ is taken with respect to the first argument $F$. It particularly implies that between two writing nodes on a path there may only be finitely many reading notes, which reflects the inductive part in the definition of $\CC_{\FF(X, Y)}$, that is the fact that the least fixed point of $\Phi^{X, Y}$ is taken with respect to the second argument $G$.  

The interpretation of such a tree as a tree transformer is easy. Given $\ar{h}$ trees $T_1$, \ldots, $T_{\ar{h}} \in \TTT^\omega_D$ as inputs, run through the tree and output a tree in $\TTT^\omega_E$ as follows:
\begin{enumerate}
\item At a writing node $[e; S_1, \ldots, S_{\ar{e}}]$ output $e$ and continue with the subtrees $S_1, \ldots, S_{\ar{e}}$.

\item At a reading node $[R_i; (S'_d)_{d \in D}]$ continue with $S'_d$, where $d$ is the root of $T_i$, and replace $T_i$ by its $\ar{d}$ immediate subtrees.

\end{enumerate}

Next, we show that the predicates $\CC_{\FF(X, Y)}$ are closed under composition. We will need the following lemma.

\begin{lemma}\label{lem-domclos}
Let $(X, D)$ be an extended IFS. If $f \in \CC_{\FF(X, Y)}$, then $f \circ d^{(i, \ar{f})} \in \CC_{\FF(X, Y)}$, for all $d \in D$ and $1 \le i \le \ar{f}$.
\end{lemma}
\begin{proof}
Up to minor modifications, the result follows as the corrresponding \cite[Lemma 4.1]{be}. We include the proof for completeness reasons.
Let $d \in D$. Moreover, set
\[
F:= \set{f \circ d^{(i, \ar{f})}}{f \in \CC_{\FF(X, Y)} \wedge 1 \le i \le \ar{f}}.
\]
We use strong co-induction to prove that $F \subseteq \CC_{\FF(X, Y)}$. That is, we derive $F \subseteq \JJJ(F \cup \CC_{\FF(X, Y)})$. In other words, we show that $\CC_{\FF(X, Y)} \subseteq G$, where
\[
G := \set{f \in \FF(X, Y)}{(\forall 1 \le i \le \ar{f})\, f \circ d^{(i, \ar{f})} \in \JJJ(F \cup \CC_{\FF(X, Y)})}.
\]
Since $\CC_{\FF(X, Y)} = \JJJ(\CC_{\FF(X, Y)})$, it suffices to prove $\JJJ(\CC_{\FF(X, Y)}) \subseteq G$, which will be done by strong induction on $\JJJ(\CC_{\FF(X, Y)})$, i.e.\ we show $\Phi(\CC_{\FF(X, Y)})(G \cap \JJJ(\CC_{\FF(X, Y)})) \subseteq G$. We have to verify Rules~(W) and (R).

(W) Let $e \in E$ and $f_1, \ldots, f_{\ar{f}} \in \CC_{\FF(X, Y)}$. We must show that $e \circ (f_1 \times \cdots \times f_{\ar{e}}) \in G$, i.e.\ for $m := \sum_{\kappa=1}^{\ar{e}} \ar{f_\kappa}$ and all $1 \le i \le m$ we have to prove that $e \circ (f_1 \times \cdots \times f_{\ar{e}}) \circ d^{(i,m)} \in \JJJ(F \cup \CC_{\FF(X, Y)})$. Let $m_i := \max \set{1 \le n \le \ar{e}}{\sum_{\kappa=1}^n \ar{f_\kappa} \le i}$ and $\hat{\imath} := i - m_i$. Then 
\[
(f_1 \times \cdots \times f_{\ar{e}}) \circ d^{(i, m)} 
= f_1 \times \cdots \times f_{m_i-1} \times (f_{m_i} \circ d^{(\hat{\imath}, m_i)}) \times f_{m_i+1} \times \cdots \times f_{\ar{e}}.
\]
By definition of $F$, we have that $f_{m_i} \circ d^{(\hat{\imath}, m_i)} \in F$. Since $\JJJ(F \cup \CC_{\FF(X, Y)})$ is closed under Rule~(W), it therefore follows that $e \circ (f_1 \times \cdots \times f_{\ar{e}}) \in G$.

(R) Let $h \in \FF(X, Y)$ and $1 \le i' \le \ar{h}$. Assume that for all $d' \in D$, $h \circ d'^{(i', \ar{h})} \in G \cap \JJJ(\CC_{\FF(X, Y)})$. We must show that $h \in G$, i.e.\ for all $1 \le i \le \ar{h}$ we have to demonstrate that
$h \circ d^{(i, \ar{h})} \in \JJJ(F \cup \CC_{\FF(X, Y)})$. 

Hereto, let $1 \le i \le \ar{h}$. In case $i = i'$, we obtain from our assumption that $h \circ d^{(i, \ar{h})} \in \JJJ(\CC_{\FF(X, Y)})$ and hence, by the monotonicity of $\JJJ$, that $h \circ d^{(i, \ar{h})} \in \JJJ(F \cup \CC_{\FF(X, Y)})$. If $i \not= i'$, then we have that for all $d' \in D$, $d'^{(i', \ar{h})} \circ d^{(i, \ar{h})} = d^{(i, \ar{h})} \circ d'^{(i', \ar{h})}$. Therefore, since by our assumption, $h \circ d'^{(i', \ar{h})} \in G$, for all $d' \in D$, we have that for all $d' \in D$, $h \circ d^{(i, \ar{h})} \circ d'^{(i', \ar{h})} \in \JJJ(F \cup \CC_{\FF(X, Y)})$. As $\JJJ(F \cup \CC_{\FF(X, Y)})$ is closed under Rule~(R), it follows that $h \circ d^{(i, \ar{h})} \in \JJJ(F \cup \CC_{\FF(X, Y)})$\end{proof}

\begin{proposition}\label{pn-clcompo}
Let $(X, D)$, $(Y, E)$ and $(Z, C)$ be extended IFS. If $f \in \CC_{\FF(Y, Z)}$ and $g_\kappa \in \CC_{\FF(X, Y)}$, for $1 \le \kappa \le \ar{f}$, then $f \circ (g_1 \times \cdots \times g_{\ar{f}}) \in \CC_{\FF(X, Z)}$.
\end{proposition}
\begin{proof} The proof is a modification of the proof of \cite[Proposition~4.2]{be}. It proceeds by co-induction. Set
\[
F := \set{f \circ (g_1 \times \cdots \times g_{\ar{f}})}{f \in \CC_{\FF(Y, Z)} \wedge g_1, \ldots, g_{\ar{f}} \in \CC_{\FF(X, Y)}}.
\] 
Then we show that $F \subseteq \JJJ^{X,Z}(F)$, that is, we show $\CC_{\FF(Y, Z)} \subseteq G$, where
\begin{equation*}
G := \set{f \in \FF(Y, Z)}{(\forall g_1, \ldots, g_{\ar{f}} \in \CC_{\FF(X, Y)})\, 
 f \circ (g_1 \times \cdots \times g_{\ar{f}}) \in \JJJ^{X,Z}(F)}.
\end{equation*}
Since $\CC_{\FF(Y, Z)} = \JJJ^{Y,Z}(\CC_{\FF(Y, Z)})$ it suffices to show $ \JJJ^{Y,Z}(\CC_{\FF(Y, Z)}) \subseteq G$. Hence, by the inductive definition of $\JJJ^{Y,Z}(\CC_{\FF(Y, Z)})$, it is sufficient to demonstrate that $\Phi^{Y,Z}(\CC_{\FF(Y, Z)})(G) \subseteq G$, that is we have to show that Rules~(W) and (R) hold.

(W) Assume that $c \in C$ and $f_1, \ldots, f_{\ar{c}} \in \CC_{\FF(Y, Z)}$. We need to show that $c \circ (f_1 \times \cdots \times f_{\ar{c}}) \in G$. For $1 \le \kappa \le \ar{c}$ assume that $g^{(\kappa)}_1, \ldots, g^{(\kappa)}_{\ar{f_\kappa}} \in \CC_{\FF(X, Y)}$. Then we must prove that $c \circ (f_1 \times \cdots \times f_{\ar{c}}) \circ (g^{(1)}_1 \times \cdots \times g^{(\ar{c})}_{\ar{f_{\ar{c}}}}) \in \JJJ^{X,Z}(F)$. Because $\JJJ^{X,Z}(F)$ is closed under Rule~(W), it suffices to demonstrate that for each $1 \le \kappa \le \ar{c}$, $f_\kappa \circ (g^{(\kappa)}_1 \times \cdots \times g^{(\kappa)}_{\ar{f_\kappa}}) \in F$, which holds by the definition of $F$.

(R) Let $f \in \FF^{(n)}(Y, Z)$ and $1 \le i \le n$. Suppose, as induction hypothesis, that for all $e \in E$, $f \circ e^{(i, n)} \in G$. We have to show that $f \in G$, i.e.\ $\CC_{\FF(X, Y)} \subseteq H^{(i)}$, where
\[
H^{(i)} := \set{g \in \FF(X, Y)}{(\forall g_1, \ldots, g_n \in \CC_{\FF(X, Y)}) [g_i = g 
\rightarrow f \circ (g_1 \times \cdots \times g_n) \in \JJJ^{X,Z}(F)]}.
\]
Since $\CC_{\FF(X, Y)} = \JJJ^{X,Y}(\CC_{\FF(X, Y)})$, it suffices to show that $\JJJ^{X,Y}(\CC_{\FF(X, Y)}) \subseteq H^{(i)}$, which we do by side induction, that is, we show that 
\[
\Phi^{X, Y}(\CC_{\FF(X, Y)})(H^{(i)}) \subseteq H^{(i)}.
\]

($\text{W}_{\text{side}}$) Let $e \in E$ and $g_1, \ldots, g_n, g'_1, \ldots, g'_{\ar{e}} \in \CC_{\FF(X, Y)}$ so that $g_i = e \circ (g'_1 \times \cdots \times g'_{\ar{e}})$. Then
\begin{align*}
f \circ (g_1 \times \cdots \times g_n)
&= f \circ (g_1 \times \cdots \times g_{i-1} \times (e \circ (g'_1 \times \cdots \times g'_{\ar{e}})) \times g_{i+1} \times \cdots \times g_n) \\
&= f \circ (e^{(i, n)} \circ (g_1 \times \cdots \times g_{i-1} \times g'_1 \times \cdots \times g'_{\ar{e}} \times g_{i+1} \times \cdots \times g_n)) \\
&= (f \circ e^{(i, n)}) \circ (g_1 \times \cdots \times g_{i-1} \times  g'_1 \times \cdots \times g'_{\ar{e}} \times g_{i+1} \times \cdots \times g_n).
\end{align*}
By the main induction hypothesis, $f \circ e^{(i, n)} \in G$ and hence 
\[
(f \circ e^{(i, n)}) \circ (g_1 \times \cdots \times g_{i-1} \times  g'_1 \times \cdots \times g'_{\ar{e}} \times g_{i+1} \times \cdots \times g_n) \in \JJJ^{X,Z}(F).
\] 
Therefore, $e \circ (g'_1 \times \cdots \times g'_{\ar{e}}) \in H^{(i)}$.

($\text{R}_{\text{side}}$) Let $g \in \CC_{\FF(X, Y)}$ and $1 \le j \le \ar{g}$. Suppose that for all $d \in D$, $g \circ d^{(j, \ar{g})} \in H^{(i)}$.  We have to show that $g \in H^{(i)}$. Assume to this end that $h_1, \ldots, h_n \in \CC_{\FF(X, Y)}$ with $h_i = g$. We must derive that $f \circ (h_1 \times \cdots \times h_n) \in \JJJ^{X,Z}(F)$. Let $m := \sum_{\kappa = 1}^n \ar{h_\kappa}$ and $\hat{\jmath} := j+\sum_{\kappa=1}^{i-1}  \ar{h_\kappa}$. Because $\JJJ^{X,Z}(F)$ is closed under Rule~(R), it suffices to show that for all $d \in D$, $(f \circ (h_1 \times \cdots \times h_n)) \circ d^{(\hat{\jmath}, m)} \in \JJJ^{X,Z}(F)$, which is the case as the $i$-th element of $(h_1 \times \cdots \times h_n) \circ d^{(\hat{\jmath}, m)}$ is $g \circ d^{(j,\ar{g})}$ and $g \circ d^{(j,\ar{g})} \in \CC_{\FF(X,Y)}$ by Lemma~\ref{lem-domclos}. 
\end{proof}

\begin{lemma}\label{lem-basics}
Let $(X, D)$ be an extended IFS. Then the following statements hold:
\begin{enumerate}

\item\label{lem-id}
Let $\fun{\id_X}{X}{X}$ be the identity on $X$. Then $\id_X \in \CC_{\FF(X, X)}$.

\item\label{lem-proj}
For any $n \in \NN_{0}$ and $1 \le i \le n$, let $\fun{\pr^{(n)}_{i}}{X^{n}}{X}$ be the projection onto the $i$-th component. Then $\pr^{(n)}_i \in \CC^{(n)}_{\FF(X, X)}$.

\item\label{lem-diag}
For $n>0$, let $\fun{\diag_X^{(n)}}{X}{X^n}$ be given by $\diag^{(n)}(x) := x^{(n)}$. Then $\diag_X^{(n)} \in \CC_{\FF(X, X^n)}$.

\end{enumerate}
\end{lemma}
\begin{proof} For all statements the proof proceeds by co-induction. 

(\ref{lem-id})
We show that $\{ \id_X \} \subseteq \JJJ(\{ \id _X \})$. Because of Rule~(R) it suffices to show that for all $d \in D$, $\id_X \circ\, d \in \JJJ(\{ \id_X \})$. By Rule~(W) we have that for $d \in D$, $d \circ (\id_X \times \cdots \times \id_X) \in \JJJ(\{ \id_X \})$. Since $d \circ (\id_X \times \cdots \times \id_X) = {\id_X} \circ d$, we are done.

(\ref{lem-proj})
Since for $d \in D$,
\[
d \circ (\pr^{(n+\ar{d})}_i \times \cdots \times \pr^{(n+\ar{d})}_{i+\ar{d}-1}) = {\pr^{(n)}_{i}} \circ  d^{(i, n)},
\]
it follows in the same way that 
\[
\set{\pr^{(m)}_j}{m>0 \wedge 1 \le j \le m} \subseteq \JJJ(\set{\pr^{(m)}_j}{m>0 \wedge 1 \le j \le m}).
\]

(\ref{lem-diag})
Let $d \in D$ and note that for $x_{1}, \ldots, x_{\ar{d}} \in X$,
\begin{align*}
(\pair{d, \ldots, d} \circ (\diag^{(n)}_{X} \times \cdots \times \diag^{(n)}_{X}))(x_{1}, \ldots, x_{\ar{d}})
&= \pair{d, \ldots, d}(x_{1}^{(n)}, \ldots, x_{\ar{d}}^{(n)}) \\
&= (d(x_{1}, \ldots, x_{\ar{d}}), \ldots, d(x_{1}, \ldots, x_{\ar{d}})) \\
&= ({\diag^{(n)}_{X}} \circ d)(x_{1}, \ldots, x_{\ar{d}}). 
\end{align*}
Thus, $\pair{d, \ldots, d} \circ (\diag^{(n)}_{X} \times \cdots \times \diag^{(n)}_{X}) = {\diag^{(n)}_{X}} \circ d$, from which the statement follows as in case~(\ref{lem-id}). 
\end{proof}

In the above case of the $n$-th power of a given space $X$, the projections were treated as $n$-ary maps. Let us now consider the general case of the product of spaces. Let to this end $(X_{1}, D_{1}), \ldots, (X_{n}, D_{n})$ be extended IFS, and let $\fun{\pr^{(n)}_{i}}{\bigtimes_{j=1}^{n} X_{j}}{X_{i}}$ again denote the projection onto the $i$-th component. Then we have for $\pair{d_{1}, \ldots, d_{n}} \in D^{\times}$ that
\[
\pr^{(n)}_{i} \circ\,\, \pair{d_{1}, \ldots, d_{n}} = d_{i} \circ (\pr^{(n)}_{i} \times \cdots \times \pr^{(n)}_{i}) \quad\text{($s_{D}$-times)},
\] 
which leads to the next result.

\begin{lemma}\label{lem-genprod}
Let $(X_{1}, D_{1}), \ldots, (X_{n}, D_{n})$ be extended IFS. Then $\pr^{(n)}_{i} \in \CC^{(1)}_{\FF(\bigtimes_{j=1}^{n} X_{j}, X_{i})}$, for $1 \le i \le n$.
\end{lemma}

\begin{proposition}\label{pn-eval}
Let $(X, D)$ and $(Y, E)$ be extended IFS. Moreover, for $m \in \NN$, let 
\[
\fun{\ev}{\FF^{(m)}(X, Y) \times X^m}{Y}
\]
with $\ev(f, x) := f(x)$ be the \emph{evaluation map}. Then
\[
\ev[\CC^{(m)}_{\FF(X, Y)} \times \CC_X^m] \subseteq \CC_Y.
\]
\end{proposition}
\begin{proof}
The statement is derived by co-induction. Let the operator $\fun{\Psi}{\PPP(Y)}{\PPP(Y)}$ be defined by
\[
\Psi(L) := \{\, y \mid (\exists e \in E) (\exists y_1, \ldots, y_{\ar{e}} \in L)\, y = e(y_1, \ldots, y_{\ar{e}}) \,\}.
\]
Then $\CC_Y = \nu\Psi$. Set 
\[
Z:= \set{y \in Y}{(\exists f \in \CC_{\FF(X, Y)}) (\exists x_1, \ldots x_{\ar{f}} \in \CC_X)\, y = f(x_1, \ldots, x_{\ar{f}})}.
\]
We need to show that $Z \subseteq \Psi(Z)$. Let to this end
\[
G := \{\, f \in \FF(X, Y) \mid f[\CC^{\ar{f}}_X] \subseteq \Psi(Z) \,\}.
\]

Then we have to verify that $\CC_{\FF(X, Y)} \subseteq G$. Since $\CC_{\FF(X, Y)} = \JJJ(\CC_{\FF(X, Y)})$, it is equivalent to prove that $\JJJ(\CC_{\FF(X, Y)}) \subseteq G$. By the inductive definition of $\JJJ(\CC_{\FF(X, Y)})$ it suffices in addition to derive that $\Phi(\CC_{\FF(X, Y)})(G) \subseteq G$, that is we need to show that
\begin{itemize}

\item[(W)] If $e \in E$ and  $(f_1, \ldots, f_{\ar{e}}) \in \CC_{\FF(X, Y)}$, then $e \circ (f_1 \times \cdots \times  f_{\ar{e}}) \in G$.

\item[(R)]  If $f \in \FF^{(m)}(X, Y)$ and $1 \le i \le \ar{f}$ so that $f \circ d^{(i, m)} \in G$, for all $d \in D$, then $f \in G$.

\end{itemize}

(W) Let $e \in E$,  $f_1, \ldots, f_{\ar{e}} \in \CC_{\FF(X, Y)}$, and $\vec{x}_\kappa \in \CC^{\ar{f_\kappa}}_X$, for $1 \le \kappa \le \ar{e}$. Then $f_\kappa(\vec{x}_\kappa) \in Z$. Hence, $e(f_1(\vec{x}_1), \ldots, f_{\ar{e}}(\vec{x}_{\ar{e}})) \in \Psi(Z)$.

(R) Let $f \in \FF^{(m)}(X, Y)$ and $1 \le i \le m$ with $f \circ d^{(i, m)} \in G$, for all $d \in D$.  Let $x_1$, \ldots, $x_m \in \CC_X$. Then there is some $d_i \in D$ and some $\vec{x}'_i \in \CC^{\ar{d_i}}_X$ so that $x_i = d_i(\vec{x}'_i)$. It follows that
\begin{align*}
f(x_1, \ldots, x_m) 
&= f(x_1, \ldots, x_{i-1}, d_i(\vec{x}'_i), x_{i+1}, \ldots, x_m) \\
&= (f \circ d^{(i, m)})(x_1, \ldots, x_{i-1}, \vec{x}'_i, x_{i+1}, \ldots, x_m) \in \Psi(Z). \qedhere
\end{align*}
\end{proof}

For $n \in \NN_{0}$, let the set $\underline{n} := \{ 0, \ldots, n-1 \}$ be endowed with the discrete topology. Moreover, for $i \in \underline{n}$, let $\fun{g_i}{\underline{n}}{\underline{n}}$ map $\underline{n}$ constantly onto $i$. Then $(\underline{n}, \{ g_1, \ldots, g_{n-1} \})$ is a covering IFS with $\underline{n} = \CC_{\underline{n}}$.

\begin{corollary}\label{cor-power}
Let $(X, D)$ be an IFS. Then $\CC^{(1)}_{\FF(\underline{n}, X)} = \CC^n_X$.
\end{corollary}
\begin{proof}
As a consequence of the preceding proposition we have that $\CC^{(1)}_{\FF(\underline{n}, X)} \subseteq \CC^n_X$. The converse inclusion follows by co-induction. We show that $\FF^{(1)}(\underline{n}, \CC_X) \subseteq \JJJ(\FF^{(1)}(\underline{n}, \CC_X))$.

Let $f \in \FF^{(1)}(\underline{n}, \CC_X)$ and $i \in \underline{n}$. Then there are $d_i \in D$ and $x_{i} \in \CC_X$ so that $f(i) = d_i(x_{i})$. Set $h_{i}(a) = x_{i}$, for $a \in \underline{n}$. Then $h_{i} \in \FF^{(1)}(\underline{n}, \CC_X)$. Moreover,  $f \circ g_i = d_i \circ h_{i}$. As $\JJJ(\FF^{(1)}(\underline{n}, \CC_X))$ is closed under Rule~(W), we obtain that $d_i \circ h_{i} \in \JJJ(\FF^{(1)}(\underline{n}, \CC_X))$. Hence, 
$f \circ g_i \in \JJJ(\FF^{(1)}(\underline{n}, \CC_X))$, for all $i \in \underline{n}$. With Rule~(R) we therefore obtain that $f \in \JJJ(\FF^{(1)}(\underline{n}, \CC_X))$.
\end{proof}

Observe the difference between $\CC_{X^{n}}$ and $\CC^{n}_{X}$. In the first case the digits operate on all components simultaneously, whereas in the other case they do this in an uncoordinated way and approximate the single components just as needed.  

\begin{proposition}\label{pn-fctcartprod}
Let $(X, D), (Y_1, E_1)$, and $(Y_2, E_2)$ be extended IFS. Then for $f \in \CC_{\FF(X, Y_1)}$ and $g \in \CC_{\FF(X, Y_2)}$, $f \times g \in \CC_{\FF(X, Y_1 \times Y_2)}$.
\end{proposition}
\begin{proof}
Without loss of generality we assume that all maps in $E_1 \cup E_2$ have the same arity, say $m$. 
The proof of the statement proceeds by co-induction. Let
\[
F := \set{f \times g}{f \in \CC_{\FF(X, Y_1)} \wedge g \in \CC_{\FF(X, Y_2)}}.
\]
Then we must show that $F \subseteq \JJJ^{X, Y_1 \times Y_2}(F)$. Set
\[
A:= \set{f \in \FF(X, Y_1)}{(\forall g \in \CC_{\FF(X, Y_2)})\, f \times g \in \JJJ^{X, Y_1 \times Y_2}(F)}.
\]
We show that $\CC_{\FF(X, Y_1)} \subseteq A$. Since $\CC_{\FF(X, Y_1)} = \JJJ^{X, Y_1}(\CC_{\FF(X, Y_1)})$, it suffices to prove that $\JJJ^{X, Y_1}(\CC_{\FF(X, Y_1)}) \subseteq A$. By induction we show that $\Phi^{X, Y_1}(\CC_{\FF(X, Y_1)})(A) \subseteq A$, which means that we have to verify Rules~(W) and (R).

($\text{W}_\text{main}$) Let $e_1 \in E_1$ and $f_1, \ldots, f_m \in \CC_{\FF(X, Y_1)}$. We need to demonstrate that $e_1 \circ (f_1 \times \cdots \times f_m) \in A$. To this end we have to show that for all $g \in \CC_{\FF(X, Y_2)}$, $(e_1 \circ (f_1 \times \cdots \times f_m)) \times g \in \JJJ^{X, Y_1 \times Y_2}(F)$. Let
\[
B(h) := \set{g \in \FF(X, Y_2)}{h \times g \in \JJJ^{X, Y_1 \times Y_2}(F)}.
\]
For $h := e_1 \circ (f_1 \times \cdots \times f_m)$ we prove that $\CC_{\FF(X, Y_2)} \subseteq B(h)$, for which it suffices to demonstrate that
 \begin{equation}\label{eq-fctcartprod}
 \Phi^{X, Y_2}(\CC_{\FF(X, Y_2)})(B(h)) \subseteq B(h),
 \end{equation}
 which is done by side induction.
 
 ($\text{W}_\text{side}$) Let $e_2 \in E_2$ and $g_1, \ldots, g_m \in \CC_{\FF(X, Y_2)}$. We need to show that
 \[
 (e_1 \circ (f_1 \times \cdots \times f_m)) \times (e_2 \circ (g_1 \times \cdots \times g_m)) \in \JJJ^{X_1 \times X_2, Y_1 \times Y_2}(F).
 \]
 We have that
 \[
(e_1 \circ (f_1 \times \cdots \times f_m)) \times (e_2 \circ (g_1 \times \cdots \times g_m)) 
= \pair{e_1, e_2} \circ ((f_1 \times g_1) \times \cdots \times (f_m \times g_m)),
\]
where $\pair{e_1, e_2} \in E^{\times}$ and $f_1 \times g_1, \ldots, f_m \times g_m \in F$. Hence, $(e_1 \circ (f_1 \times \cdots \times f_m)) \times (e_2 \circ (g_1 \times \cdots \times g_m)) \in \JJJ^{X, Y_1 \times Y_2}(F)$, by Rule~(W).

($\text{R}_\text{side}$) Let $g \in \FF(X, Y_2)$ and $1 \le i \le \ar{g}$ so that $g \circ d^{(i, \ar{g})}
\in B(h)$, for all $d \in D$. Then $h \times (g \circ d^{(i, \ar{g})}) \in \JJJ^{X, Y_1 \times Y_2}(F)$. Let $\tilde{h} := h \times g$. Then $\ar{\tilde{h}} = \ar{h} + \ar{g}$. For $j := \ar{h} + i$ we therefore have that $\tilde{h} \circ d^{(j, \ar{h}+\ar{g})} = h \times (g \circ d^{(i, \ar{g})})$. It follows that $\tilde{h} \circ d^{(j, \ar{h}+\ar{g})} \in \JJJ^{X, Y_1 \times Y_2}$, for all $d \in D$, from which we obtain by Rule~(R) that $\tilde{h} \in \JJJ^{X, Y_1 \times Y_2}$. Hence, $g \in B(h)$. This proves (\ref{eq-fctcartprod}).

($\text{R}_\text{main}$) It remains to verify Rule~(R) in the main induction. Let $f \in A$ and $1 \le i \le \ar{f}$ such that for all $d \in D$, $f \circ d^{(i, \ar{f})} \in A$. Thus, $(f \circ d^{(i, \ar{f})}) \times g \in \JJJ^{X, Y_1 \times Y_2}(F)$, for all $d \in D$  and $g \in \CC_{\FF(X, Y_{2})}$. Set $\bar{h}_g := f \times g$. Then $\ar{\bar{h}_g} = \ar{f} + \ar{g}$ and for all $d \in D$,
\[
\bar{h}_g \circ d^{(i, \ar{f}+\ar{g})} = (f \circ d^{(i, \ar{f})}) \times g.
\]
Since  $(f \circ d^{(i, \ar{f})}) \times g \in \JJJ^{X, Y_1 \times Y_2}(F)$ by our assumption, we obtain with Rule~(R) that $\bar{h}_g \in \JJJ^{X, Y_1 \times Y_2}(F)$. Thus, $f \in A$.
\end{proof}

Let $\mathbf{eIFS}$ be the category with extended IFS $(X, D)$ as objects, and for two objects $(X_{1}, D_{1})$ and $(X_{2}, D_{2})$, $\CC_{\FF(X_{1}, X_{2})}$ as set of morphisms from $(X_{1}, D_{1})$ to $(X_{2}, D_{2})$. Then it follows that $\prod_{i=1}^{2} (X_{i}, D_{i}) := (X_{1} \times X_{2}, D^{\times})$ with canonical projection $\pr^{(2)}_{i}$, for $i = 1, 2$, is the categorical product of $(X_{1}, D_{1})$ and $(X_{2}, D_{2})$.

Let $(Y_{1}, D_{1})$ and $(Y_{2}, D_{2})$ be further objects in $\mathbf{eIFS}$ and $f_{i} \in \CC_{\FF(X_{i}, Y_{i})}$, for $i = 1, 2$. Set $h_{i} := f_{i} \circ (\pr^{(2)}_{i} \times \cdots \times \pr^{(2)}_{i})\,\, \text{($\ar{f_{i}}$-times)}$. Then $h_{i} \in \CC_{\FF(X_{1} \times X_{2}, Y_{i})}$. Hence, $h_{1} \times h_{2} \in \CC_{\FF(X_{1} \times X_{2}, Y_{1} \times Y_{2})}$. Define the action of $\prod$ on  $f_{1}, f_{2}$ by $\prod(f_{1}, f_{2}) := h_{1} \times h_{2}$. Then $\fun{\prod}{\mathbf{eIFS} \times \mathbf{eIFS}}{\mathbf{eIFS}}$ is a functor, showing that the product construction presented in Section~\ref{sec-prod} is functorial. 

As we have already seen, in the framework of computable digit spaces computability is the appropriate requirement for a map between such spaces to be considered uniformly continuous, constructively. We will now show that the above co-inductive-inductively defined function class captures exactly these functions.
Let 
\begin{equation*}
\UC(X, Y) := \{\, f \in \FF(X, Y) \mid \text{$f$ computable} \,\}.
\end{equation*}

\begin{proposition}\label{pn-contophi}\sloppy
Let $(X, D, Q_X)$ and $(Y, E, Q_Y)$ be computable digit spaces so that $(Y, E, Q_Y)$ is decidable and well-covering with approximable choice. Then $\UC(X, Y) \subseteq \CC_{\FF(X, Y)}$.
\end{proposition}
\begin{proof}
We show the statement by co-induction, that is we prove that 
\[
\UC(X, Y) \subseteq \JJJ(\UC(X, Y)).
\]

Let $\varepsilon_Y \in \QQ_+$ be a well-covering number for $Y$ and set
\[
V := \set{(y, y') \in Y \times Y}{\rho_Y(y, y') < \varepsilon_Y / 2}.
\]
Moreover, let $M_X$ be a bound of $X$ and $q_X < 1$ the maximum of the contraction factors of the $d \in D$. Let $\vec{x} \in X^m$ denote the $m$-tuple $(x_1, \ldots, x_m)$, and similarly for $\vec{x}'$ and $\vec{\imath}$.
Define
\begin{equation*}
U_{\vec{\imath}}^{(m)} := \{\,(\vec{x}, \vec{x}') \in X^m \times X^m \mid 
 (\forall 1 \le \kappa \le m) \rho_X(x_\kappa, x'_\kappa) \le q_X^{i_\kappa} \cdot M_X \,\}, 
\end{equation*}
and set
\[
\UC^{(m)}_{\vec{\imath}} := \set{f \in \UC^{(m)}(X, Y)}{U^{(m)}_{\vec{\imath}} \subseteq (f \times f)^{-1}[V]}.
\]
Since every $f \in \UC^{(m)}(X, Y)$ comes equipped with a computable modulus of continuity, an $\vec{i} \in \NN^{m}$ can be computed for each such $f$ so that $f \in \UC^{(m)}_{\vec i}$. Therefore, it suffices to show that $\UC^{(m)}_{\vec{\imath}} \subseteq \JJJ(\UC(X, Y))$, for all $m$ and every $\vec{\imath} \in \NN^m$, which will be done by induction on $i := \max \set{i_\kappa}{1 \le \kappa \le m}$.

For $i = 0$ we have that $U^{(m)}_{\vec{\imath}} = X^m \times X^m$. Therefore, it follows for $f \in \UC^{(m)}_{\vec{\imath}}$ and $\vec{u} \in Q_X^m$ that $f[X^m] \subseteq \ball{\rho_Y}{f(\vec{u})}{\varepsilon_Y / 2}$. Since $f$ is computable, a basic element $v \in Q_Y$ can be found with $\rho_Y(f(\vec{u}), v) < \varepsilon_{Y} / 2$. Use decidability to pick some $e \in E$ with $\ball{\rho_Y}{v}{\varepsilon_Y} \subseteq \range(e)$. Then $f[X^m] \subseteq \ball{\rho_Y}{v}{\varepsilon_Y} \subseteq \range(e)$. 

Note that by Proposition~\ref{pn-epsrightinv}, digit $e$ has a computable right inverse $e'$. For $1 \le \kappa \le \ar{e}$, set $g_\kappa := {\pr^{(\ar{e})}_\kappa} \circ\,  e' \circ f$. Then $g_{\kappa}$ is computable, that is  $g_\kappa \in \UC^{(m)}(X, Y)$. Since $f = e \circ (g_1 \times \cdots \times g_{\ar{e}})$, we therefore obtain with Rule~(W) that $f \in \JJJ(\UC(X, Y))$.

Now, assume that $i > 0$ and let $f \in \UC^{(m)}_{\vec{\imath}}$. Then $(f \times f)^{-1}[V] \supseteq U^{(m)}_{\vec{\imath}}$. Choose $1 \le j \le m$ such that $i_j > 0$.
We want to apply Rule~(R). Therefore, we have to show that for all $d \in D$, $f \circ d^{(j, m)} \in \JJJ(\UC(X, Y))$. Since  $d \in D$ is a contraction, there is some $\ell < i_j$ so that 
\[
(d \times d)^{-1}[U^{(1)}_{i_j}]\supseteq \set{(\vec{z}, \vec{z}') \in X^{\ar{d}} \times X^{\ar{d}}}{(\forall 1 \le \sigma 
\le \ar{d})\, \rho_{X}(z_{\sigma}, z'_{\sigma}) \le q^{\ell}_{X} \cdot M_{X}}.
\]
For $1 \le \kappa \le \ar{d}+m-1$ set 
\[
k_\kappa := \begin{cases}
				 i_\kappa & \text{if $1 \le \kappa < j$ or $\ar{d}+j \le \kappa < \ar{d}+m$} \\
				 \ell & \text{if $j \le \kappa < \ar{d}+j$.} 
        		\end{cases}
\]		
Let $\bar{k} := \max \set{k_\kappa}{1 \le \kappa \le \ar{d}+m-1}$. Then $\bar{k} < i$. Moreover, we have that $(d^{(j, m)} \times d^{(j, m)})^{-1}[U^{(m)}_{\vec{\imath}}] \supseteq U^{(\ar{d}+m-1)}_{\vec{k}}$.
Hence,  $(f \circ d^{(j, m)} \times f \circ d^{(j, m)})^{-1}[V] \supseteq U^{(\ar{d}+m-1)}_{\vec{k}}$, from which it follows with the induction hypothesis that $ f \circ d^{(j, m)} \in \JJJ(\UC(X, Y))$.
\end{proof}

\begin{proposition}\label{pn-phitocon}
Let $(X, D, Q_X)$ and $(Y, E, Q_Y)$ be decidable and well-covering computable digit spaces with approximable choice. Then $\CC_{\FF(X, Y)} \subseteq \UC(X, Y)$.
\end{proposition}
\begin{proof}
Since computable digit spaces with approximable choice are constructively dense, Proposition~\ref{pn-beq} permits to derive the assertion for the dense subspaces $Q^{(z)}_{D}$ and  $Q^{(z')}_{E}$, respectively, instead of $Q_{X}$ and $Q_{Y}$. To simplify notation we write $Q_{D}$ instead of $Q^{(z)}_{D}$ and similarly for $Q^{(z')}_{E}$. Let $z \in Q_{X}$ and $z' \in Q_{Y}$ be fixed for the rest of the proof.

Now, let $M_X$ and $M_{Y}$, respectively, be bounds of $X$ and $Y$, and $q_{X}, q_Y < 1$ be the maximum of the contraction factors of the $d \in D$ and/or $e \in E$.  For $j \in \NN$, set
\[
V_j := \set{(y, y') \in Y \times Y}{\rho_Y(y, y') \le q_Y^j \cdot M_Y}.
\]
Moreover, let 
\[
U^{(m)}_{i} := \set{(\vec{x}, \vec{x}') \in X^{m} \times X^{m}}{(\forall 1 \le \kappa \le m)\, \rho_{X}(x_{\kappa}, x'_{\kappa}) \le q_{X}^{i} \cdot M_{X}}
\]
and define
\begin{equation*}
\begin{split}
G_j := \{\, f \in \FF(X, Y) \mid (\exists i \in \NN)\, U^{(\ar{f})}_{i} \subseteq (f \times f)^{-1}[V_j] \wedge \hspace{3cm} \\
(\forall \vec{u} \in Q^{\ar{f}}_{D}) (\exists v \in Q_{E})\, (f(\vec{u}), v) \in V_{j} \,\}.
\end{split}
\end{equation*}
Then we have to show that $\CC_{\FF(X, Y)} \subseteq G_j$, for all $j \in \NN$. We proceed by induction on $j$.

The case $j = 0$ is obvious, as $V_j = Y \times Y$ and hence for $i := 0$, $(f \times f)^{-1}[V_j] \supseteq U^{(\ar{f})}_{i}$, for every $f \in \FF(X, Y)$. In addition, $(f(\vec{u}), v) \in V_{j}$, for any $\vec{u} \in Q^{\ar{f}}_{D}$ and $v \in Q_{E}$. 

Now, assume that $j > 0$. We show by side induction that $\JJJ(\CC_{\FF(X, Y)}) \subseteq G_j$. Note here that $\CC_{\FF(X, Y)} = \JJJ(\CC_{\FF(X, Y)})$. By the inductive definition of $\JJJ(\CC_{\FF(X, Y)})$ it therefore suffices to show that $\Phi(\CC_{\FF(X, Y)})(G_j) \subseteq G_j$, that is we need to show that
\begin{itemize}
\item[(W)] If $e\in E$ and $f_1, \ldots, f_{\ar{e}} \in \CC_{\FF(X, Y)}$, then $e \circ (f_1 \times \cdots \times f_{\ar{e}}) \in G_j$.

\item[(R)] If $f \in \FF(X, Y)$ and $1 \le k \le \ar{f}$  so that $f \circ d^{(k, \ar{f})} \in G_j$, for all $d \in D$, then $f \in G_j$.

\end{itemize}

(W) Let $e \in E$ and $f_1, \ldots, f_{\ar{e}} \in \CC_{\FF(X, Y)}$. As a digit map, $e$ in particular has a computable modulus of continuity. Thus, we can effectively find some $\ell \in \NN$ so that $(e \times e)^{-1}[V_{j}] \supseteq V_{\ell}^{\ar{e}}$. Since $e$ is a contraction, it follows that $\ell < j$.  By the main induction hypothesis $f_\kappa \in G_{\ell}$, for $1 \le \kappa \le \ar{e}$. Thus, we are  given $i_{1},\ldots, i_{\ar{e}} \in \NN$ so that $U^{(\ar{f_\kappa})}_{i_{\kappa}} \subseteq (f_\kappa \times f_\kappa)^{-1}[V_{\ell}]$, for $1 \le \kappa \le \ar{e}$. Furthermore, for each such $\kappa$ and every $\vec{u}_{\kappa} \in Q^{\ar{f_{\kappa}}}_{D}$ we can effectively find some $v_{\kappa} \in Q_{E}$ with $(f_{\kappa}(\vec{u}_{\kappa}), v_{\kappa}) \in V_{\ell}$. Let $\bar{\imath} := \max \set{i_{\kappa}}{1 \le \kappa \le \ar{e}}$ and $n := \sum_{\kappa=1}^{\ar{e}} \ar{f_\kappa}$. Then $U^{(\ar{f_{\kappa}})}_{\bar{\imath}} \subseteq U^{(\ar{f_{\kappa}})}_{i_{\kappa}}$, for $1 \le \kappa \le \ar{e}$, and
\begin{align*}
&((e\circ (f_1 \times \cdots \times f_{\ar{e}})) \times (e\circ (f_1 \times \cdots \times f_{\ar{e}})))[U^{(n)}_{\bar{\imath}}] \\
&\hspace{2em} (e \times E) \circ ((f_{1} \times f_{1}) \times \cdots \times (f_{\ar{e}} \times f_{\ar{e}}))[U^{(\ar{f_{1}})}_{i_{1}} \times \cdots \times U^{(\ar{f_{\ar{e}}})}_{i_{\ar{e}}}] \\
&\hspace{2em} (e \times e)[(f_{1} \times f_{1})[U^{(\ar{f_{1}})}_{i_{1}}] \times \cdots \times (f_{\ar{e}} \times f_{\ar{e}})[U^{(\ar{f_{\ar{e}}})}_{i_{\ar{e}}}]] \\
& \hspace{2em} \subseteq (e \times e)[V_{\ell}^{\ar{e}}] \\
& \hspace{2em} \subseteq V_j.
\end{align*}
Moreover, we have that 
\begin{align*}
&((e \circ (f_{1} \times \cdots \times f_{\ar{e}}))(\vec{u}_{1}, \ldots, \vec{u}_{\ar{e}}), e(v_{1}, \ldots, v_{\ar{e}})) \\
&\hspace{2em} = (e(f_{1}(\vec{u}_{1}), \ldots, f_{\ar{e}}(\vec{u}_{\ar{e}})), e(v_{1}, \ldots, v_{\ar{e}})) \\
&\hspace{2em} \in (e \times e)[V_{\ell}^{\ar{e}}] \\
&\hspace{2em} \subseteq V_{j}.
\end{align*}
Note that by definition, $e(v_{1}, \ldots, v_{\ar{e}}) \in Q_{E}$. Thus, we have that $e \circ (f_1 \times \cdots \times f_{\ar{e}}) \in G_j$.

(R) Let $f \in \FF(X, Y)$ and $1 \le i \le \ar{f}$ so that $f \circ d^{(i, \ar{f})} \in G_j$, for all $d \in D$. Then, for every $d \in D$, we are given some $\ell_d \in \NN$ with $((f \circ d^{(i, \ar{f})}) \times (f \circ d^{(i, \ar{f})}))[U^{(\ar{d}+\ar{f}-1)}_{\ell_d}] \subseteq V_j$. By Proposition~\ref{pn-epsrightinv}  each $d \in D$ has a computable right inverse $d'$. Thus, some $r_d \in \NN$ can effectively be found with $(d' \times d')^{-1}[U^{(\ar{d})}_{\ell_d}] \supseteq U^{(1)}_{r_d}$. Define $\bar{r} := \max \set{r_d}{d \in D}$ and $\hat{r} : = \min \set{r \ge \bar{r}}{q_X^r \le \varepsilon_X / (2M_X)}$, where $\varepsilon_X$ is a well-covering number for $X$.

Now, let $\vec{x} = (x_1, \ldots, x_{\ar{f}}) \in X^{\ar{f}}$. Then $\ball{\rho_X}{x_i}{q_X^{\hat{r}} \cdot M_X} \subseteq \ball{\rho_X}{x_i}{\varepsilon_X / 2}$. By Theorem~\ref{thm-ctoa} we can therefore compute a basic element $u \in \ball{\rho_x}{x_i}{\varepsilon_X /2}$. Use decidability to pick some $\tilde{d} \in D$ with $\ball{\rho_X}{u}{\varepsilon_X} \subseteq \range(\tilde{d})$. Then $\ball{\rho_X}{x_i}{q_X^{\hat{r}} \cdot M_X} \subseteq \range(\tilde{d})$. Let $\tilde{d}'(x_i) = (\tilde{z}_1, \ldots, \tilde{z}_{\ar{\tilde{d}}})$. As  $\hat{r} \ge r_{\tilde{d}}$, we moreover have that $U^{(1)}_{\hat{r}} \subseteq U^{(1)}_{r_{\tilde{d}}}$ and hence that
\begin{equation}\label{eq-phitocon}
\tilde{d}'[\ball{\rho_x}{x_i}{q^{\hat{r}} \cdot M_X}]  
\subseteq \set{(z_1, \ldots, z_{\ar{\tilde{d}}})}{(\forall 1 \le \kappa \le \ar{\tilde{d}})\, \rho_X(z_\kappa, \tilde{z}_\kappa) \le q_X^{\ell_{\tilde{d}}} \cdot M_X}.
\end{equation}

Define $k := \max \{ \ell_{\tilde{d}}, \hat{r} \}$. We will show that $(f \times f)[U^{(\ar{f})}_{k}] \subseteq V_j$. 
Let to this end $(\vec{x}, \vec{z}) \in U^{(\ar{f})}_{k}$ with $\vec{x} = (x_1, \ldots, x_{\ar{f}})$ and similarly for $\vec{z}$. Then $z_i \in \ball{\rho_X}{x_i}{q_X^{\hat{r}}\cdot M_X} \subseteq \range(\tilde{d})$. With (\ref{eq-phitocon}) and the assumption it follows that 
\begin{equation*}
\begin{split}
(f \times f)(\vec{x}, \vec{z}) 
= ((f \times f) \circ (\tilde{d}^{(i, m)} \times \tilde{d}^{(i, m)}) \circ (\tilde{d}'^{(i, m)} \times \tilde{d}'^{(i, m)}))(\vec{x}, \vec{z}) \hspace{2cm} \mbox{} \\
 \in ((f \times f) \circ (\tilde{d}^{(i, m)} \times \tilde{d}^{(i, m)}))[U^{(\ar{\tilde{d}}+\ar{f}-1)}_{\ell_{\tilde{d}}}] \subseteq V_j.
\end{split}
\end{equation*}

It remains to verify the second requirement in the definition of $G_{j}$. Let to this end $\vec{u} \in Q_{D}^{\ar{f}}$ with $\vec{u} = (u_{1}, \ldots, u_{\ar{f}})$. Then there is some $d \in D$ and there are $w_{1}, \ldots, w_{\ar{d}} \in Q_{D}$ such that $u_{i} = d(w_{1}, \ldots, w_{\ar{d}})$. By our assumption, there is some $v \in Q_{E}$ with 
\[
((f \circ d^{(i, \ar{f})})(u_{1}, \ldots, u_{i-1}, w_{1}, \ldots, w_{\ar{d}}, u_{i+1}, \ldots, u_{\ar{f}}), v) \in V_{j}.
\]
Whence, $((f(\vec{u}), v) \in V_{j}$. This shows that $f \in G_{j}$.

Now, let $\CC_{\FF(X, Y)}$ and $p \in \QQ_{+}$. Then some $j \in \NN$ can be computed with $q_{Y}^{j} \cdot M_{Y} \le p$. The above proof shows how for given $j \in \NN$ some $i \in \NN$ can effectively be obtained so that, whenever $\vec{x}, \vec{x}' \in X^{\ar{f}}$ with $\rho_{X}(\vec{x}, \vec{x}') \le q_{X}^{i} \cdot M_{X}$ then $\rho_{Y}(f(\vec{x}), f(\vec{x}')) \le q_{Y}^{j} \cdot M_{Y} \le p$, which shows that $f$ has a computable modulus of continuity. Similarly, it follows that $f$ satisfies the second requirement for computability.
\end{proof}

Summing up we obtain the following result.

\begin{theorem}\label{thm-phiteqcon}
Let $(X, D, Q_X)$ and $(Y, E, Q_Y)$ be decidable and well-covering computable digit spaces with approximable choice. Then 
\[
\CC_{\FF(X, Y)} = \UC(X, Y).
\] 
\end{theorem}

\section{Compact-valued functions}\label{sec-prop}

After having discussed how to obtain co-inductive characterisations of product and hyperspaces, as well as of uniformly continuous functions, which are the usual morphisms in the category of compact Hausdorff spaces, we will now investigate compact-valued maps in this framework. 

A central goal is to derive the functoriality of the compact hyperspace construction within the scope of co-inductive characterisations of spaces and morphisms. That is, we show that for maps $f \in \CC_{\FF(X, Y)}$, $\KKK(f) \in \CC_{\FF(\KKK(X), \KKK(Y))}$, where for compact sets $K$, $\KKK(f)(K) := f[K]$.  Classically, this is a well known fact. However, the proofs given here are new, and, as has already been pointed out, algorithms computing the action of the functor on morphisms can be extracted in case that the underlying IFS are covering, compact and weakly hyperbolic; similarly in the other cases considered in this section.

In a first step we will show that the continuous image of a compact set is compact again. As follows from the definition, the algorithm evaluating a function in $\CC_{\FF(X, Y)}$ on a given tuple of inputs uses the information coming with the components of the input tuple only in as much as it is needed, and also by dealing with each input component separately, not in parallel. This property seems to be too weak to be able to lift the function to the hyperspace of non-empty compact sets. The problem is that in the hyperspace, the compact sets are approximated by iterating digit maps that operate on the hyperspace. By doing so, the elements of the hyperspace are considered as abstract objects, not as sets of points or tuples of such. The internal structure is invisible to the IFS on the hyperspace.

In the subsequent Proposition we get around the problem by dealing with compact sets that are cubes, but without such information being available it is not clear how to proceed. In the general case we therefore restrict ourselves to unary maps. Multi-ary maps $\fun{f}{X^{n}}{Y}$ are then subsumed by using the IFS $(X^{n}, D^{\times})$ instead of $(X, D)$.

\begin{proposition}\label{pn-cmorph}
Let $(X, D), (Y, E)$ be compact IFS, $f \in \CC_{\FF(X, Y)}$ and $K_1, \ldots, K_{\ar{f}} \in \CC_{\KKK(X)}$. Then $f[K_1 \times \cdots \times K_{\ar{f}}] \in \CC_{\KKK(Y)}$. 
\end{proposition}
\begin{proof} 
For the co-inductive proof we need to derive a stronger statement.
Let to this end $f_{1}, \ldots. f_{k} \in \CC_{\FF(X, Y)}$ and define
\[
(\bigcup_{i=1}^{k} f_{i})(\vec x) := \bigcup_{i=1}^{k} \{ f_{i}(\pr^{(k)}_{i}(\vec x)) \},
\]
for $\vec x \in \bigtimes_{i=1}^{k} X^{\ar{f_{i}}} $. Set
\[
\MM(X, Y) := \set{\bigcup_{i=1}^k f_i}{k \ge 1 \wedge (\forall 1 \le j \le k)\, f_j \in \CC_{\FF(X, Y)}}.
\]
In general, the maps in $\MM(X, Y)$ will be multi-valued. For such maps $\mfun{h}{X}{Y}$ and $K \in \KKK(X)$, $h[K] := \bigcup \set{h(x)}{x \in K}$.

Let
\[
H := \set{g[K_1\times \cdots \times K_{\ar{g}}]}{g \in \MM(X, Y) \wedge K_1, \ldots, K_{\ar{g}} \in \CC_{\KKK(X)}}. 
\]
We will show that $H \subseteq \CC_{\KKK(Y)}$. For $\ZZ \subseteq \KKK(Y)$ define 
\[
\Omega_{\KKK(Y)}(\ZZ) := \set{M \subseteq Y}{(\exists \bar{d} \in \KKK(D)) (\exists M_1, \ldots, M_{\ar{\bar{d}}} \in \ZZ)\, M = \bar{d}(M_1, \ldots, M_{\ar{\bar{d}}})}.
\]
Then $\CC_{\KKK(Y)} = \nu \Omega_{\KKK(Y)}$. Hence, we must prove that $H \subseteq \Omega_{\KKK(Y)}(H)$.  For $g \in \MM(X, Y)$ let
\[
H(g) := \set{g[K_1 \times \cdots \times K_{\ar{g}}]}{K_1, \ldots, K_{\ar{g}} \in \CC_{\KKK(X)}}.
\]
We first show that for all $f \in \CC_{\FF(X, Y)}$, $H(f) \subseteq \Omega_{\KKK(Y)}(H)$.

Since $\CC_{\FF(X, Y)} = \JJJ^{X,Y}(\CC_{\FF(X, Y)}) = \mu \Phi^{X,Y}(\CC_{\FF(X, Y)})$, we can use induction to this aim. Set
\[
G := \set{f \in \FF(X, Y)}{H(f) \subseteq \Omega_{\KKK(Y)}(H)},
\]
then we need derive that
\[
\Phi_{X,Y}(\CC_{\FF(X,Y)})(G) \subseteq G.
\]
We have to consider the two cases:
\begin{itemize}
\item[(W)] If $e \in E$ and $h \in \CC_{\FF(X, Y)}$, then $H(e \circ h) \subseteq \Omega_{\KKK(Y)}(H)$.

\item[(R)] If $f \in \FF(X, Y)$ and $1 \le i \le \ar{f}$ so that for all $d \in D$, $H(f \circ d^{(i, \ar{f})}) \subseteq \Omega_{\KKK(Y)}(H)$, then $H(f) \subseteq \Omega_{\KKK(Y)}(H)$.

\end{itemize}

(W) Assume that $e \in E$  and $h \in \CC_{\FF(X, Y)}$, and let $K_{1}, \ldots, K_{\ar{h}} \in \CC_{\KKK(X)}$. Then we need to prove that $(e \circ h)[K_{1} \times \cdots \times K_{\ar{h}}] \in \Omega_{\KKK(Y)}(H)$.

We have that
\[
(e \circ h)[K_{1} \times \cdots \times K_{\ar{h}}] 
= e[h[K_{1} \times \cdots \times K_{\ar{h}}]]
= [e](h[K_{1} \times \cdots \times K_{\ar{h}}]).
\]
Since, $h \in \CC_{\FF(X, Y)}$ and $K_{1}, \ldots, K_{\ar{h}} \in \CC_{\KKK(X)}$, it follows that $h[K_{1} \times \cdots \times K_{\ar{h}}] \in H$. Therefore, $(e \circ h)[K_{1} \times \cdots \times K_{\ar{h}}] \in \Omega_{\KKK(Y)}(H)$.

(R) Let $f \in \FF(X, Y)$ and $1 \le i \le \ar{f}$ so that for all $d \in D$, $H(f \circ d^{(i, \ar{f})}) \subseteq \Omega_{\KKK(Y)}(H)$. Moreover, for $1 \le \kappa \le \ar{f}$, let $K_\kappa \in \CC_{\KKK(X)}$. Then there exist $d_1, \ldots, d_{r} \in D$ and $N_1, \ldots, N_r\in \CC_{\KKK(X)}$ so that $K_i = [d_1, \ldots, d_{r}](N_1, \ldots, N_r)$. It follows that
\begin{align*}
&f[K_1 \times \cdots \times K_{\ar{f}}]\\
&\quad= f[K_1\times \cdots \times K_{i-1} \times
[d_1, \ldots, d_r](N_1, \ldots, N_r) \times K_{i+1} \times \cdots \times K_{\ar{f}}] \\
&\quad= f[K_1 \times \ldots \times K_{i-1} \times
(\bigcup_{\kappa=1}^{r} d_\kappa[N_\kappa]) \times K_{i+1} \times \cdots K_{\ar{f}}] \\
&\quad= \bigcup_{\kappa=1}^{r} f[K_1 \times \cdots \times K_{i-1} \times d_\kappa[N_\kappa] \times K_{i+1} \times \cdots K_{\ar{f}}] \\
&\quad= \bigcup_{\kappa=1}^{r} (f \circ d_\kappa^{(i, \ar{f})})[K_1 \times \cdots \times K_{i-1} \times N_\kappa \times K_{i+1} \times \cdots \times K_{\ar{f}}].
\end{align*}

Since $d_\kappa \in D$, for $1 \le \kappa \le r$, it follows with our assumption that there are $e^{(\kappa)}_1, \ldots, e^{(\kappa)}_{s_\kappa} \in E$ and $M^{(\kappa)}_1, \ldots, M^{(\kappa)}_{s_\kappa} \in H$ such that
\begin{equation*}
\begin{split}
(f \circ d_\kappa^{(i, \ar{f})})[K_1 \times \cdots \times K_{i-1} \times N_\kappa \times K_{i+1} \times \cdots \times K_{\ar{f}}] \hspace{2cm} \mbox{} \\ = [e^{(\kappa)}, \ldots, e^{(\kappa)}_{s_\kappa}](M^{(\kappa)}_1, \ldots, M^{(\kappa)}_{s_\kappa}).
\end{split}
\end{equation*}
Thus,
\[
\bigcup_{\kappa=1}^{r} (f \circ d_\kappa^{(i, \ar{f})})[K_1 \times \cdots \times K_{i-1} \times N_\kappa \times K_{i+1} \times \cdots \times K_{\ar{f}}] = \bigcup_{\kappa=1}^r  \bigcup_{\sigma_{\kappa}=1}^{n_\kappa} e^{(\kappa)}_{\sigma_\kappa}[M^{(\kappa)}_{\sigma_\kappa}].
\]
Note that the maps $e^{(\kappa)}_{\sigma_\kappa}$ with $1 \le \kappa \le r$ and $1 \le \sigma_\kappa \le n_\kappa$ need not be pairwise distinct. Let $e_1, \ldots, e_\ell \in E$ be pairwise distinct so that $\{ e_1, \ldots, e_\ell \} = \bigcup_{\kappa=1}^r \{ e^{(\kappa)}_1, \ldots, e^{(\kappa)}_{n_\kappa} \}$. Then
\[ 
\bigcup_{\kappa=1}^r  \bigcup_{\sigma_{\kappa}=1}^{n_\kappa} e^{(\kappa)}_{\sigma_\kappa}[M^{(\kappa)}_{\sigma_\kappa}]
= [e_1, \ldots, e_\ell](L_1, \ldots, L_r),
\]
where for $1 \le \iota \le \ell$, $L_\iota := \bigcup \set{M^{(\kappa)}_{\sigma_\kappa}}{1 \le \kappa \le r \wedge 1 \le \sigma_\kappa \le n_\kappa \wedge e^{(\kappa)}_{\sigma_\kappa} = e_\iota}$. Since $M^{(\kappa)}_{\sigma_\kappa} = h^{(\kappa)}_{\sigma_\kappa}[\bigtimes_{\iota_{\kappa, \sigma_{\kappa}}=1}^{\ar{h^{(\kappa)}_{\sigma_{\kappa}}}} \widetilde{M}^{(\kappa, \sigma_{\kappa})}_{\iota_{\kappa, \sigma_{\kappa}}}]$, for some $h^{(\kappa)}_{\sigma_\kappa} \in \MM(X,Y)$ and $\widetilde{M}^{(\kappa, \sigma_{\kappa})}_{\iota_{\kappa, \sigma_{\kappa}}} \in \CC_{\KKK(X)}$, we have
\begin{equation}\label{eq-multi}
\begin{split}
L_\iota = \hspace{12.78cm} \\
 (\bigcup\nolimits_{1 \le \kappa \le r, 1 \le \sigma_\kappa \le n_\kappa: \KKK(e^{(\kappa)}_{\sigma_\kappa}) = \KKK(e_\iota)} h^{(\kappa)}_{\sigma_\kappa})[\bigtimes\nolimits_{1 \le \kappa \le r, 1 \le \sigma_\kappa \le n_\kappa: \KKK(e^{(\kappa)}_{\sigma_\kappa}) = \KKK(e_\iota)} \bigtimes\nolimits_{\iota_{\kappa, \sigma_{\kappa}}=1}^{\ar{h^{(\kappa)}_{\sigma_{\kappa}}}} \widetilde{M}^{(\kappa, \sigma_{\kappa})}_{\iota_{\kappa, \sigma_\kappa}}].
\end{split}
\end{equation}
It follows that $L_\iota \in H$ and thus $f[K_1 \times \cdots \times K_{\ar{f}}] \in \Omega_{\KKK(Y)}(H)$.

Observe that because of the possibility that the $e^{(\kappa)}_{\sigma_{\kappa}}$ are not pairwise distinct and in particular Line~(\ref{eq-multi}) we have to consider multi-valued maps in this proof.

Finally, we let $g \in \MM(X, Y)$. It remains to show that $H(g) \subseteq \Omega_{\KKK(Y)}(H)$.  Let to this end  $f_1, \ldots, f_k \in \CC_{\FF(X, Y)}$ with $g = f_1 \cup \cdots \cup f_k$, and $K^{(\kappa)}_{\sigma_\kappa} \in \CC_{\KKK(X)}$, for $1 \le \kappa \le k$ and $1 \le \sigma_\kappa \le \ar{f_\kappa}$.

As we have just seen, $f_\kappa[K^{(\kappa)}_1 \times \cdots \times K^{(\kappa)}_{\ar{f_\kappa}}] \in \Omega_{\KKK(Y)}(H)$. Hence, there are $e^{(\kappa)}_1, \ldots, e^{(\kappa)}_{r_\kappa} \in E$ and $M^{(\kappa)}_1, \ldots, M^{(\kappa)}_{r_\kappa} \in H$, for $1 \le \kappa \le k$, with
\[
f_\kappa[K^{(\kappa)}_1 \times \cdots \times K^{(\kappa)}_{\ar{f_\kappa}}] = [e^{(\kappa)}_1, \ldots, e^{(\kappa)}_{r_\kappa}](M^{(\kappa)}_1, \ldots, M^{(\kappa)}_{r_\kappa}).
\]
Again the $e^{(\kappa)}_{\iota_\kappa}$ with $1 \le \kappa \le k$ and $1 \le \iota_\kappa \le r_\kappa$ need not be pairwise distinct. By proceeding as in the previous proof step we can find pairwise distinct  $e'_1, \ldots, e'_\ell \in E$  and $L_1, \ldots, L_\ell \in H$ so that  $\{ e'_1, \ldots, e'_\ell \} = \{\, e^{(\kappa)}_{\iota_\kappa} \mid 1 \le \kappa \le k \wedge 1 \le \iota_\kappa \le r_\kappa \,\}$ and
\begin{align*}
g[K^{(1)}_1 \times \cdots \times K^{(k)}_{\ar{f_k}}] 
&= \bigcup_{\kappa=1}^k f_\kappa[K^{(\kappa)}_1 \times \cdots \times K^{(\kappa)}_{\ar{f_\kappa}}] \\
&= \bigcup_{\kappa=1}^k [e^{(\kappa)}_1, \ldots, e^{(\kappa)}_{r_\kappa}](M^{(\kappa)}_1, \ldots, M^{(\kappa)}_{r_\kappa}) \\
&= [e'_1, \ldots, e'_\ell](L_1, \ldots, L_\ell),
\end{align*}
showing that  $g[K^{(1)}_1 \times \cdots \times K^{(k)}_{\ar{f_k}}] \in \Omega_{\KKK(Y)}(H)$.
\end{proof}

For the subsequent technical lemma we extend the union taking operation to compact-valued maps. For  Hausdorff spaces $X, Y, Z$ such that $Z$ is compact, and maps $\fun{f}{X}{\KKK(Z)}$ and $\fun{g}{Y}{\KKK(Z)}$ define $f \cup g \colon X \times Y \to \KKK(Z)$ by
\[
(f \cup g)(x, y) := f(x) \cup g(y).
\]

\begin{lemma}\label{lem-cmorphtech}
Let $(X, D)$, $(Y, E)$ be compact IFS and 
\[
F := \set{\bigcup_{i=1}^{k} \KKK(f_{i})}{k \ge 1 \wedge (\forall 1 \le i \le k)\, f_{i} \in \CC^{(1)}_{\FF(X,Y)}}.
\]
Then $\JJJ^{\KKK(X), \KKK(Y)}(F)$ is closed under union.
\end{lemma}
\begin{proof}
Let 
\[
A := \set{f \in \FF(\KKK(X), \KKK(Y))}{(\forall g \in \JJJ^{\KKK(X), \KKK(Y)}(F))\, f \cup g \in \JJJ^{\KKK(X), \KKK(Y)}(F)}.
\]
We will prove by induction that $\JJJ^{\KKK(X), \KKK(Y)}(F) \subseteq A$, that is we have to show that 
\[
\Phi^{\KKK(X), \KKK(Y)}(F)(A) \subseteq A,
\]
which in turn means that we have to verify Rules~(W) and (R).

($\text{W}_{\text{main}}$) Let $\bar{e} \in \KKK(E)$ with $\bar{e} = [e_{1}, \ldots, e_{r}]$, for $e_{1}, \ldots, e_{r} \in E$. Moreover, let $f_{1}, \ldots, f_{r} \in F$. We must demonstrate that $\bar{e} \circ (f_{1} \times \cdots \times f_{r}) \in A$. To this end we have to show that for all $g \in \JJJ^{\KKK(X), \KKK(Y)}(F)$,  $\bar{e} \circ (f_{1} \times \cdots \times f_{r}) \cup g \in \JJJ^{\KKK(X), \KKK(Y)}(F)$. For $h \in \FF(\KKK(X), \KKK(Y))$ set
\[
B(h) := \set{g \in \FF(\KKK(X), \KKK(Y))}{h \cup g \in \JJJ^{\KKK(X), \KKK(Y)}(F)}.
\]
Let $h := \bar{e} \circ (f_{1} \times \cdots \times f_{r})$. We prove that $\JJJ^{\KKK(X), \KKK(Y)}(F) \subseteq B(h)$, that is we have to demonstrate that 
\begin{equation}\label{eq-tech}
\Phi^{\KKK(X), \KKK(Y)}(F)(B(h)) \subseteq B(h),
\end{equation}
which is done by side induction.

($\text{W}_{\text{side}}$) Let $\hat{e} \in \KKK(E)$, say $\hat{e} = [e'_{1}, \ldots, e'_{s}]$ with $e'_{1}, \ldots, e'_{s} \in E$. Furthermore, let $g_{1}, \ldots, g_{s} \in F$. We need to show that
\[
(\bar{e} \circ (f_{1} \times \cdots \times f_{r})) \cup (\hat{e} \circ (g_{1} \times \cdots \times g_{s})) \in \JJJ^{\KKK(X), \KKK(Y)}(F).
\]
Note that $\bar{e} = \KKK(e_{1}) \cup \cdots \cup \KKK(e_{r})$ and analogously for $\hat{e}$. Then 
\begin{equation*}
\begin{split}
(\bar{e} \circ (f_{1} \times \cdots \times f_{r})) \cup (\hat{e} \circ (g_{1} \times \cdots \times g_{s})) = \hspace{4.5cm} \\
\KKK(e_{1}) \circ f_{1} \cup \cdots \cup \KKK(e_{r}) \circ f_{r} \cup \KKK(e'_{1}) \circ g_{1} \cup \cdots \cup \KKK(e'_{s}) \circ g_{s}.
\end{split}
\end{equation*}
Let $R := \set{(j_{1}, j_{2})}{1 \le j_{1} \le r \wedge 1 \le j_{2} \le s \wedge e_{j_{1}} = e'_{j_{2}}}$. Then we have that
\begin{align*}
&\KKK(e_{1}) \circ f_{1} \cup \cdots \cup \KKK(e_{r}) \circ f_{r} \cup \KKK(e'_{1}) \circ g_{1} \cup \cdots \cup \KKK(e'_{s}) \circ g_{s} \\
&\quad= 
[e_{i_{1}}, \ldots, e_{i_{k}}, e_{i_{k+1}}, \ldots, e_{i_{k'}}, e'_{i_{k'+1}}, \ldots, e'_{i_{\ell}}] \circ ((f_{i_{1}} \cup g_{i_{1}}) \times \cdots \times (f_{i_{k}} \cup g_{i_{k}}) \times \\
&\hspace{8cm} f_{i_{k+1}} \times \cdots \times f_{i_{k'}} \times g_{i_{k'+1}} \times \cdots \times g_{i_{\ell}}),
\end{align*}
where $1 \le k \le k' \le \ell$ and $i_{1}, \ldots, i_{\ell}, i'_{1}, \ldots, i'_{k} \in \NN$ so that
\begin{gather*}
R = \{ (i_{1}, i'_{1}), \ldots, (i_{k}, i'_{k}) \},  \\
\{ 1, \ldots, r \} \setminus \pr^{(2)}_{1}[R] = \{ i_{k+1}, \ldots, i_{k'} \}, \\
\{ 1, \ldots, s \} \setminus \pr^{(2)}_{2}[R] =\{ i_{k'+1}, \ldots, i_{\ell} \}.
\end{gather*}

\sloppy It follows that $e_{i_{1}}, \ldots, e_{i_{k}}, e_{i_{k+1}}, \ldots, e_{i_{k'}}, e'_{i_{k'+1}}, \ldots, e'_{i_{\ell}}$ are pairwise distinct and $f_{i_{1}} \cup g_{i_{1}}, \ldots, f_{i_{k}} \cup g_{i_{k}}, f_{i_{k+1}}, \ldots, f_{i_{k'}}, g_{i_{k'+1}}, \ldots, g_{i_{\ell}} \in F$. Thus, $(\bar{e} \circ (f_{1} \times \cdots \times f_{r})) \cup (\hat{e} \circ (g_{1} \times \cdots \times g_{s})) \in \JJJ^{\KKK(X), \KKK(Y)}(F)$, by Rule~(W).

($\text{R}_{\text{side}}$) Let $g \in \FF(\KKK(X), \KKK(Y))$ and $1 \le i \le \ar{g}$ so that $g \circ \bar{e}^{(i, \ar{g})} \in B(h)$, for all $\bar{e} \in \KKK(E)$. Then $h \cup (g \circ \bar{e}^{(i, \ar{g})}) \in \JJJ^{\KKK(X), \KKK(Y)}(F)$. Set $\hat{h} := h \cup g$. Obviously, $\ar{\hat{h}} = \ar{h} + \ar{g}$. For $j := \ar{h} + i$ we therefore have that $\hat{h} \circ \bar{e}^{(j, \ar{\hat{h}})} = h \cup (g \circ \bar{e}^{(i, \ar{g})})$. It follows that $\hat{h} \circ \bar{e}^{(j, \ar{\hat{h}})} \in \JJJ^{\KKK(X), \KKK(Y)}(F)$, for all $\bar{e} \in \KKK(E)$, from which we obtain with Rule~(R) that $\hat{h} \in \JJJ^{\KKK(X), \KKK(Y)}(F)$. Hence, $g \in B(h)$. This proves (\ref{eq-tech}).

($\text{R}_{\text{main}}$) It remains to verify Rule~(R) in the main induction. Let $f \in \FF(\KKK(X), \KKK(Y))$ and $1 \le i \le \ar{f}$ such that for all $\bar{d} \in \KKK(D)$, $f \circ \bar{d}^{(i, \ar{f})} \in A$. Thus, $(f \circ \bar{d}^{(i, \ar{f})}) \cup g \in \JJJ^{\KKK(X), \KKK(Y)}(F)$, for all $g \in \JJJ^{\KKK(X), \KKK(Y)}(F)$. Set $\tilde{h}_{g} := f \cup g$. Then $\ar{\tilde{h}_{g}} = \ar{f} + \ar{g}$ and for all $\bar{d} \in \KKK(D)$,
\[
\tilde{h}_{g} \circ \bar{d}^{(i, \ar{\tilde{h}_{g}})} = (f \circ \bar{d}^{(i, \ar{f})}) \cup g.
\]
Since $(f \circ \bar{d}^{(i, \ar{f})})\cup g \in \JJJ^{\KKK(X), \KKK(Y)}(F)$, by our assumption, we obtain with Rule~(R) that $\tilde{h}_{g} \in \JJJ^{\KKK(X), \KKK(Y)}(F)$. Thus $f \in A$.
\end{proof}
 
\begin{theorem}\label{thm-cfct}
Let $(X, D), (Y, E)$ be compact IFS. Then for all $f \in \CC^{(1)}_{\FF(X, Y)}$,  
\[
\KKK(f) \in \CC_{\FF(\KKK(X), \KKK(Y))}.
\]
\end{theorem}
\begin{proof}
Again we need to derive a stronger statement. Let 
\[
F := \set{\bigcup_{i=1}^k \KKK(f_i)}{k \ge 1 \wedge (\forall 1 \le i \le k)\, f_i \in \CC^{(1)}_{\FF(X, Y)}}.
\]
We will show that $F \subseteq \CC_{\FF(\KKK(X), \KKK(Y))}$. By using co-induction it is sufficient to prove that $F \subseteq \JJJ^{\KKK(X), \KKK(Y)}(F)$. 

Set 
\[
B := \set{f \in \FF(X, Y)}{\ar{f} = 1 \rightarrow \KKK(f) \in  \JJJ^{\KKK(X), \KKK(Y)}(F)}.
\]
Then we show first that $\CC_{\FF(X, Y)} \subseteq B$.  Because $\CC_{\FF(X, Y)} = \JJJ^{X, Y}(\CC_{\FF(X, Y)})$, we only need to derive that $\JJJ^{X, Y}(\CC_{\FF(X, Y)}) \subseteq B$. To achieve this we use induction, that is we show $\Phi^{X, Y}(\CC_{\FF(X, Y)})(B) \subseteq B$. Hereto we have to verify the following two rules:
\begin{itemize}

\item[(W)] If $e \in E$ and $h \in \CC_{\FF(X, Y)}$, then $e \circ h \in B$.

\item[(R)] If $f \in \FF(X, Y)$ and $1 \le i \le \ar{f}$ such that for all $d \in D$, $f \circ d^{(i, \ar{f})} \in B$, then $f \in B$.

\end{itemize}

(W) We have that $\KKK(e \circ h) = \KKK(e) \circ \KKK(h) = [e] \circ \KKK(h)$. Now, assume that $\ar{e \circ h} = 1$, then also $\ar{h} = 1$.
Since $\KKK(h) \in F$ and $[e] \in \KKK(E)$, it follows with Rule~(W) for $\JJJ^{\KKK(X), \KKK(Y)}(F)$ that $[e] \circ \KKK(h) \in \JJJ^{\KKK(X), \KKK(Y)}(F)$. Thus, $e \circ h \in B$ 

(R) Let $f \in \FF(X, Y)$ and $1 \le i \le \ar{f}$ so that for all $d \in D$, $f \circ d^{(i, \ar{f})} \in B$. Suppose that  $\ar{f \circ d^{(i, \ar{f})}} = 1$. Then $\ar{f} = 1$ and hence $i = 1$.
Moreover,  $\KKK(f \circ d) \in \JJJ^{\KKK(X), \KKK(Y)}(F)$. Now, let $\bar{d} \in \KKK(D)$ with $\bar{d} = [d_1, \ldots, d_s]$. Note that $[d_1, \ldots, d_s] = \KKK(d_1) \cup \cdots \cup \KKK(d_s)$. It follows that 
\[
\KKK(f) \circ [d_1, \ldots, d_s] 
= \KKK(f) \circ (\bigcup_{\kappa=1}^s \KKK(d_\kappa)) 
= \bigcup_{\kappa=1}^s \KKK(f) \circ \KKK(d_\kappa) 
= \bigcup_{\kappa=1}^s \KKK(f \circ d_\kappa).
\]

By our assumption, $f \circ {d_\kappa} \in B$. Hence, $\KKK(f \circ d_\kappa) \in \JJJ^{\KKK(X), \KKK(Y)}(F)$, for all $1 \le \kappa \le s$, as a consequence of which we obtain with Lemma~\ref{lem-cmorphtech} that $\bigcup_{\kappa=1}^s \KKK(f \circ d_\kappa) \in \JJJ^{\KKK(X), \KKK(Y)}(F)$. This shows that $\KKK(f) \circ \bar{d} \in \JJJ^{\KKK(X), \KKK(Y)}(F)$, for all $\bar{d} \in \KKK(D)$. By Rule~(R) for $\JJJ^{\KKK(X), \KKK(Y)}(F)$ we thus have that $\KKK(f) \in \JJJ^{\KKK(X), \KKK(Y)}(F)$, which means  that $f \in B$.

Finally, let $f_{1}, \ldots, f_{k} \in \CC^{(1)}_{\FF(X, Y)}$. Then $f_{1}, \ldots, f_{k} \in B$ and hence $\KKK(f_{j}) \in \JJJ^{\KKK(X), \KKK(Y)}(F)$, for $1 \le j \le k$. With Lemma~\ref{lem-cmorphtech} we therefore have that also $\bigcup_{j=1}^{k} \KKK(f_{j}) \in \JJJ^{\KKK(X), \KKK(Y)}(F)$, which shows that $F \subseteq \JJJ^{\KKK(X), \KKK(Y)}(F)$.
\end{proof}

The simplest compact sets are the singleton sets of points of the underlying space.
\begin{lemma}\label{lem-singl}
Let $(X, D)$ be a compact IFS. Then $\set{ \{ x \}}{x \in \CC_X} \subseteq  \CC_{\KKK(X)}$.
\end{lemma}
\begin{proof}
Let $\ZZ := \set{ \{ x \} }{x \in \CC_X}$. The statement will be derived by co-induction. To this end we have to show that $\ZZ \subseteq \Omega_{\KKK(X)}(\ZZ)$.

Let $x \in X$. Since $x \in \CC_X$,  there are $d \in D$ and $y \in \CC_X$ with $x = d(y)$. Then $\{ y \} \in \ZZ$ and 
\[
[d](\{ y \}) = d[\{ y \}] = \{ d(y) \} = \{ x \}. \qedhere
\]
\end{proof}

Define $\fun{\eta}{\CC_X}{\CC_{\KKK(X)}}$ by letting
\[
\eta(x) := \{ x \}.
\]

\begin{proposition}\label{pn-eta}
Let $(X, D)$ be a compact IFS. Then $\eta \in \CC_{\FF(X, \KKK(X))}$.
\end{proposition}
\begin{proof}
Again we use co-induction to derive the statement: We show that $\{ \eta \} \subseteq \JJJ^{X,\KKK(X)}(\{ \eta \})$. Note that for any $d \in D$, $[d] \in \KKK(D)$. Moreover, $[d] \circ \eta = \eta \circ d$. By Rule~(W) we have that $[d] \circ \eta \in \JJJ^{X,\KKK(X)}(\{ \eta \})$, i.e., $\eta \circ d \in \JJJ^{X,\KKK(X)}(\{ \eta \})$. Therefore $\eta \in \JJJ^{X,\KKK(X)}(\{ \eta \})$, by Rule~(R).
\end{proof}

We have already introduced the union taking operation for compact-valued maps.  
\begin{theorem}\label{thm-compval}
 Let $(X, D)$, $(Y,E)$ be extended IFS and $(Z, C)$ be a compact IFS. Then, for $f \in \CC_{\FF(X, \KKK(Z))}$ and $g \in  \CC_{\FF(Y, \KKK(Z))}$, $f \cup g \in \CC_{\FF(X \times Y, \KKK(Z))}$.
 \end{theorem}
 \begin{proof}
 The proof of this result essentially follows the lines of the proof of Lemma~\ref{lem-cmorphtech}. However, now we use strong co-induction. Let
 \[
 A := \set{f \cup g}{f \in \CC_{\FF(X, \KKK(Z))} \wedge g \in  \CC_{\FF(Y, \KKK(Z))}}. 
 \]
 Then we need to show that $A \subseteq \JJJ^{X\times Y, \KKK(Z)}(A \cup \CC_{\FF(X \times Y, \KKK(Z))})$. Set
 \[
 B:= \set{f \in  \FF(X, \KKK(Z))}{(\forall g \in  \CC_{\FF(Y, \KKK(Z))})\, f \cup g \in \JJJ^{X\times Y, \KKK(Z)}(A \cup \CC_{\FF(X \times Y, \KKK(Z))})}.
 \]
We show that $\CC_{\FF(X, \KKK(Z))} \subseteq B$. Since $ \CC_{\FF(X, \KKK(Z))} = \JJJ^{X, \KKK(Z)}(\CC_{\FF(X, \KKK(Z))})$, it suffices to prove that $\JJJ^{X, \KKK(Z)}(\CC_{\FF(X, \KKK(Z))}) \subseteq B$. By induction we show that $\Phi^{X, \KKK(Z)}(\CC_{\FF(X, \KKK(Z))})(B) \subseteq B$, which means that we have to verify Rules~(W) and (R).

($\text{W}_\text{main}$) Let $\bar{c} \in \KKK(C)$ with $\bar{c} = [c_1, \ldots, c_r]$, for $c_1, \ldots, c_r \in C$. Moreover, let $f_1, \ldots, f_r \in \CC_{\FF(X, \KKK(Z))}$. We must demonstrate that $\bar{c} \circ (f_1 \times \cdots \times f_r) \in B$. To this end we have to show that for all $g \in \CC_{\FF(Y, \KKK(Z))}$, $(\bar{c} \circ (f_1 \times \cdots f_r)) \cup g \in  \JJJ^{X\times Y, \KKK(Z)}(A \cup \CC_{\FF(X \times Y, \KKK(Z))})$. For $h \in \FF(X, \KKK(Z))$ set 
\[
H(h) := \set{g \in \FF(Y, \KKK(Z))}{h \cup g \in \JJJ^{X\times Y, \KKK(Z)}(A \cup \CC_{\FF(X \times Y, \KKK(Z))})}.
\]
Let $h := \bar{c} \circ (f_1 \times \cdots \times f_r)$. We prove that $\CC_{\FF(Y, \KKK(Z))} \subseteq H(h)$, for which it suffices to demonstrate that
\begin{equation}\label{eq-compval}
\Phi^{Y, \KKK(Z)}(\CC_{\FF(Y, \KKK(Z))})(H(h)) \subseteq H(h),
\end{equation}
which is done by side induction. 
 
($\text{W}_\text{side}$) Let $\hat{c} \in \KKK(C)$, say $\hat{c} = [c'_1, \ldots, c'_s]$ with $c'_1, \ldots, c'_s \in C$. Furthermore, let $g_1, \ldots, g_s \in \CC_{\FF(Y, \KKK(Z))}$. We need to show that
\[
(\bar{c} \circ (f_1 \times \cdots \times f_r)) \cup (\hat{c} \circ (g_1 \times \cdots \times g_s)) \in \JJJ^{X\times Y, \KKK(Z)}(A \cup \CC_{\FF(X \times Y, \KKK(Z))}).
\]
Note again that $\bar{c} = \KKK(c_{1}) \cup \cdots \cup \KKK(c_{r})$, and analogously for $\hat{c}$. Then 
\begin{equation*}
\begin{split}
(\bar{c} \circ (f_1 \times \cdots \times f_r)) \cup (\hat{c} \circ (g_1 \times \cdots \times g_s)) = \hspace{4.5cm} \\
\KKK(c_{1}) \circ f_{1} \cup \cdots \cup \KKK(c_{r}) \circ f_{r} \cup \KKK(c'_{1}) \circ g_{1} \cup \cdots \cup \KKK(c'_{s}) \circ g_{s}.
\end{split}
\end{equation*}
 
As in the proof of Lemma~\ref{lem-cmorphtech} it follows that there are pairwise distinct $c''_{i_1}$, \ldots, $c''_{i_k}$, $c''_{i_{k+1}}$, \ldots, $c''_{i_{k'}}$, \ldots, $c''_{i_{k'+1}}$, $c''_\ell \in \{ c_1, \ldots, c_r, c'_1, \ldots, c'_s \}$ so that (up to some reshuffling of arguments)
\begin{equation*}\begin{split}
\KKK(c_{1}) \circ f_{1} \cup \cdots \cup \KKK(c_{r}) \circ f_{r} \cup \KKK(c'_{1}) \circ g_{1} \cup \cdots \cup \KKK(c'_{s}) \circ g_{s} \hspace{4.6cm}\mbox{}\\
= [c''_{i_1}, \ldots, c''{i_\ell}] \circ  ((f_{i_1} \cup g_{i_1}) \times \cdots \times (f_{i_k} \cup g'_{i_k}) \times f_{i_{k+1}} \times \cdots \times f_{i_{k'}} \times g_{i_{k'+1}} \times \cdots \times g_{i_\ell}).
\end{split}\end{equation*}
Then $[c''_{i_1}, \ldots, c''{i_\ell}] \in \KKK(C)$, $f_{i_1} \cup g_{i_1}, \ldots, f_{i_k} \cup g_{i_k} \in A$, and $f_{i_{k+1}}, \ldots, f_{i_{k'}}, g_{i_{k'+1}}, \ldots, g_{i_\ell} \in \CC_{\FF(X \times Y, \KKK(Z))}$, where we identify maps $\fun{f}{X^{n}}{Z}$ with $f \circ (\pr^{(n+1)}_{1} \times \cdots \times \pr^{(n+1)}_{n}) \colon X^{n} \times Y \to Z$, and similarly for maps $\fun{g}{Y^{m}}{Z}$. Thus, $(\bar{c} \circ (f_1 \times \cdots \times f_r)) \cup (\hat{c} \circ (g_1 \times \cdots \times g_s)) \in \JJJ^{X\times Y, \KKK(Z)}(A \cup \CC_{\FF(X \times Y, \KKK(Z))})$ by Rule~(W).

Note that the remaining proof steps are essentially the same as in the proof of Lemma~\ref{lem-cmorphtech}. We include them only for completeness reasons.

($\text{R}_\text{side}$) Let $g \in \FF(Y, \KKK(Z))$ and $1 \le i \le \ar{g}$ so that $g \circ e^{(i, \ar{g})} \in B(h)$, for all $e \in E$. Then $h \cup (g \circ e^{(i, \ar{g})}) \in \JJJ^{X\times Y, \KKK(Z)}(F \cup \CC_{\FF(X \times Y, \KKK(Z))})$. Let $\hat{h} := h \cup g$. Then $\ar{\hat{h}} = \ar{h} + \ar{g}$. For $j := \ar{h} + i$ we therefore have that $\hat{h} \circ e^{(j, \ar{\hat{h}})} = h \cup (g \circ e^{(i, \ar{g})})$. It follows that $\hat{h} \circ e^{(j, \ar{\hat{h}})} \in \JJJ^{X\times Y, \KKK(Z)}(F \cup \CC_{\FF(X \times Y, \KKK(Z))})$, for all $e \in E$, from which we obtain by Rule~(R) that $\hat{h} \in \JJJ^{X\times Y, \KKK(Z)}(F \cup \CC_{\FF(X \times Y, \KKK(Z))})$. Hence, $g \in B(h)$. This proves (\ref{eq-compval}).

($\text{R}_\text{main}$) It remains to verify Rule~(R) in the main induction. Let $f \in A$ and $1 \le i \le \ar{f}$ such that for all $d \in D$, $f \circ d^{(i, \ar{f})} \in A$. Thus, $(f \circ d^{(i, \ar{f})}) \cup g \in \JJJ^{X\times Y, \KKK(Z)}(F \cup \CC_{\FF(X \times Y, \KKK(Z))})$, for all $g \in \CC_{\FF(Y, \KKK(Z))}$. Set  $\bar{h}_g := f \cup g$. Then $\ar{\bar{h}_g} = \ar{f} + \ar{g}$ and for all $d \in D$, 
\[
\bar{h}_g \circ d^{(i, \ar{\bar{h}_g})} = (f \circ d^{(i, \ar{f})}) \cup g.
\]
Since $(f \circ d^{(i, \ar{f})}) \cup g \in  \JJJ^{X\times Y, \KKK(Z)}(F \cup \CC_{\FF(X \times Y, \KKK(Z))})$ by our assumption, we obtain with Rule~(R) that $\bar{h}_g \in \JJJ^{X\times Y, \KKK(Z)}(F \cup \CC_{\FF(X \times Y, \KKK(Z))})$. Thus $ f \in A$.
\end{proof}

Since the identity on any extended IFS is in $\CC_{\FF(X,X)}$, by Lemma~\ref{lem-basics}(\ref{lem-id}), it follows in particular that the binary operation of taking unions is in $\CC_{\FF(\KKK(X) \times \KKK(X), \KKK(X))}$.

\begin{corollary}\label{cor-un}
Let $(X, D)$ be a compact IFS. Then the following two statements hold:
\begin{enumerate}

\item\label{cor-un-1}
$\CC_{\KKK(X)}$ is closed under taking finite unions, that is, for all $K, M \subseteq X$,
\[
K, M \in \CC_{\KKK(X)} \Rightarrow K \cup M \in \CC_{\KKK(X)}.
\]

\item\label{cor-un-2}
$\cup \in \CC_{\FF(\KKK(X) \times \KKK(X), \KKK(X))}$.

\end{enumerate}
\end{corollary}
The first statement follows from the second one by Proposition~\ref{pn-eval}.

\section{Michael's Theorem}\label{sec-mich}

In his seminal 1951 paper on spaces of subsets~\cite{mi}, E.~Michael showed that the union of all sets in a compact set of compact sets is compact again. We will reprove this result in a non-topological way by using only the co-inductive characterisations of the spaces involved.

The difficulty we have to overcome herewith is that, as we have seen in Section~\ref{sec-compact}, in general the canonical maps in $\KKK^2(D)$ are no longer of type $\KKK^2(X)^r \to \KKK^2(X)$, even if all maps in $D$ are unary. Thus, $(\KKK^2(X), \KKK^2(D))$ will not be an extended IFS any more. 

Let $[[d^{(1)}_1, \ldots, d^{(1)}_{r_1}], \ldots, [d^{(n)}_1, \ldots, d^{(n)}_{r_n}]] \in \KKK^2(D)$ and $\KK_1, \ldots, \KK_n \in \KKK^2(X)$. Then we have
\begin{align*}
&\bigcup [[d^{(1)}_1, \ldots, d^{(1)}_{r_1}], \ldots, [d^{(n)}_1, \ldots, d^{(n)}_{r_n}]](\KK_1, \ldots, \KK_n) \\
&\quad= \bigcup \bigcup_{\kappa=1}^n [d^{(\kappa)}_1, \ldots, d^{(\kappa)}_{r_\kappa}][\KK_\kappa] \\
&\quad= \bigcup \bigcup_{\kappa=1}^n \set{\bigcup_{{\sigma_\kappa}=1}^{r_\kappa} d^{(\kappa)}_{\sigma_\kappa}[K^{(\kappa)}_{\sigma_\kappa}]}{(K^{(\kappa)}_1, \ldots, K^{(\kappa)}_{r_\kappa}) \in \KK_\kappa}  \\
&\quad= \bigcup_{\kappa=1}^n \bigcup \set{\bigcup_{{\sigma_\kappa}=1}^{r_\kappa} d^{(\kappa)}_{\sigma_\kappa}[K^{(\kappa)}_{\sigma_\kappa}]}{(K^{(\kappa)}_1, \ldots, K^{(\kappa)}_{r_\kappa}) \in \KK_\kappa} \\
&\quad=  \bigcup_{\kappa=1}^n \bigcup_{{\sigma_\kappa}=1}^{r_\kappa} \bigcup \set{d^{(\kappa)}_{\sigma_\kappa}[K^{(\kappa)}_{\sigma_\kappa}]}{K^{(\kappa)}_{\sigma_\kappa} \in \pr^{({r_\kappa})}_{\sigma_\kappa}[\KK_\kappa]} \\
&\quad=  \bigcup_{\kappa=1}^n \bigcup_{{\sigma_\kappa}=1}^{r_\kappa} \bigcup \KKK(d^{(\kappa)}_{\sigma_\kappa})[\pr^{({r_\kappa})}_{\sigma_\kappa}[\KK_\kappa]].
\end{align*}
Now, let $e_1, \ldots, e_m \in D$ be pairwise distinct so that 
\[
\{\KKK( e_1), \ldots, \KKK(e_m )\} = \set{\KKK(d^{(\kappa)}_{\sigma_\kappa})}{1 \le \sigma_\kappa \le r_\kappa \wedge 1 \le \kappa \le n}.
\]
Moreover, set
\[
\MM_\iota := \bigcup \set{\pr^{(r_\kappa)}_{\sigma_\kappa}[\KK_\kappa]}{\KKK(d^{(\kappa)}_{\sigma_\kappa}) = \KKK(e_\iota) \wedge 1 \le \sigma_\kappa \le r_\kappa \wedge 1 \le \kappa \le n},
\]
for $1 \le \iota \le m$. Then
\begin{equation*}
\bigcup_{\kappa=1}^n \bigcup_{{\sigma_\kappa}=1}^{r_\kappa} \KKK(d^{(\kappa)}_{\sigma_\kappa})[\pr^{({r_\kappa})}_{\sigma_\kappa}[\KK_\kappa]] 
= \bigcup_{\iota=1}^m \KKK(e_\iota)[\MM_\iota] 
= [\KKK(e_1), \ldots, \KKK(e_m)](\MM_1, \ldots, \MM_m).
\end{equation*}
Let us summarise what we have just seen.
\begin{lemma}\label{lem-unioneq}
Let $(X, D)$ be an IFS. Then, given $[[d^{(1)}_1, \ldots, d^{(1)}_{r_1}], \ldots, [d^{(n)}_1, \ldots, d^{(n)}_{r_n}]] \in \KKK^2(D)$, we can compute $e_1, \ldots, e_m \in D$ so that when given in addition $\KK_1, \ldots, \KK_n \in \KKK^2(X)$, we can furthermore define $\MM_1, \ldots, \MM_m \in \KKK^2(X)$ such that
\begin{equation*}
\bigcup [[d^{(1)}_1, \ldots, d^{(1)}_{r_1}], \ldots, [d^{(n)}_1, \ldots, d^{(n)}_{r_n}]](\KK_1, \ldots, \KK_n) 
= \bigcup [\KKK(e_1), \ldots, \KKK(e_m)](\MM_1, \ldots, \MM_m).
\end{equation*}
\end{lemma}

This gives us a hint on how to deal with unions over compact collections of compact sets. 

For $\KK, \MM \in \KKK^2(X)$, we write $\KK \equ \MM$, , if 
\[
\bigcup \KK = \bigcup \MM.
\]
Moreover, we set $\KKK_{D}:= \set{\KKK(d)}{d \in D}$. Note that $(\KKK^{2}(X), \KKK(\KKK_{D}))$ is a compact IFS, and let
\begin{equation*}\begin{split}
\Omega^{\pair{2}}_X(\ZZZ) := \{\, \KK \in \KKK^{2}(X) \mid (\exists d_{1}, \ldots, d_{s} \in D) \hspace{5cm} \\
 (\exists \MM_1, \ldots, M_m \in \ZZZ)\, \KK \equ [\KKK(d_{1}), \ldots, \KKK(d_{s})](\MM_1, \ldots, \MM_s)\,\},
\end{split}\end{equation*}
for $\ZZZ \subseteq \KKK^2(X)$. We define $\CC'_{\KKK^{2}(X)}:= \nu \Omega^{\pair{2}}_X$.

\begin{lemma}\label{lem-mich}
Let $(X, D)$ be a compact covering IFS. Then $\KKK^{2}(X) = \CC'_{\KKK^{2}(X)}$.
\end{lemma}
\begin{proof}
We have that $\CC'_{\KKK^{2}(X)} \subseteq \KKK^{2}(X)$ by definition. For the converse inclusion note that it follows as in the proof of Proposition~\ref{pn-comifs}(\ref{pn-comifs-1}) that
\[
\KKK^{2}(X) = \bigcup \set{\range(\bar{d})}{\bar{d} \in  \KKK^{2}(D)}.
\]
Therefore, if $\KK \in \KKK^{2}(X)$, then there is some $[[d^{(1)}_1, \ldots, d^{(1)}_{r_1}], \ldots, [d^{(n)}_1, \ldots, d^{(n)}_{r_n}]] \in \KKK^2(D)$ and there are $\KK_{1}, \ldots, \KK_{n} \in \KKK^{2}(X)$ so that $\KK = [[d^{(1)}_1, \ldots, d^{(1)}_{r_1}], \ldots, [d^{(n)}_1, \ldots, d^{(n)}_{r_n}]](\KK_{1}, \ldots, \KK_{n})$. Because of Lemma~\ref{lem-unioneq} we can now compute $e_{1}, \ldots, e_{m} \in D$ and define $\MM_{1}, \ldots, \MM_{m} \in \KKK^{2}(X)$ so that $\KK \equ [\KKK(e_{1}), \ldots, \KKK(e_{m})](\MM_{1}, \ldots, \MM_{m})$. Thus, $\KK \in \Omega_{X}^{\pair{2}}(\KKK^{2}(X))$. With co-induction it follows that $\KKK^{2}(X) \subseteq \CC'_{\KKK^{2}(X)}$.
\end{proof}

Let $d_{1}, \ldots, s_{s} \in D$ and note that
\begin{equation}\label{eq-mapstar}
\begin{split}
\bigcup [\KKK(d_{1}), \ldots, \KKK(d_{s})](\MM_{1}, \ldots, \MM_{s})
&= \bigcup \bigcup_{\kappa=1}^{s} \KKK(d_{\kappa})[\MM_{\kappa}] \\
&= \bigcup_{\kappa=1}^{s} \bigcup_{M \in \MM_{\kappa}} \KKK(d_{\kappa})(M) \\
&= \bigcup_{\kappa=1}^{s} \bigcup_{M \in \MM_{\kappa}} d_{\kappa}[M]  \\
&= \bigcup_{\kappa=1}^{s} d_{\kappa} [\bigcup \MM_{\kappa}]  \\
&= [d_{1}, \ldots, d_{s}](\bigcup \MM_{1}, \ldots, \bigcup \MM_{s}).
\end{split}
\end{equation}
It follows that $[\KKK(d_{1}), \ldots, \KKK(d_{s})]$ respects $\equ$ and can hence be lifted to the quotient $\KKK^{2}(X)/\equ$. We denote the lifted map by $[\KKK(d_{1}), \ldots, \KKK(d_{s})]^{*}$ and set
\[
\KKK^{*}(\KKK_{D}):= \{\, [\KKK(d_{1}), \ldots, \KKK(d_{s})]^{*} \mid s >0 \wedge d_{1}, \ldots, d_{s} \in D \,\}.
\]

\begin{lemma}\label{lem-quot}
Let $(X, D)$ be a compact covering IFS. Then also $(\KKK^{2}(X)/\equ, \KKK^{*}(\KKK_{D}))$ is a compact covering IFS. Moreover, $\CC_{\KKK^{2}(X)/\equ} = \CC'_{\KKK^{2}(X)}/\equ$.
\end{lemma}
\begin{proof}
$\KKK^{2}(X)/\equ$ is a topological space with the quotient topology. Since $\bigcup \{ K \} = K$, for every $K \in \KKK(X)$, $\fun{\bigcup}{\KKK^{2}(X)}{\KKK(X)}$ is an onto map, which is well known, to be continuous~\cite[Corollary 7.2.4]{kt}. Since $X$ is compact, also $\KKK(X)$, $\KKK^{2}(X)$ and $\KKK^{2}(X)/\equ$ are compact. Moreover, $\KKK(X)$ and $\KKK^{2}(X)$ are Hausdorff spaces. Thus, $\bigcup$ is a closed mapping and, as a consequence, the Vietoris topology on $\KKK(X)$ is equivalent to the quotient topology with respect to $\bigcup$~\cite[Theorem 9.2]{wil}. It follows that the quotient topology on $\KKK^{2}(X)/\equ$ is Hausdorff.

By Lemma~\ref{lem-unioneq}, $(\KKK^{2}(X)/\equ, \KKK^{*}(\KKK_{D}))$ is covering. The remaining statement follows with co-induction.

Let $\KK \in \CC_{\KKK^{2}(X)/\equ}$. Then there are $d_{1}, \ldots, d_{s} \in D$ and $\MM_{1}, \ldots, \MM_{s} \in \KKK^{2}(X)$, so that 
\[
[\KK]_{\equ} 
= [\KKK(d_{1}), \ldots, \KKK(d_{s})]^{*}([\MM_{1}]_{\equ}, \ldots, [\MM_{s}]_{\equ})
= [[\KKK(d_{1}), \ldots, \KKK(d_{s})](\MM_{1}, \ldots, \MM_{s})]_{\equ}.
\]
Hence, $\KK \equ [\KKK(d_{1}), \ldots, \KKK(d_{s})](\MM_{1}, \ldots, \MM_{s})$, 
which shows that $\CC_{\KKK^{2}(X)/\equ} \subseteq \CC'_{\KKK^{2}(X)}/\equ$. The converse inclusion follows similarly.
\end{proof}
Assume that $(\KKK^{2}(X)/\equ, \KKK^{*}(\KKK_{D}))$ is weakly hyperbolic. Then it follows that from each $\KKK^{*}(\KKK_{D})$-tree representing an equivalence class in $\KKK^{2}(X)/\equ$ one obtains a $\KKK(\KKK_{D})$-tree representing the same class by stripping off the `$*$' decoration. This also shows that  $\KKK(\KKK_{D})$-trees witnessing that $\KK \in \CC'_{\KKK^{2}(X)}$ do not only represent $\KK$ but all sets in the $\equ$-equivalence class of $\KK$. Such multi-representations are also common in Weihrauch's Type-Two Theory of Effectivity.

\begin{proposition}\label{pn-constmi}
Let $(X, D)$ be a compact covering IFS. Then $\set{\bigcup \KK}{\KK \in \CC'_{\KKK^{2}(X)}} \subseteq \CC_{\KKK(X)}$.
\end{proposition}
\begin{proof}
Let $\KK \in \CC'_{\KKK^{2}(X)}$. Then there are (constructively) $d_{1}, \ldots, d_{s} \in D$ and $\MM_{1}, \ldots, \MM_{s} \in \CC'_{\KK^{2}(X)}$ so that 
\[
\bigcup \KK 
= \bigcup [\KKK(d_{1}), \ldots, \KKK(d_{s})](\MM_{1}, \ldots, \MM_{s}) 
=[d_{1}, \ldots, d_{s}](\bigcup \MM_{1}, \ldots, \bigcup \MM_{s}),
\]
where the last equality holds by (\ref{eq-mapstar}).
It follows that 
\[
\set{\bigcup \KK}{\KK \in \CC'_{\KKK^{2}(X)}} \subseteq  \Omega_{\KKK(X)}(\set{\bigcup \KK}{\KK \in \CC'_{\KKK^{2}(X)}}).
\] 
With co-induction we therefore obtain that  $\set{\bigcup \KK}{\KK \in \CC'_{\KKK^{2}(X)}} \subseteq \CC_{\KKK(X)}$.
\end{proof}

This is a constructive version of Michael's Theorem: If all IFS involved are weakly hyperbolic, given a $\KKK_{D}$-tree representing $\KK$ we can compute a $\KKK(D)$-tree representing $\bigcup \KK$. 

The next results even allows the extraction of an algorithm for the computation of union as an operation from $\KKK^{2}(X)$ to $\KKK(X)$.

\begin{theorem}\label{thm-unifun}
Let $(X, D)$ be a compact covering IFS. Then $\bigcup \in \CC_{\FF(\KKK^{2}(X), \KKK(X))}$.
\end{theorem}
\begin{proof}
We again use co-induction and show that $\{ \bigcup \} \subseteq \JJJ^{\KKK^{2}(X), \KKK(X)}(\{ \bigcup \})$. Let to this end $d_{1}, \ldots, d_{s} \in D$. As we have seen in the preceding proof,
\[
[d_{1}, \ldots, d_{s}] \circ (\bigcup \times \cdots \times \bigcup) 
= \bigcup \circ [\KKK(d_{1}), \ldots, \KKK(d_{s})].
\]

By Rule~(W) for $\JJJ^{\KKK^{2}(X), \KKK(X)}$ we have for all $d_{1}, \ldots, d_{s} \in D$ that $[d_{1}, \ldots, d_{s}] \circ (\bigcup \times \cdots \times \bigcup)  \in \JJJ^{\KKK^{2}(X), \KKK(X)}(\{ \bigcup \})$. Thus, $\bigcup \circ [\KKK(d_{1}), \ldots, \KKK(d_{s})] \in \JJJ^{\KKK^{2}(X), \KKK(X)}(\{ \bigcup \})$, for all $d_{1}, \ldots, d_{s} \in D$, from which it follows with Rule~(R) that $ \bigcup \in \JJJ^{\KKK^{2}(X), \KKK(X)}(\{ \bigcup \})$.
\end{proof}

\section{Conclusion}\label{sec-conc}

In this paper, a uniform framework for computing with infinite objects like real numbers, compact sets, tuples of such, and uniformly continuous maps is presented. It combines and extends the approaches developed in a series of papers by Berger and co-authors~\cite{ub,be,bh,bse,bs}.
In particular, it allows to deal with compact-valued maps and their selection functions. Maps of this kind abundantly occur in applied mathematics. They are studied in set-valued analysis~\cite{ac,af}, and have applications in areas such as optimal control and mathematical economics, to mention a few. In addition, they are used to model non-determinism.

The framework is based on covering extended iterated function systems, where the underlying spaces are compact metric spaces and the contraction maps in the function systems are allowed to be multi-ary, or, more generally, weakly hyperbolic compact covering extended iterated function systems operating on Hausdorff spaces with, not necessarily unary, maps in the function system. Because of the covering condition co-inductive characterisations of the functions systems can be given. Results with computational content are then derived by constructively reasoning on the basis of these characterisations. Realisability facilitates the extraction of algorithms from the corresponding proofs. In so doing, points of the spaces are represented by finitely branching infinite trees. The computational power of the approach is that of Type-Two Theory of Effectivity.

\section*{Acknowledgement}\label{sec-ack}

The author is grateful to the anonymous referees for their careful reading of the manuscript and valuable comments helping to improve the paper.

\end{document}